\documentclass[11pt, a4paper]{article}

\pdfoutput=1
\RequirePackage[numbers]{natbib}

\usepackage[bookmarks=false]{hyperref}
\usepackage{amsthm,amsfonts,amssymb,amsmath,natbib}
\usepackage{mathrsfs}

\usepackage{multirow}

\usepackage{listings}
\usepackage[latin1]{inputenc}
\usepackage[T1]{fontenc}
\usepackage[english]{babel}

\usepackage{subcaption}

\usepackage{articlestyle}
\newcommand{\scrD}{\mathscr{D}}
\usepackage{algorithm}
\usepackage{authblk}

\usepackage{fullpage}

\usepackage{ifpdf}
\usepackage{units}

\usepackage{mathtools}
\mathtoolsset{showonlyrefs}

\usepackage{xr}

\usepackage[normalem]{ulem}

\title{Level set Cox processes}
\author[1]{Anders Hildeman}
\author[1]{David Bolin}
\author[2]{Jonas Wallin}
\author[3]{Janine B. Illian}
\date{}

\affil[1]{Department of Mathematical Sciences, Chalmers University of Technology and University of Gothenburg, Sweden}
\affil[2]{Department of Statistics, Lund University, Sweden}
\affil[3]{School of Mathematics and Statistics, University of St Andrews, Scotland}

\begin{document}

\maketitle

\setcounter{tocdepth}{2}

% !TEX root = shell.tex
\begin{abstract}
\label{sec:abstract}
The log-Gaussian Cox process (LGCP) is a popular point process for modeling non-interacting spatial point patterns.
This paper extends the LGCP model to handle data exhibiting fundamentally different behaviors in different subregions of the spatial domain.
The aim of the analyst might be either to identify and classify these regions, to perform kriging, or to derive some properties of the parameters driving the random field in one or several of the subregions.
The extension is based on replacing the latent Gaussian random field in the LGCP by a latent spatial mixture model. 
The mixture model is specified using a latent, categorically valued, random field induced by level set operations on a Gaussian random field. Conditional on the classification, the intensity surface for each class is modeled by a set of independent Gaussian random fields.
This allows for standard stationary covariance structures, such as the Mat\'{e}rn family, to be used to model Gaussian random fields with some degree of general smoothness but also occasional and structured sharp discontinuities.

A computationally efficient MCMC method is proposed for Bayesian inference and we show consistency of finite dimensional approximations of the model. Finally, the model is fitted to point pattern data derived from a tropical rainforest on Barro Colorado island, Panama. We show that the proposed model is able to capture behavior for which inference based on the standard LGCP is biased.
\end{abstract}

%\keywords{preconditioned Crank-Nicholson MALA, Mixture Model, Log-Gaussian Cox Process}

% !TEX root = shell.tex
\section{Introduction}\label{sec:intro}
Cox processes, and in particular log-Gaussian Cox processes (LGCP), have been used extensively as flexible models of spatial point pattern data \citep{lit:moller, lit:moller2, lit:illian, lit:diggle}. These are hierachical point process models where the point locations are assumed to be independent given a random intensity function 
\begin{equation}
\ints(\psp) = \exp\{\covars(\psp)\regCoef + \latf(\psp)\}, 
\label{intensity:intro}
\end{equation}
where $\covars(\psp)$ is a, possibly multivariate, function of covariates and $ \latf(\psp)$ is a Gaussian random field, which  is typically assumed to be stationary. The random field captures spatial structure in the point pattern that the given covariates cannot capture.  In this paper, we relax the assumption that a single stationary Gaussian field can account for those remaining spatial structures and develop a mixture model based on level set inversion.

To motivate the relevance of the approach we consider a point pattern data set formed by the locations of the tree species \textit{Beilschmiedia Pendula}, one of the species in the tropical rainforest plot on Barro Colorado Island \citep{condit:98, condital:00, burslemal:01, hubbellal:05}. The point pattern comprises $2461$ point locations   in a rectangular observation window (500 m $\times$ 1000 m), see Figure \ref{fig:treesIntro} (a).  This pattern has been analysed repeatedly in the literature and forms part of the example patterns in the \texttt{R} \citep{lit:r} package \textit{spatstat} \citep{lit:spatstat}.
Previous analyses have fitted a log-Gaussian Cox process \citep{lit:moller2}  to this and related data sets to draw conclusions on the association of habitat preferences based on a number of spatial covariates reflecting local soil chemistry and topography \citep{lit:moller2, lit:illian}. We initially fitted a log Gaussian Cox process to this pattern, with an intensity function as in Equation \eqref{intensity:intro}, using $11$ covariates, see Section \ref{sec:examples}.

On close inspection, the pattern in Figure \ref{fig:treesIntro} (a) shows large areas of very low point intensity where hardly any trees can be found. The estimated posterior mean using the LGCP model predicts large regions of low intensity, as plotted in Figure \ref{fig:lgcpIntro}. Anecdotal knowledge reveals that these regions are covered by a swamp, where the tree species is known to be very unlikely to grow, independent of local soil covariates and topography. However, data on the exact extent of the swamp is not available. When a LGCP model that ignores the presence of swamp is fitted to this pattern, the swamp is likely to act as a confounding factor and this is likely to impact on inference. Hence, any conclusions on habitat preferences of the species will be heavily biased. Covariates associated with the presence of the swamp may appear to have a significant correlation with the intensity of the tree growth, or important covariates might appear non-significant as they vary indepedently of the presence of the swamp.

\begin{figure}
\centering
\begin{subfigure}{0.32\textwidth}
\centering
\includegraphics[width = 1 \textwidth, height = 5 cm]{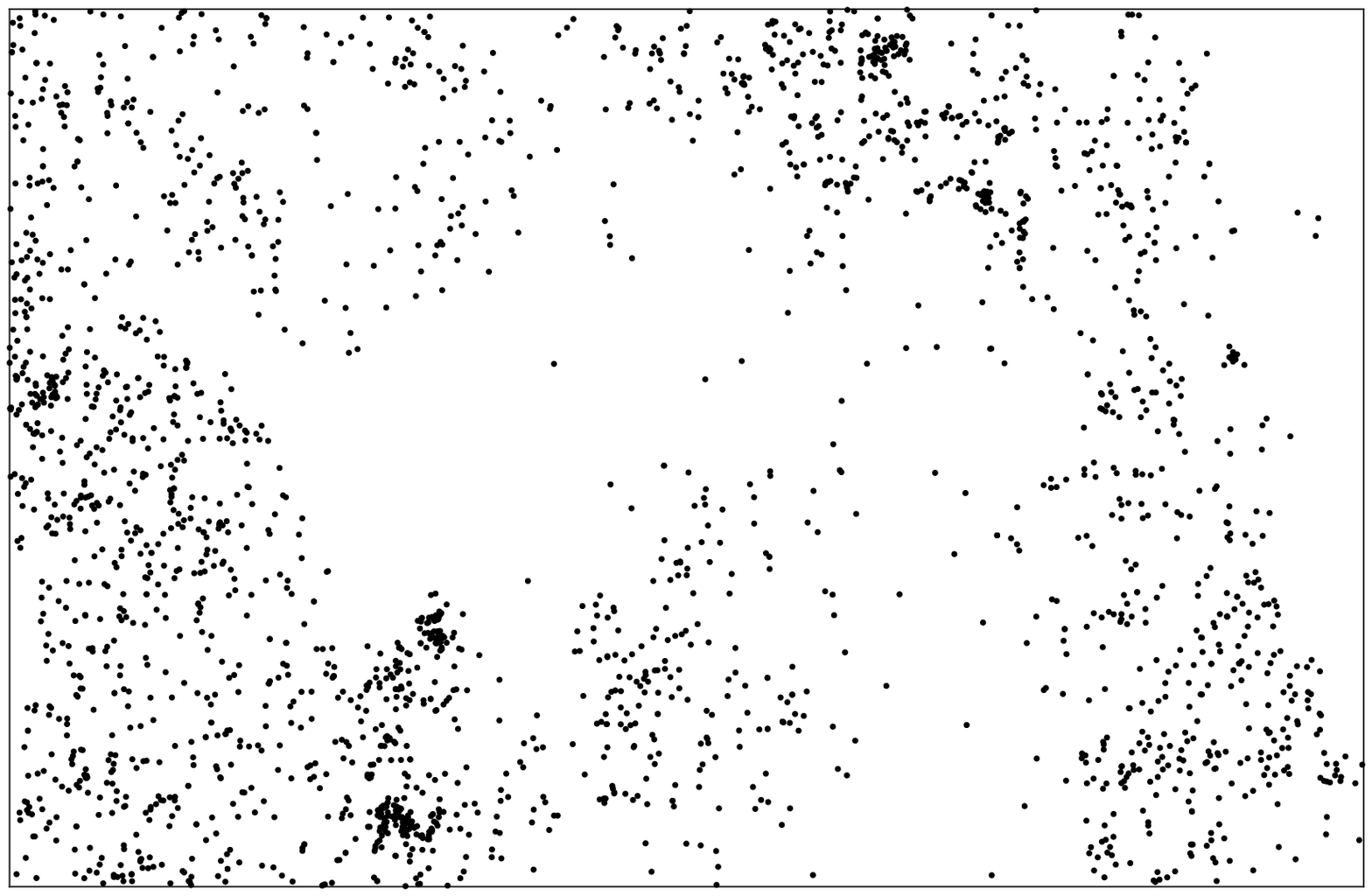}
\caption{}
\label{fig:treesIntro}
\end{subfigure}
\begin{subfigure}{0.32\textwidth}
\centering
\includegraphics[width = 1 \textwidth, height=5cm]{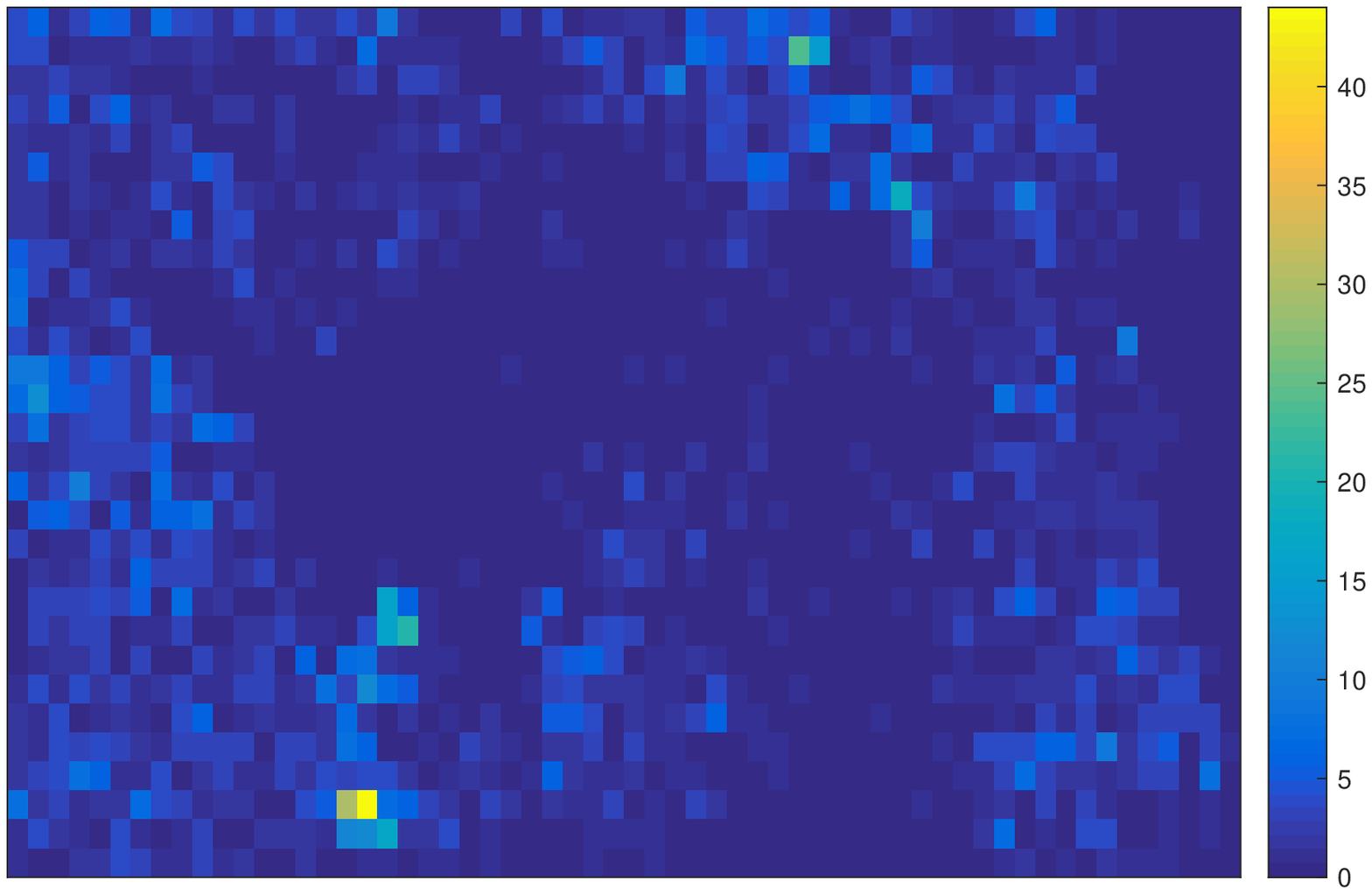}
\caption{ }
\label{fig:griddedTrees}
\end{subfigure}
\begin{subfigure}{0.32\textwidth}
\centering
\includegraphics[width = 1 \textwidth, height = 5 cm]{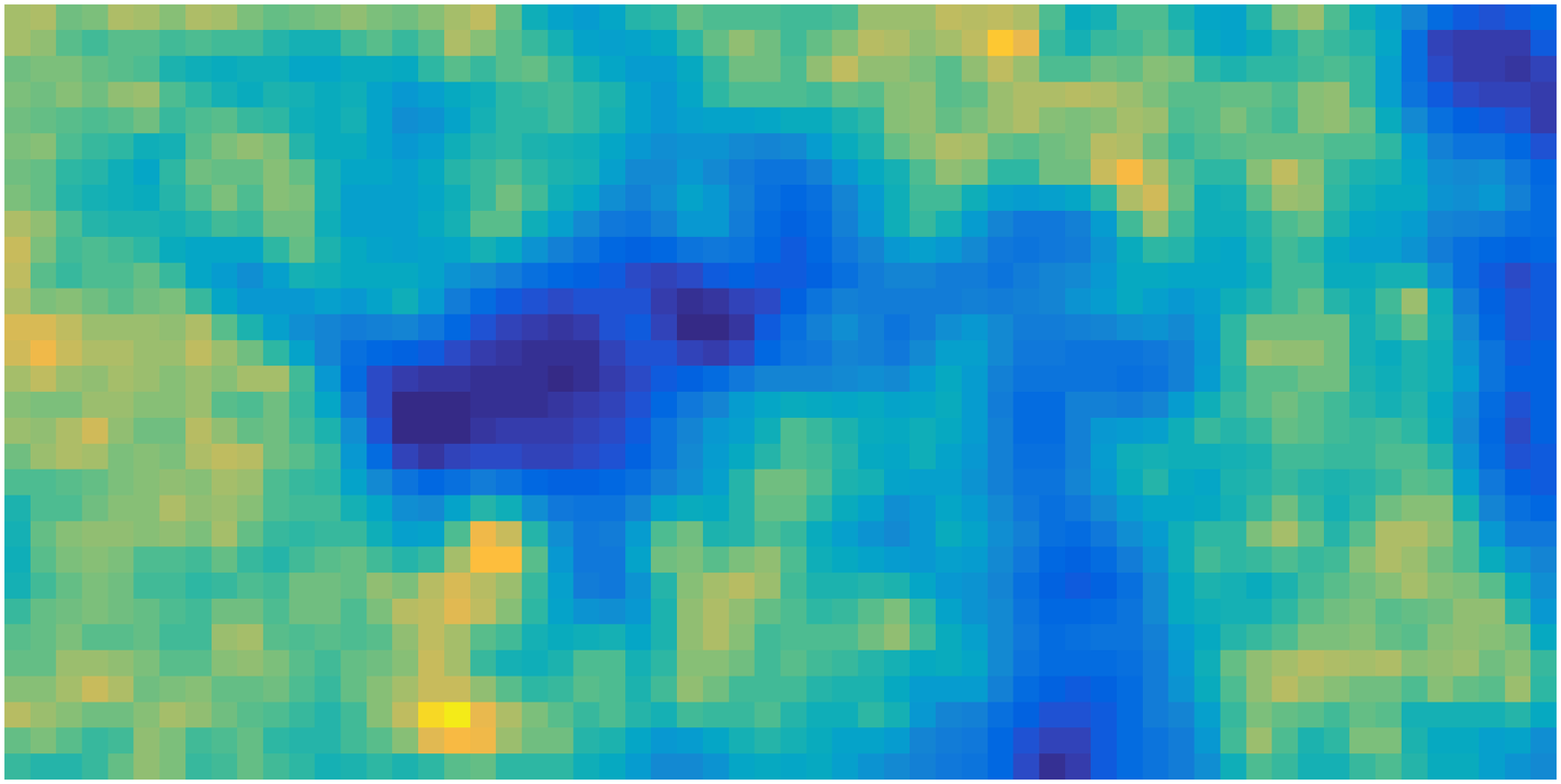}
\caption{}
\label{fig:lgcpIntro}
\end{subfigure}

\caption{Spatial point pattern formed by the locations of trees of the species \textit{Beilschmiedia pendula} in a 500 m $\times$ 1000 m rainforest plot on Barro Colorado Island (a), a gridded version of the data (b), and posterior mean of log intensity using a log-Gaussian Cox process model (c).}
\end{figure}

The approach we take here is designed to capture sharp discontinuities in the intensity surface that result from qualitative yet unavailable covariates or environmental conditions as the one seen in this example. These effects cannot be captured by the classical Gaussian random field approach. Further examples of data where such a model could be important is ecological data with several distinct types of habitat, spatial regions with different treatment regimes in medical data, or materials exhibiting separate regions of differing properties in material science. 
Specifically, we consider a Cox process model  where the intensity surface is modeled using a Bayesian level set approach. The proposed model is an extension of the log-Gaussian Cox process with increased flexibility resulting from a random segmentation of the spatial region into $K$ classes. The intensity surfaces of the regions associated with the $K$ different classes can be modeled separately of each other by latent log-Gaussian random fields with simple covariance structures, while still maintaining flexibility. We refer to the proposed model as the \textit{\ac{lscp}}.

Level set inversion \citep{lit:santosa, lit:burger} are geometric inverse problems where the main objective is to find interfaces between geometrical regions based on observed data. In this approach, the interfaces are modeled as level sets of an unknown \textit{level set function}. Level set inversion has been used extensively for segmentation \citep{lit:chung, lit:lorentzen, lit:scheuermann}, for multiphase flow modeling \citep{lit:barman, lit:desjardins}, and for statistical modeling of porous materials \citep{lit:mourzenko}. \citet{lit:higgs} modeled spatially correlated categorical data using a Bayesian level set approach, where the level set function was modeled as a Gaussian random field. This probabilistic approach, which \citet{lit:iglesias} and \citet{lit:dunlop} extended to more general inverse problems, has the advantage that the level sets can be estimated through the posterior distribution of the level set function given the observed data. 

The \ac{lscp} is, like the LGCP, a continuous process. In order to use the model in practical inference some finite dimensional approximations are required. 
We show that the classical lattice approximation of the \ac{lscp} converges, in total variation distance, to the continuous model as the grid gets finer. Further, we propose a computationally efficient Markov chain Monte-Carlo (MCMC) algorithm for Bayesian inference on the model parameters, based on preconditioned Crank-Nicholson Langevin proposals \citep{lit:cotter}.

This paper is structured as follows. A detailed model description is given in Section \ref{sec:model}. In Section \ref{sec:estimation}, we derive the MCMC algorithm for the method. Section \ref{sec:examples} analyses the \textit{Beilschmiedia Pendula} point pattern of rainforest trees with the new approach. Finally, Section \ref{sec:discussion} discusses the presented material and possible future extensions of it. The theoretical results and proofs are given in two appendices.

%\citet{lit:iglesias} and \citet{lit:dunlop} modeled more general level set inversion problems using the Bayesian framework and Gaussian prior probability measures, they also showed how the preconditioned Crank-Nicholson random walk MCMC method \citep{lit:cotter} can be used for efficient inference. 

% !TEX root = shell.tex
\section{The model and its properties}
\label{sec:model}
In this section, we first introduce the level set Cox process in Subsection \ref{sec:pointModel}. Some examples of the model are presented in Subsection \ref{sec:specialModels} and basic properties of the model are presented in Subsection \ref{sec:properties}. Finally, Subsection \ref{sec:findim} introduces finite dimensional approximations of the model necessary for infererence.

\subsection{Level set Cox process model}
\label{sec:pointModel}
Let $\domSp \subset \R^2$ be a bounded domain. The Bayesian level set inversion problem of \citet{lit:iglesias} corresponds to reconstructing a latent field of the form
\begin{equation}
\latf(\psp) = \sum_{k = 1}^K \latf_k \indicator{\psp \in \domSp_k},
\label{eq:levelSetFunction}
\end{equation}
given noisy data.
Here $\domSp_k\subset \domSp$ is the spatial region associated with segmentation class $k$, and $\latf_k$ are fixed values. 
If the constants $\{\latf_k\}_k$ are known, the partition $\{\domSp_k\}_{k=1}^K$ characterizes $\latf$. \citet{lit:iglesias} defined $\domSp_k$ as an excursion set of an unknown random continuous \textit{level set function}, $\latf_0$, such as $\domSp_k = \{\psp: \threshParam_{k-1} < \latf_0(\psp) \le \threshParam_k \}$. Here $\threshParam_k$ are constants such that $\{-\infty = \threshParam_0 < \threshParam_1 < ... < \threshParam_{K+1} = \infty \}$ and $\latf_0$ is assumed to be a realization of a Gaussian random field. Thus, this model corresponds to the level set problem for categorical data by  \citet{lit:higgs}. The level set model using a latent Gaussian random field is not identifiable with regards to the parameter triplet \textit{threshold values}, \textit{mean}, and \textit{marginal variance} of the level set field, $\latf_0$. Hence, we define $\latf_0$ to have standard normal marginal distributions in order to make the model identifiable. 

We extend the level set function of \eqref{eq:levelSetFunction} by replacing the fixed constants $\latf_k$ by Gaussian random fields and denote these Gaussian random fields as $\latf_k(\psp) + \mu_k(\psp)$, where $\mu_k$ is a deterministic mean function and $\latf_k$ is a centered Gaussian random field. 
\begin{equation}
\latf(\psp) = \sum_{k=1}^K \left( \latf_k(\psp) + \mu_k(\psp) \right) \indicator{c_{k-1} < X_0(\psp) + \mu_0(\psp) < c_{k}}.
\label{eq:loglambda}
\end{equation}
 This can be regarded as a mixture model of Gaussian fields related to the non-stationary
geostatistical model proposed by \citet{lit:fuentes}.  We use this model to specify a statistical model for spatial point process data through a Cox process \citep{lit:diggle}, modeling the number of occurrences of some event in a subregion $\subdomSp \subseteq \domSp$ as an inhomogeneous Poisson process conditional on a realization of $\latf$, i.e.
\begin{align}
\obsf(\subdomSp) \sim \pPOIS \left( \int_{\subdomSp} \ints(\psp) d\psp \right),
\label{eq:Xs}
\end{align}
where the intensity surface is $\ints(\psp) = \exp\{\latf(\psp) \}$.

 A common usage of point process models is to study the effect of covariates on observed point patterns. A simple way of doing this is through a standard Poisson regression, where the log-intensity of the point process is of the form 
 $\log \ints(\psp) = \covars(\psp) \regCoef$, where $\covars(\psp)$ are the covariates of interest. This can easily be incorporated into the \ac{lscp} model by letting $\mu_k(\psp) = \covars(\psp) \regCoef_k$ or $\mu_k(\psp)= \mu(\psp) = \covars(\psp) \regCoef$.

\subsection{Model examples}\label{sec:specialModels}
Poisson regression and log-Gaussian Cox processes are special cases of the  LSCP model. For an illustration of the flexibility of the model, Figure \ref{fig:fourExamps} shows the log intensity for four special cases simulated in the unit square. 
In this figure, all Gaussian random fields are assumed to have constant means $\mu$ and Mat\'{e}rn covariance functions \citep{lit:matern}, 
\begin{align}
\C( \latf_k(\psp_1) , \latf_k(\psp_2) ) = \C( h ) = \frac{ \std^2 }{2^{\smoothParam - 1}\Gamma(\smoothParam)}(\kappa h)^{\smoothParam}K_{\smoothParam}(\kappa h),
\label{eq:maternCov}
\end{align}
where $h = \norm{\psp_1 - \psp_2}$, $\std^2 = \Var( \latf_k(\psp) )$, $\kappa = \frac{\sqrt{8\smoothParam}}{\rangParam}$, and $\smoothParam$ is a smoothness parameter. Further, $\rangParam$ is the correlation range approximately corresponding to the value of $h$ where the correlation is $0.1$, $K_{\smoothParam}$ is a modified Bessel function of the second kind, and $\Gamma$ is the Gamma function.

The patterns were generated using the same random seed such that the level set function is the same for all cases, yielding comparable results. A realization of $\log \ints(\psp)$ using two classes can be seen in Panel (a). The log intensity surface of the first class has $\mu = 2$ and $r = 0.1$, whereas the second class has $\mu = 0$  and $r = 0.2$. Both fields have $\sigma = \nu = 1$.
The level set field, $X_0$, has a threshold value at the origin, $\threshParam_1 = 0$, and range $\rangParam = 0.4$. 
In the figure, the regions belonging to the two classes, and the difference in spatial correlation range is clearly visible. 

A simplification of the model is obtained by assuming that the intensity for one of the two classes is constant (change $X_1(\psp)$ to a constant $X_1$, for instance). A realization of such a log intensity surface can be seen in Panel (b). This  model might be relevant in applications where some unknown factor makes it unlikely to observe points in certain subregions and may be regarded as spatially varying zero-inflation \citep{lit:lambert}. 
If a standard LGCP is fitted to data of this type some overdispersion unexplained and the estimated mean field and covariance parameters will be biased; this is not the case for the LSCP model.  We discuss and example of this in Section \ref{sec:examples}. 
The two-class model can of course be simplified further by assuming a constant intensity for both classes, and $\log \ints(\psp)$ is then of the form \eqref{eq:levelSetFunction}. A realization of this simplified model is shown in Panel (c).

The last model example uses the structure of the level set formulation to capture effects on the boundary between two regions. 
For a model with three classes, the second class takes on the role of an interface layer between the first and third class as can be seen in Panel (d).
The log intensity is in this case $\log \ints(\psp) = \mixProb_1(\psp) X_1 + \mixProb_2(\psp) \latf_2(\psp) + \mixProb_3(\psp) X_3$.
This can be used to model effects present on the boundary between two regions. Examples of potential applications are activity on shore lines between water and land or mixing regions between fluids. 

\begin{figure}
\centering	
\begin{subfigure}{0.24\textwidth}
\centering
\includegraphics[keepaspectratio, width = \textwidth]{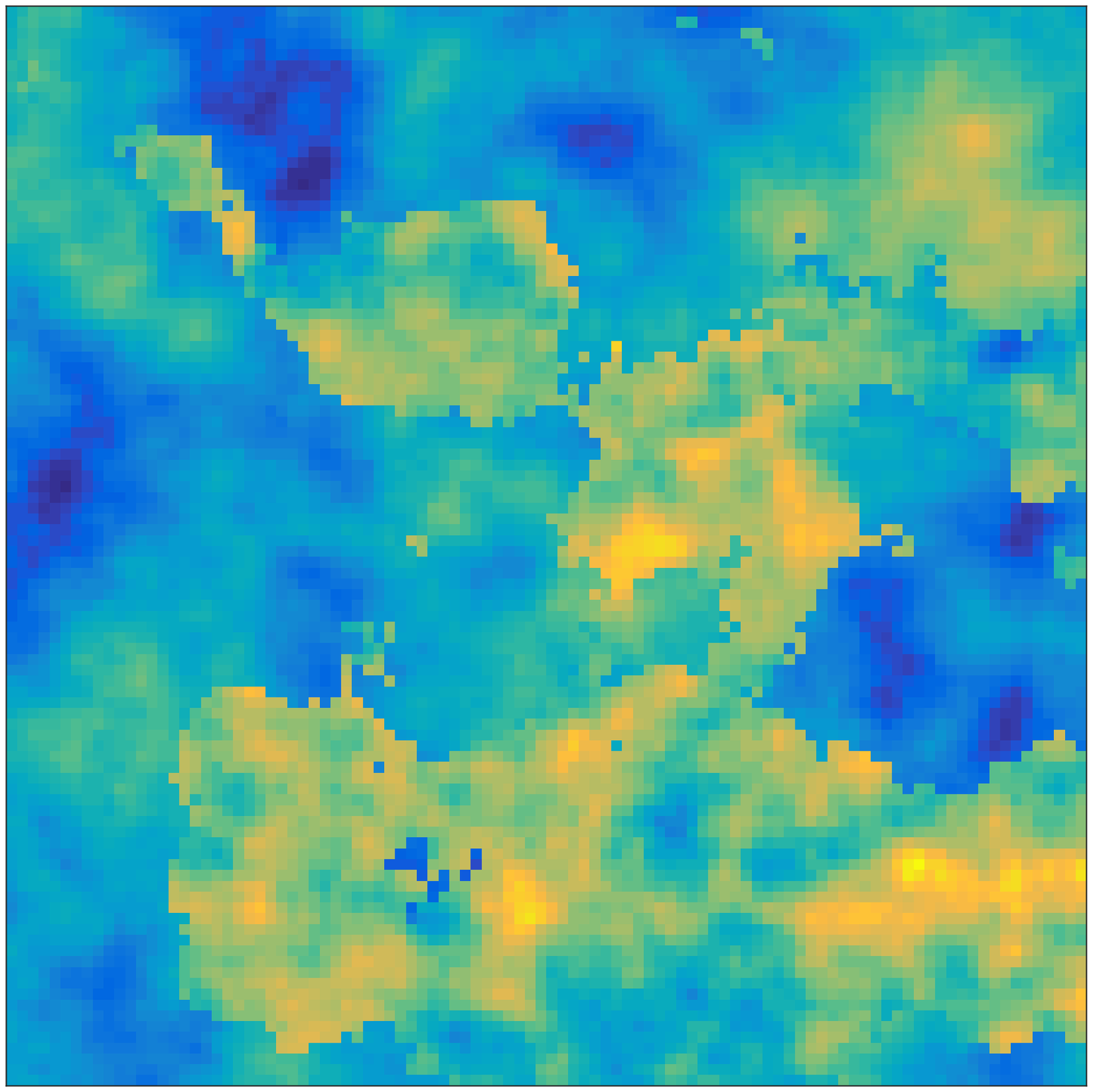}\\
\caption{ }
\label{fig:genExTwoLatf}
\end{subfigure}
\begin{subfigure}{0.24\textwidth}
\centering
\includegraphics[keepaspectratio, width = \textwidth]{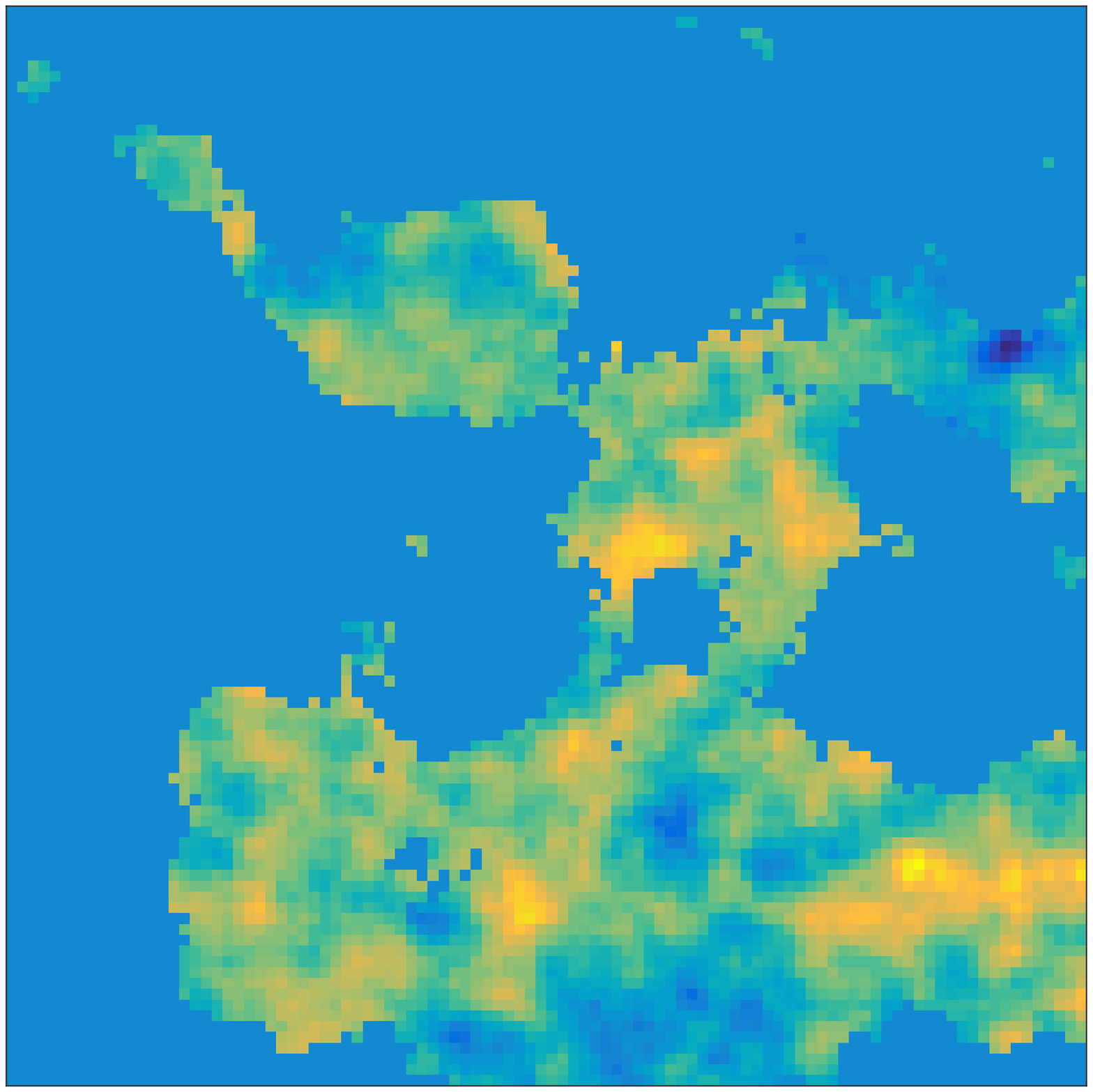} \\
\caption{}
\label{fig:genExOneLatf}
\end{subfigure} 
\begin{subfigure}{0.24\textwidth}
\centering
\includegraphics[keepaspectratio, width = \textwidth]{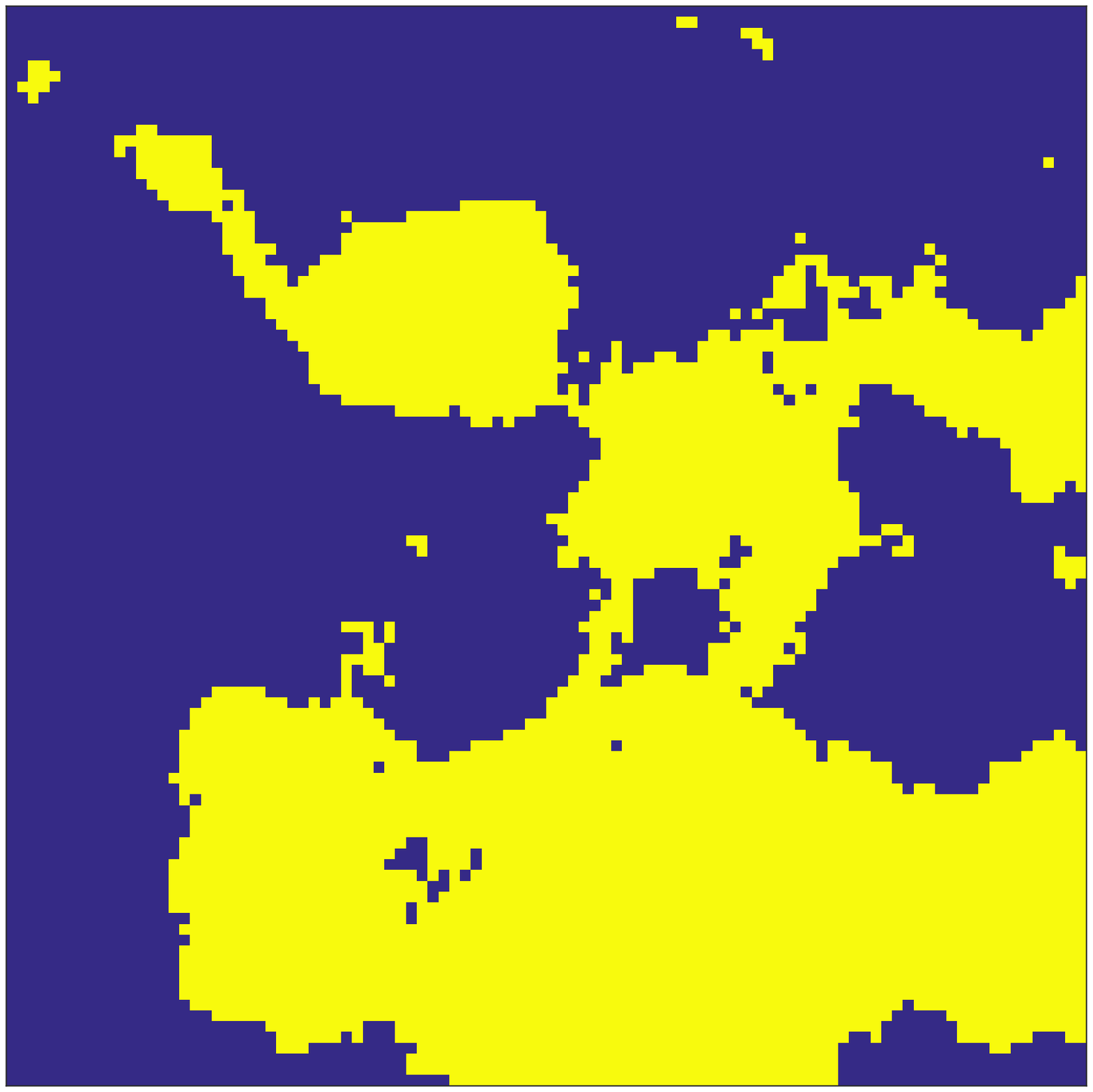} \\
\caption{ }
\label{fig:genExNoLatf}
\end{subfigure}
\begin{subfigure}{0.24\textwidth}
\centering
\includegraphics[keepaspectratio, width = \textwidth]{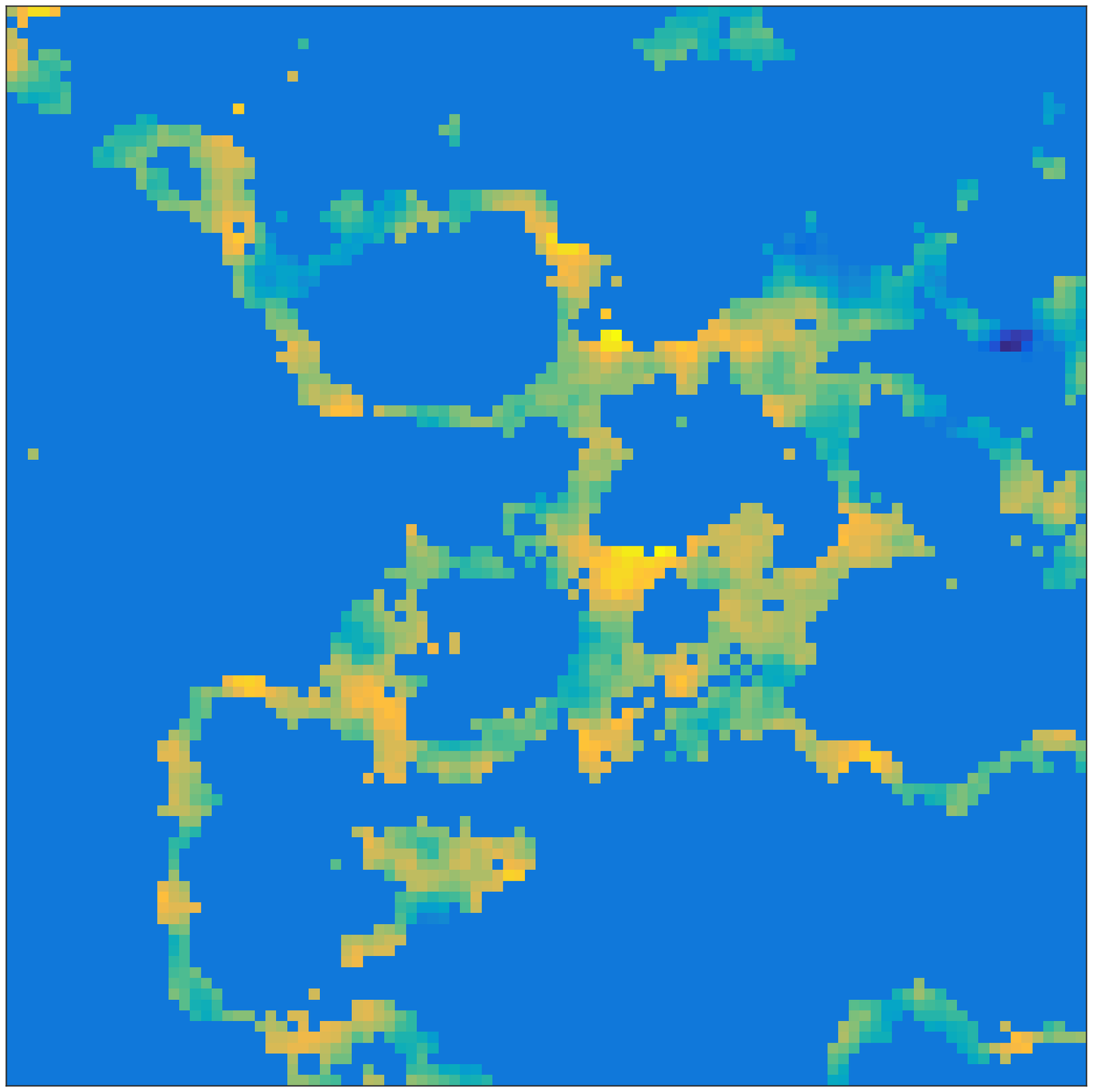} \\
\caption{}
\label{fig:genExSnake}
\end{subfigure} \\

\caption{Realization of the log intensity surface, $\log \ints(\psp)$, for the four models presented in Section \ref{sec:specialModels}. 
Panel a) corresponds to the model with two random classes, Panel b) with one constant and one random, panel c) with two constant, and Panel d) is the model with two constant and a third random boundary class.} 
\label{fig:fourExamps}
\end{figure}

\subsection{Model properties}\label{sec:properties}
The intensity measure $\Lambda = \{ \Lambda(\subdomSp) = \int_{\subdomSp} \ints(\psp) d\psp ; \subdomSp \subseteq \domSp  \}$ for a Cox process is well-defined if $\ints$ is almost surely finite and integrable. The \ac{lscp} model with $K=1$ reduces to the standard LGCP model, which has a well-defined random intensity measure if realizations of the Gaussian field are identified with its continuous modification \citep{lit:moller}. For $K>1$, a continuous modification does not need to exist but almost sure integrability follows if $\latf_0$ is a.s.~continuous which ensures that the sets $\{ \psp : c_{k-1} < \latf_0(\psp) \le \threshParam_k \}$ are a.s.\ Lebesgue measurable for all $k \in \{1, ..., K\}$. Hence the \ac{lscp} model is well-defined when the realizations of all Gaussian fields are identified with their continuous modification with respect to the Lebesgue measure. By the same argument as in Theorem 3 of \citet{lit:moller}, ergodicity of the \ac{lscp} model follows from ergodicity of $\log \ints$. Thus, the \ac{lscp} model is ergodic if all latent Gaussian fields are ergodic.

The following proposition gives semi-explicit formulas for the two first product densities. 
\begin{prop}\label{prop:intensmoments}
	For a level set Cox processes with log intensity \eqref{eq:loglambda},
	where $\{X_k\}_{k=1}^K$ are zero-mean stationary random fields with covariance functions $r_k$, the first moment of the intensity function equals
	\begin{align}
	\rho_1(\psp) = \expect{ \ints(\psp) } &= %\sum_{k=1}^K \exp\left(\mu_k(\psp) + \frac{r_k(0)}{2}\right) \prob{ \latf_0(\psp) \in C_k} \\
	\sum_{k=1}^K \exp\left(\mu_k(\psp) + \frac{r_k(0)}{2}\right) \left( \Phi\left(c_k-\mu_0(\psp) \right)  - \Phi\left(c_{k-1}-\mu_0(\psp) \right) \right),
	\label{eq:intensity}
	\end{align}
	where $\Phi$ is the CDF of a standard normal distribution.
	Further, the second moment of $\lambda$, $\rho_2(\psp_1,\psp_2) $, corresponding to the second order product density equals
	\begin{align}
	\rho_2(\psp_1,\psp_2) &=  \sum_{k=1}^K \exp\left(\mu_k(\psp_1) + \mu_k(\psp_2) + r_k(0) + r_k(|\psp_1-\psp_2|) \right) p_{kk} \\
	&+ \sum_{k=1}^K\sum_{l\neq k} p_{lk} \exp\left(\mu_l(\psp_1) + \mu_k(\psp_2) + \frac{r_k(0) + r_l(0)}{2}\right).
	\end{align}
	Here 
	\begin{align*}
		p_{lk} =  
%\expect{ \left(\Phi\left(\frac{c_{l}-\mu^*(\latf_0(\psp_1))}{\sigma^*(\latf_0(\psp_1))} \right) - \Phi\left(\frac{c_{l-1}-\mu^*(\latf_0(\psp_1))}{\sigma^*(\latf_0(\psp_1))} \right) \right) \indicator{\latf_0(\psp_1) \in C_k} } 
\int_{\threshParam_{k-1}}^{\threshParam_{k}} 		\left(\Phi\left(\frac{c_{l}-\mu^*(u)}{\sigma^*(u)} \right) - \Phi\left(\frac{c_{l-1}-\mu^*(u)}{\sigma^*(u)} \right) \right) \frac{e^{-\frac{(u-\mu_0(\psp_1))^2}{2}}  }{\sqrt{2\pi}} du,  \\		
	\end{align*}
	where $\mu^*(u) = \mu_0(\psp_2) + \frac{r_0(|u-\psp_2|)}{r_0(0)}(u - \mu_0(u))$ and $\sigma^*(u) = \sqrt{r_0(0)- \frac{r^2_0(|u-\psp_2|)}{r_0(0)}}$.
\end{prop}
The proof is given in \ref{sec:proofs}. The form of the pair-correlation function for the \ac{lscp} model is given by 
$
g(\psp_1, \psp_2) = \frac{\rho_2(\psp_1, \psp_2)}{\rho_1(\psp_1) \rho_1(\psp_2)}
$, and can hence be expressed using the first and second product densities given in Proposition \ref{prop:intensmoments}. 
The integral in  $p_{lk}$ has to be evaluated numerically.
If $g$ is translation invariant we can compute the inhomogeneous K-function \citep{lit:baddeley} of the process as $K(r) = \int_{B(0,r)} g_0(h)dh$, where $B(0,r)$ is a ball with radius $r$ centered at the origin and $g_0(\psp) = g(0, \psp)$. In the case of a homogeneous intensity, the K-function shows the expected number of other points at a radius of $r$ from a specific point.

Finally, the inhomogeneous empty space function \citep{lit:baddeley,lit:daley2}, $F(r)$,  for the general model is given by Proposition \ref{prop:emptyspace}.
\begin{prop}\label{prop:emptyspace}
For a level set Cox processes with log intensity \eqref{eq:loglambda}
	where $X_k$ are zero-mean stationary random fields with covariance functions $r_k$, the inhomogeneous empty space functions is given by
	\begin{align}
	F(\psp_0, r) = 1 -  \expect{ \prod_{k=1}^K \exp\left( -  \int_{\domSp_k \cap B(\psp_0,r)}  e^{ \mu_k(\psp)} e^{ \latf_k(\psp)} d\psp \right) },
	\end{align}
	where for a given realization of $\latf_0$, $\domSp_k$ is the region classified as $k$.
\end{prop}
The proof is given in \ref{sec:proofs}.

\subsection{Finite dimensional approximation}\label{sec:findim}
As for standard LGCP models, some finite dimensional approximation of the LSCP model is needed if it is to be used for inference.
The discretization we will use is a classical lattice approximation. The observational domain is discretized into subregions $\domSp_i$ of a regular lattice over the domain, and the point locations are replaced by counts $\obsf_i$ of the number of observations within each subregion $\domSp_i$. This yields the discretized model $\obsf_i \sim \pPOIS(\ints_i)$, where $\ints_i = \int_{\domSp_j} \ints(\psp)d \! \psp$ and the information on the fine-scale behaviour of the point pattern behavior is lost. The stochastic integral in the definition of $\lambda_i$ is not Gaussian and generally difficult to handle. Therefore, a common approximation is to use $\ints_j \approx  |\domSp_j| \ints(\psp_j)$, for some location (usually the center) $\psp_j \in \domSp_j$ \citep{lit:moller}. In Appendix \ref{sec:theoretical}, we show consistency of this finite dimensional approximation of the likelihood for the \ac{lscp} model. More precisely, we show that the posterior distribution for the latent fields $\{X_k\}_k$ computed using the lattice approximation converges, in total variation distance, to the posterior distribution of the continuous process. 

For any fixed lattice approximation there is a positive probability that the level set field takes values in several of the intervals $\{c_k, c_{k+1}\}$ in any fixed lattice cell. 
Since the spatial information about the level set field on a finer scale than the lattice discretization are lost we propose adding a ``nugget'' effect, $\whiteNoise_j$, for each lattice cell $\domSp_j$. The ``nugget'' effect will model the within-cell classification uncertainty.
This gives the discretization $\lambda_j \approx |\domSp_j|\tilde{\ints}(\psp_j)$, where 
\begin{equation}
\log\tilde{\ints}(\psp_j) = \sum_{k=1}^K \indicator{c_{k-1} < X_0(\psp_j) + \mu_0(\psp_j) + \whiteNoise_j < c_{k}}\latf_k(\psp_j),
\label{eq:fuzzylambda}
\end{equation}
and $\whiteNoise_j \sim \pN(0,\fuzzy^2)$. The \textit{nugget variance}, $\fuzzy^2$, controls the amount of mixing between the classes for a given realization of $\latf_0$. This classification mechanism is equivalent to the ordered probit model discussed in \citet{lit:dunlop}. 
%In Appendix \ref{sec:theoretical}, we show that the discretized model with nugget effect also gives consistent approximations of the posterior distribution given that $\std_{\whiteNoise} = C/\sqrt{N}$ for some constant $C$ \todo{I will check up on this, possibly that it was $C/N$}. 
In practice, it is typically difficult to objectively discern an appropriate value for $\std_{\whiteNoise}$ and hence we therefore let $\std_{\whiteNoise}$ be a regular parameter to be estimated for a fixed discretization.

% !TEX root = shell.tex
\section{Inference}
\label{sec:estimation}

It is common to fit LGCP models in a Bayesian setting. A popular approach is through Markov chain Monte Carlo (MCMC) methodology, for instance using the Metropolis adjusted Langevin algorithm (MALA) \citep{lit:roberts} which was suggested by \citet{lit:moller}. Another approach is through integrated nested Laplace approximation (INLA) \citep{lit:illian, lit:inla, lit:simpson}, which when applicable can have beneficial computational properties. 
In this work we use a Bayesian MCMC approach for estimating the model parameters of the \ac{lscp} model. Specifically, we propose a method based on the preconditioned Crank-Nicholson (pCN) MALA MCMC method of \citet{lit:cotter}. An important property of the pCN MALA is the optimal step length invariance to mesh refinement, which regular MALA does not have. However, for the \ac{lscp} model the main advantage is that it can be combined with efficient simulation methods based on the fast Fourier transform \citep{lit:lang} to decrease the computational cost; we provide more details on this below.  

Denote the parameters associated with class $k$ as $\theta_k$. For the level set field, $\latf_0$, we also include the nugget variance, $\fuzzy$, and the thresholds, $\{\threshParam_k\}_k$ in $\theta_0$. By introducing an auxiliary field $\clf$ defined such that $\prob{\clf(\psp_j) = k}  = \normcdf\left(\frac{\threshParam_{k} - \latf_0(\psp_j)}{\fuzzy}\right) - \normcdf \left(\frac{\threshParam_{k-1} - \latf_0(\psp_j)}{\fuzzy}\right)$, we have
\begin{align}
\log \tilde{\lambda}(\psp) \overset{d}{=} \sum_{k=1}^K \clf(\psp) \latf_k(\psp).
\end{align}
This means that parameters and latent fields of different classes, $\{\latf_k, \theta_k\}$, are conditionally independent given $\clf$. We use this to construct a Metropolis-within-Gibbs algorithm \citep{lit:casella} to sample from the joint posterior. In the $i$th iteration of the algorithm, the following three steps are performed
\begin{enumerate}
\item Sample from $\clf | \{\latf_k, \theta_k\}_k, \obsf$. The sampling can be performed exactly since $\clf(\psp_i) \perp \clf(\psp_i), \forall i \neq j$ given $\{\latf_k, \theta_k\}_k, \obsf$ and $\prob{\clf(\psp_i) = k}$ is known up to a normalizing constant. 
	
\item Sample from $\theta_k | \clf, \latf_k$ using the MALA random walk sampler. Since parameters from different classes are conditionally independent, the sampling can be performed separately, and in parallell, for each $\theta_k$. 
	
\item Sample from $\latf_k | \clf, \theta_k, \obsf$ using the pCN MALA algorithm of \citet{lit:cotter}. Also in this step, the updates for different $k$ can be done in parallel since the different Gaussian fields are conditionally independent.

\end{enumerate}

The computational bottleneck of the algorithm is the third step, where the latent Gaussian fields are sampled. If the model is discretized into a lattice with $N$ grid cells, the sampling of the Gaussian fields in the third step of the estimation method generally requires $\O(KN^3)$ operations. An approach to remedy this would be to acquire a Gaussian Markov random field approximation of the problem. This idea has been studied by \citep{lit:lindgren, lit:rue, lit:simpson} revealing computationally attractive properties on arbitrary domains. An adaptation of the method by \citet{lit:simpson} to the LSCP model would reduce the computational cost to $O(KN^{3/2})$. We can reduce this cost further by using the fact that proposals in the pCN MALA algorithm are drawn from the prior distribution of the fields. 

If we restrict ourselves to square domains and assume that the fields have stationary and isotropic covariance functions with known spectral density, we can represent $\{ \latf_k  \}_{k=0}^K$ using Fourier series expansions. By truncating these series, the fast Fourier transform can be used to sample the field on a regular lattice over the region. This means that the proposals can be generated with a $O(KN\log(N))$ computational complexity. Working in the spectral domain also allows for efficient computation of all gradients and acceptance probabilities needed, making the spectral approach and the pCN-MALA method in combination very favorable. In Appendix \ref{sec:theoretical} we justify this truncation theoretically by showing that convergence of the lattice approximation still holds given certain bounds on the spectral densities.

% !TEX root = shell.tex
\section{Application}\label{sec:examples}

To further illustrate our approach we return to the tropical rainforest data example in Section \ref{sec:intro} to compare the effect of considering \ac{lscp} models to a simple Poisson regression model as well as to the LGCP model. 

\subsection{Data}
The dataset consists of $2461$ locations of trees of the species  \textit{Beilschmiedia pendula} in a 50 ha rectangular study plot ($500$ x $1000$ meter) on the island of Barro Colorado in Panama, Figure \ref{fig:treesIntro}. The data were acquired from the first census of a major ongoing ecological study that started in the 1980s, designed to understand the mechanisms maintaining species richness, consisting of the observed positions of a large number of tree species (\citep{hubbellal:05,lit:hubell,lit:condit}). The study deliberately considers a spatially mapped rainforest community, arguing that population and community dynamics occur in a spatial context \citep{hubbell:01}.  In addition to the spatial pattern formed by the tree locations, measurements of topographical variables and soil nutrients that potentially influence the spatial distribution of the trees are available \citep{johnal:07, schreeg:10, lit:bcisoil},  with the aim of linking spatial patterns to spatial environmental variations, reflected by observed topography and soil nutrients. In the statistical literature some of the point patterns derived from the study have been considered, for example in \citep{lit:moller2, lit:simpson, lit:illian, lit:rajala} and the \textit{Beilschmiedia pendula} data are available in the \texttt{spatstat} package \citep{lit:spatstat} for the R project \citep{lit:r}.  

Elevation was measured and sampled on a 5x5 meter grid, and based on this an approximation of the slope at each of these grid points was calculated using a Sobel filter \citep{lit:Sonka}. Soil samples were taken at 300 locations, for which the amount of 12 soil constituents (Al, B, Ca, Cu, Fe, K, Mg, Mn, N, Nmin, P, Zn) as well as the pH level were measured; these were interpolated to yield spatially continuous covariates. Since the covariates derived from the soil samples and elevation were not sampled with the grid resolution they had to be interpolated to a common latice. In this example, the model was discretized to $30\times 60$ subregions over the observational window, giving a spatial resolution of $16.7 \times 16.7$ meters. 
The number of observed points in each subregion is shown as a two dimensional histogram in Figure \ref{fig:griddedTrees}.
The spatial interpolation of the covariates to this lattice grid was performed using bi-cubic splines with the function \textit{interp2} in Matlab (R2016a); Figure  \ref{fig:allCovs} shows the standardized covariates.

To avoid problems with multicollinearity among the covariates we chose to discard the covariates corresponding to high variance inflation factors (VIF) \citep{lit:neter}. The covariates were discarded iteratively by first computing the VIF for all the covariates, removing the covariate corresponding to the highest VIF value if it exceeds 5 and then starting over on the new reduced set of covariates. The algorithm was stopped when none of the VIFs exceeded 5. By this procedure, the covariates B, Ca, K, and Zn were discarded, leaving 11 covariates for further analysis.

\begin{comment}
\begin{figure}
\centering
\begin{subfigure}{0.49\textwidth}
\centering
\includegraphics[width = 1 \textwidth, keepaspectratio]{./figs/trees.eps}
\caption{Observed locations}
\label{fig:trees}
\end{subfigure}
\begin{subfigure}{0.49\textwidth}
\centering
\includegraphics[width = 1 \textwidth, keepaspectratio]{./figs/griddedTrees.eps}
\caption{Number of locations in each subdomain.}
\label{fig:griddedTrees}
\end{subfigure}
\caption{Observed points pattern of the tree species Beilschmiedia pendula. }
\end{figure}
\end{comment}

\begin{figure}[t]
%\centering
\begin{subfigure}{0.24\textwidth}
\centering
Elevation
\includegraphics[width = 1 \textwidth, keepaspectratio]{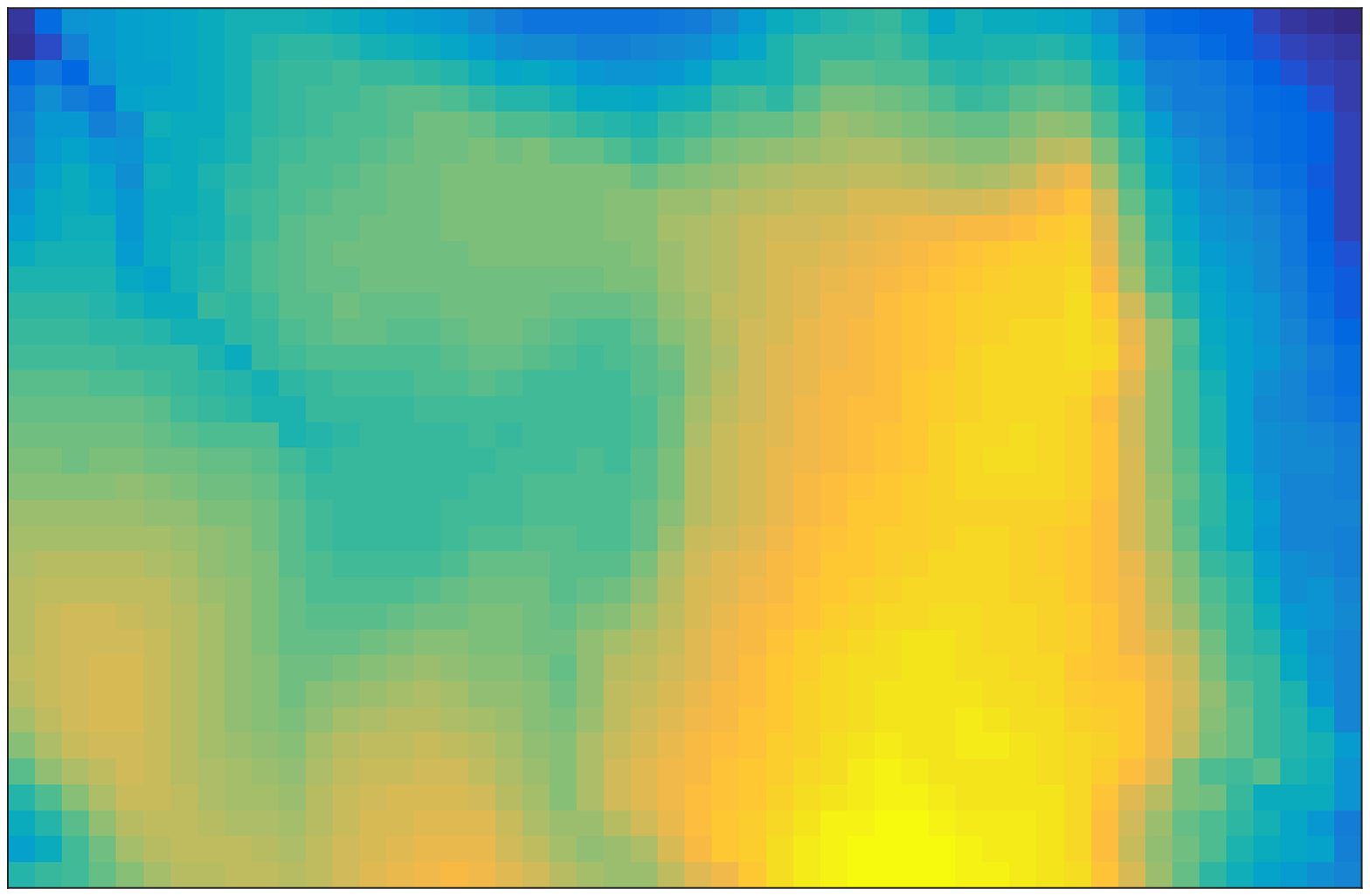}
\end{subfigure}
\begin{subfigure}{0.24\textwidth}
\centering
Slope
\includegraphics[width = 1 \textwidth, keepaspectratio]{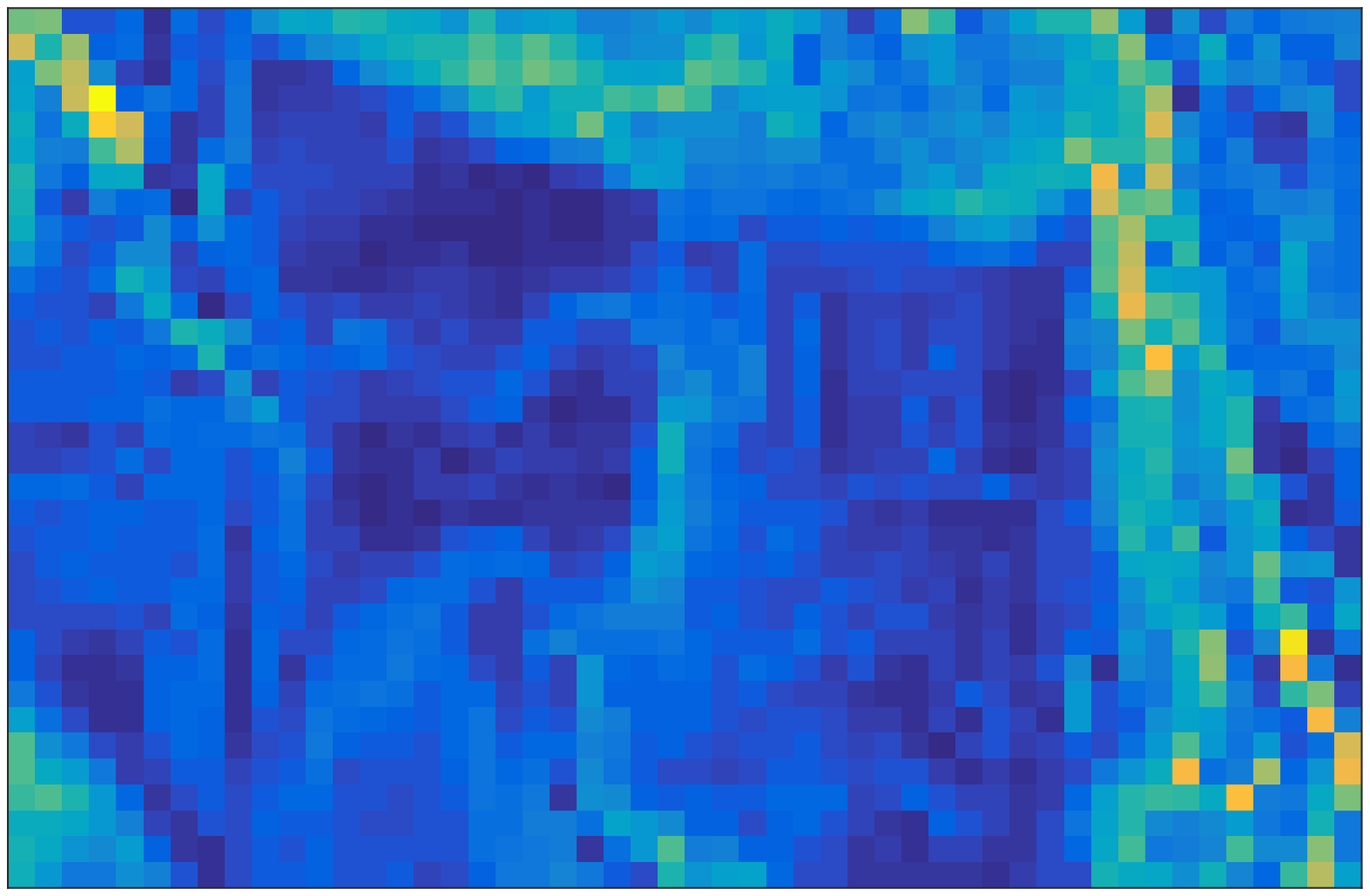}
\end{subfigure}
\begin{subfigure}{0.24\textwidth}
\centering
Al
\includegraphics[width = 1 \textwidth, keepaspectratio]{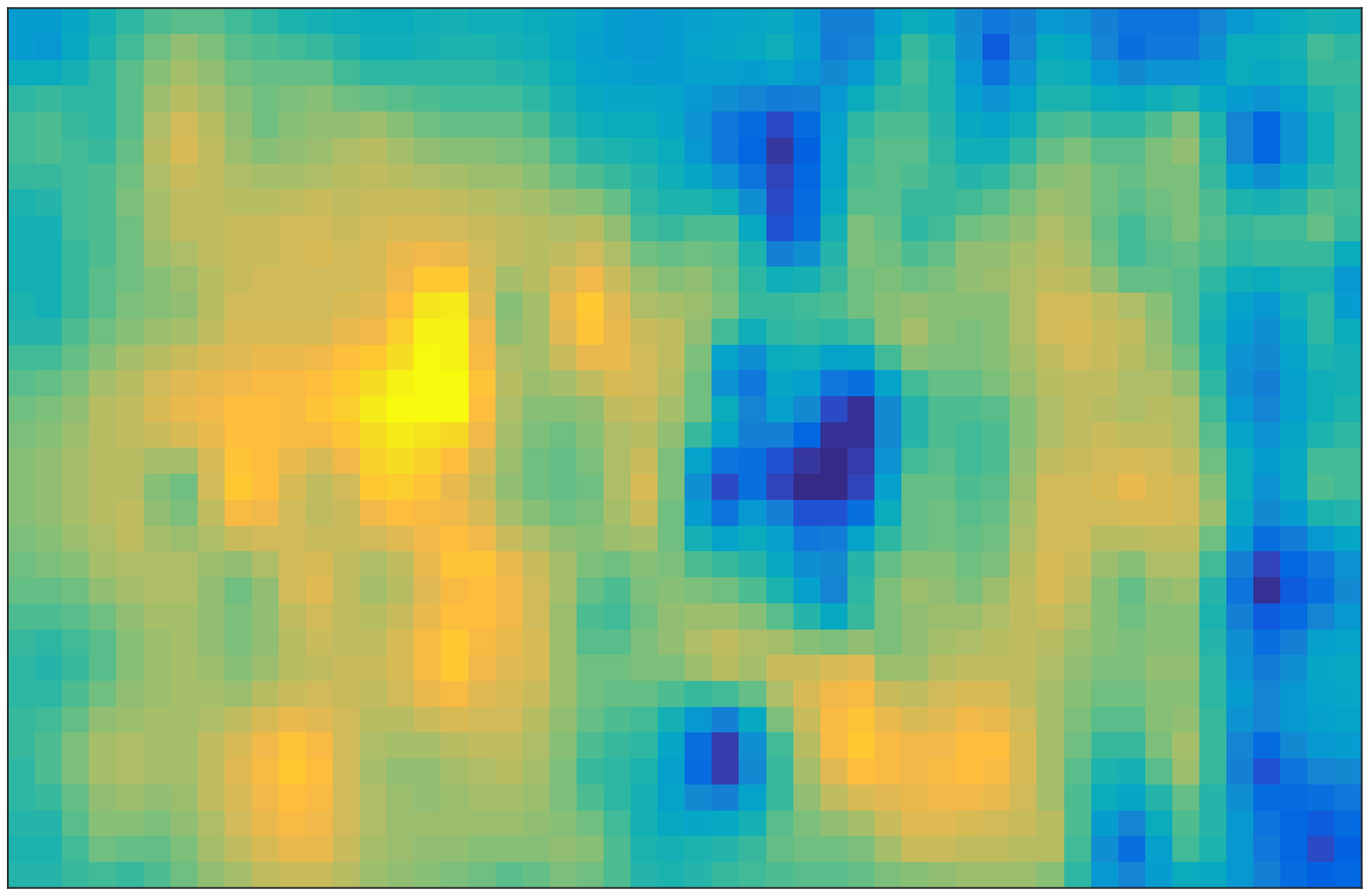}
\end{subfigure}
\begin{subfigure}{0.24\textwidth}
\centering
B
\includegraphics[width = 1 \textwidth, keepaspectratio]{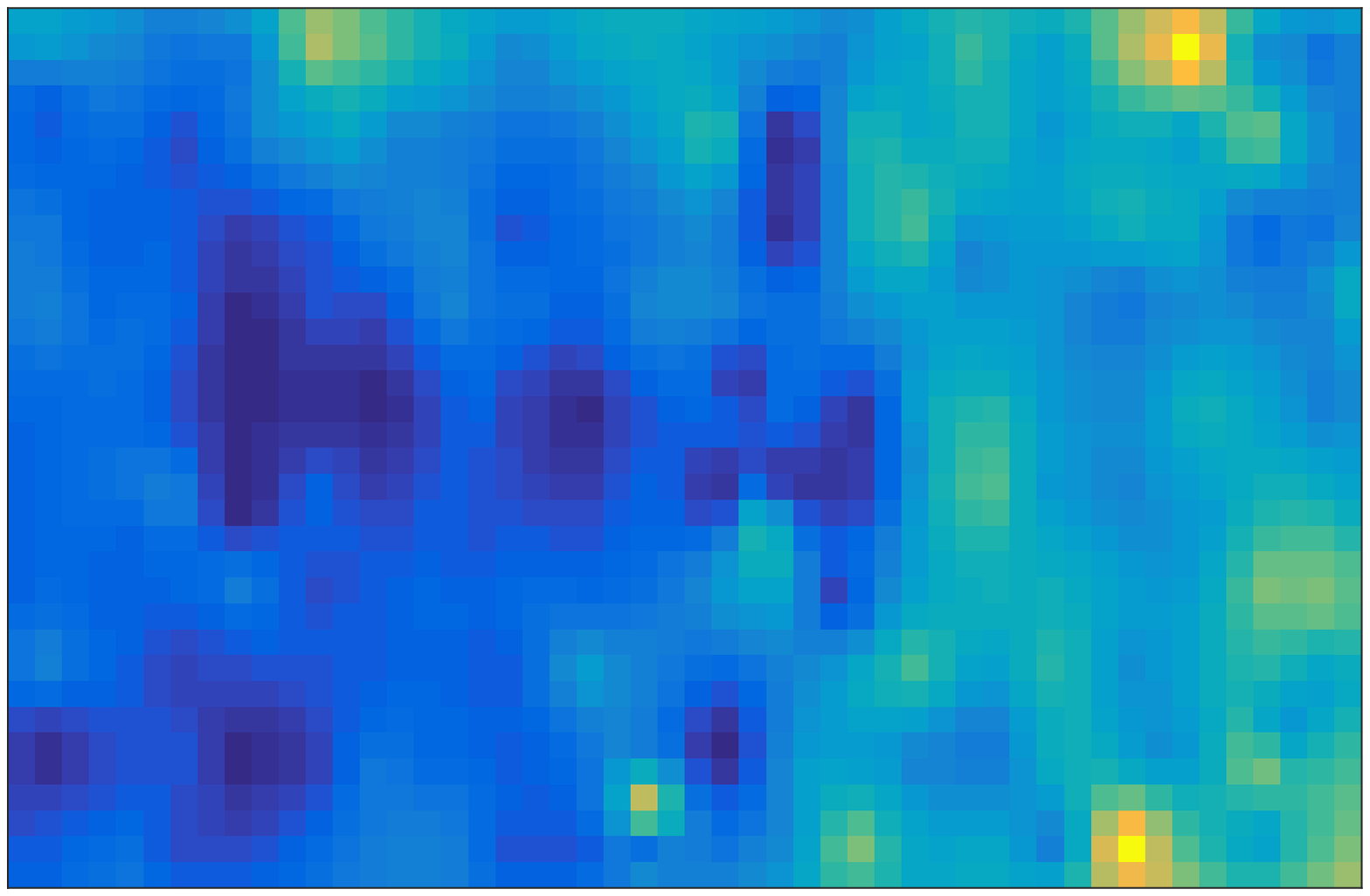}
\end{subfigure} \\

\begin{subfigure}{0.24\textwidth}
\centering
Ca
\includegraphics[width = 1 \textwidth, keepaspectratio]{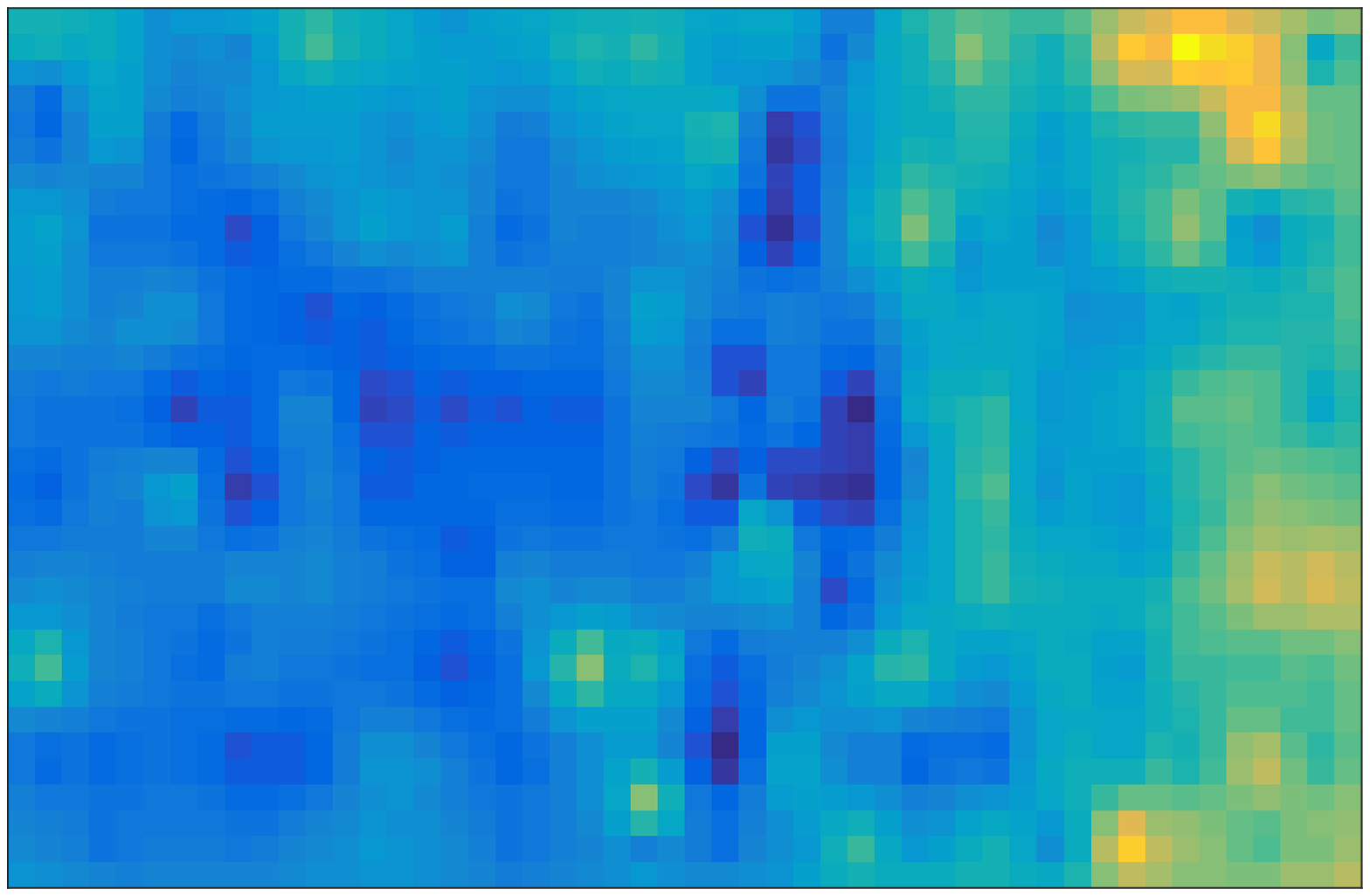}
\end{subfigure}
\begin{subfigure}{0.24\textwidth}
\centering
Cu
\includegraphics[width = 1 \textwidth, keepaspectratio]{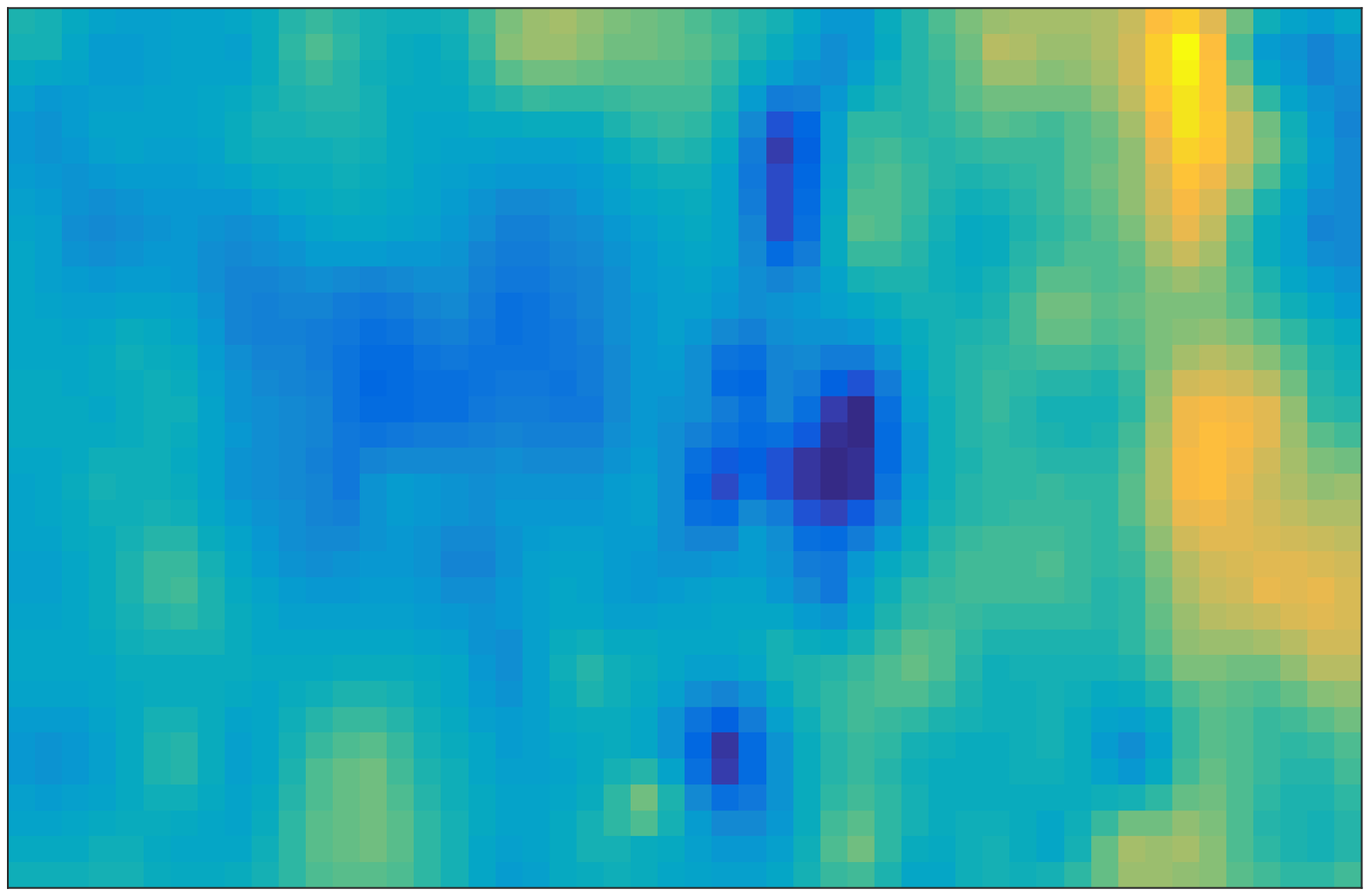}
\end{subfigure}
\begin{subfigure}{0.24\textwidth}
\centering
Fe
\includegraphics[width = 1 \textwidth, keepaspectratio]{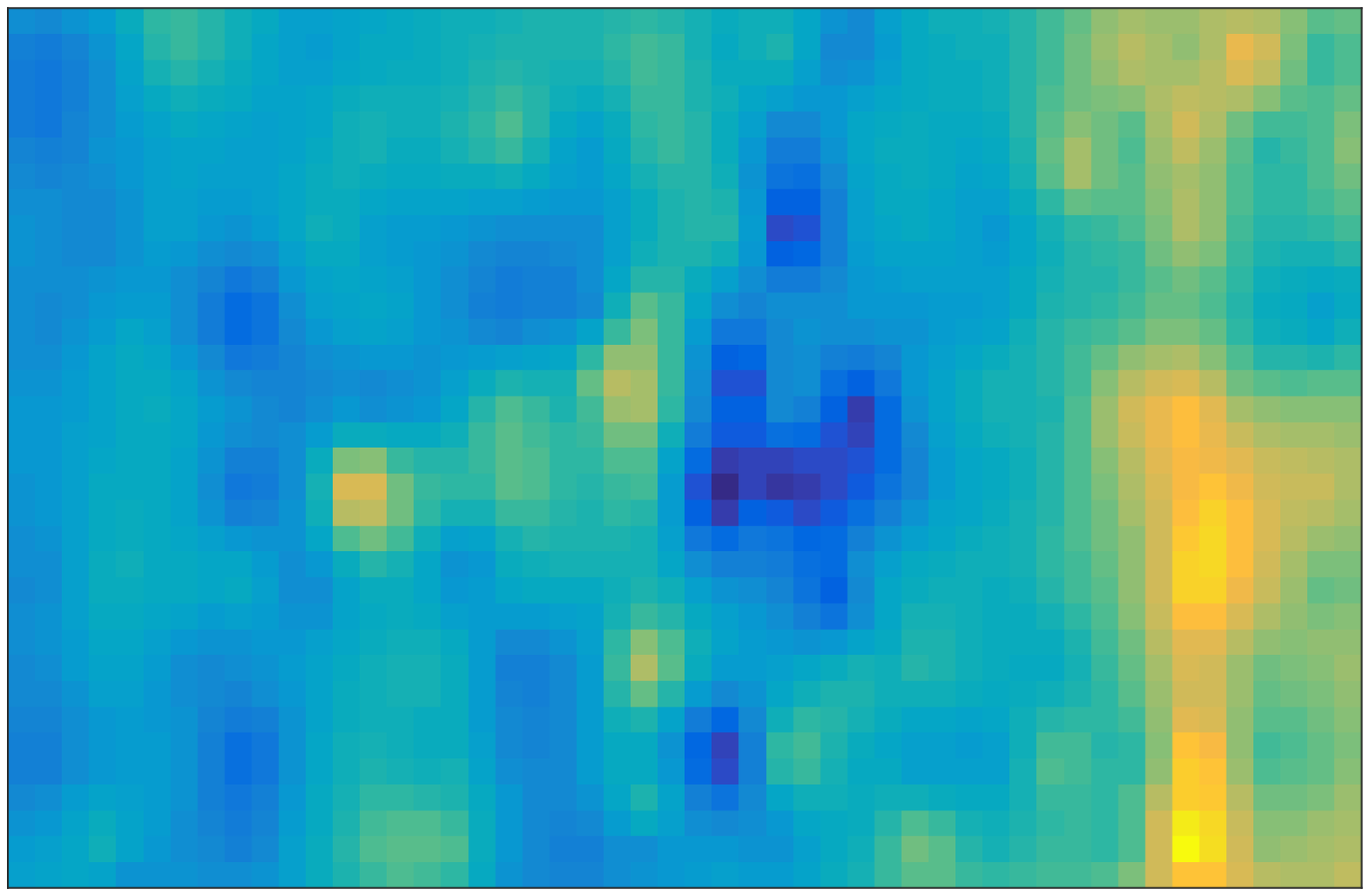}
\end{subfigure}
\begin{subfigure}{0.24\textwidth}
\centering
K
\includegraphics[width = 1 \textwidth, keepaspectratio]{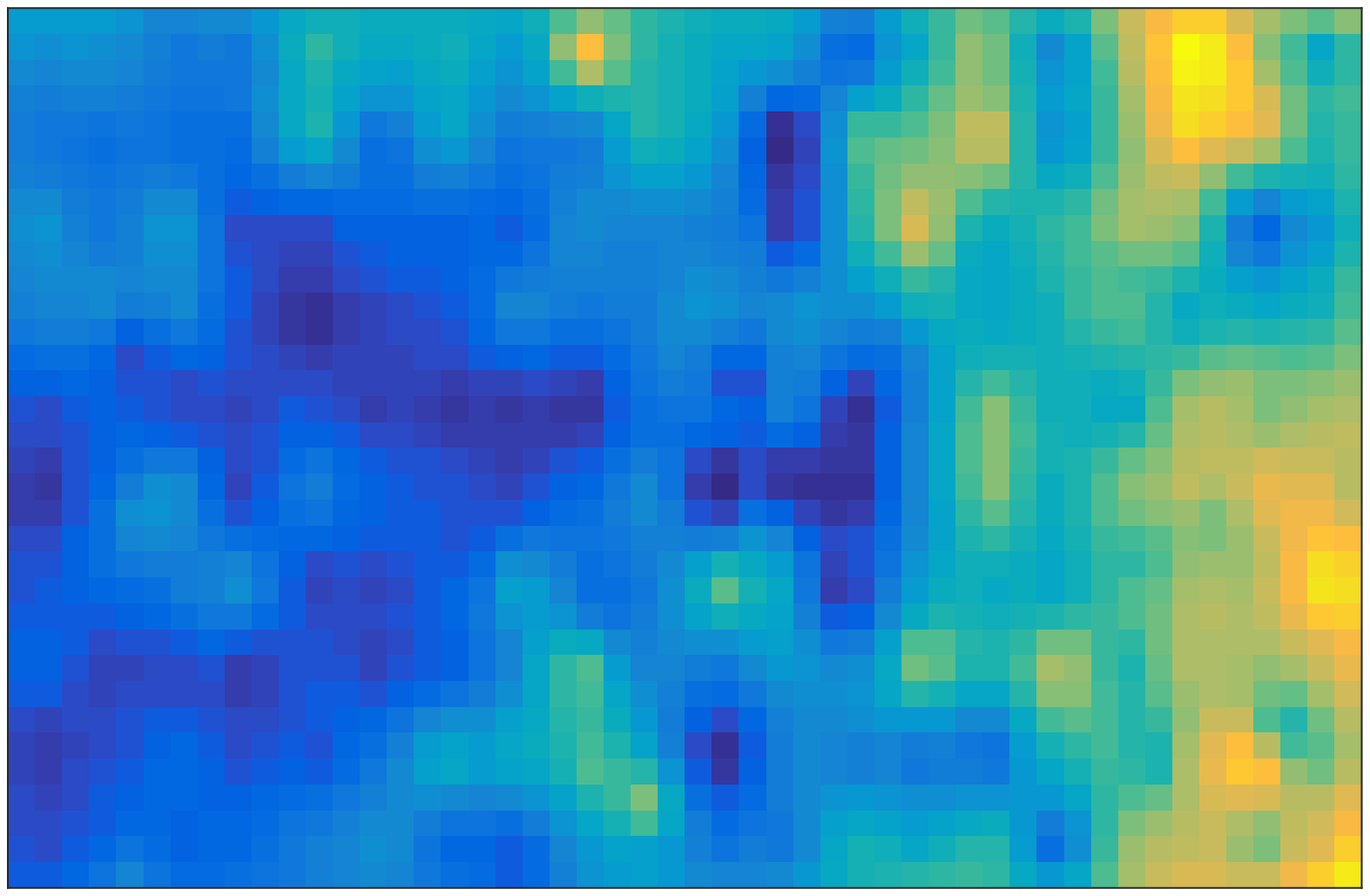}
\end{subfigure} \\

\begin{subfigure}{0.24\textwidth}
\centering
Mg
\includegraphics[width = 1 \textwidth, keepaspectratio]{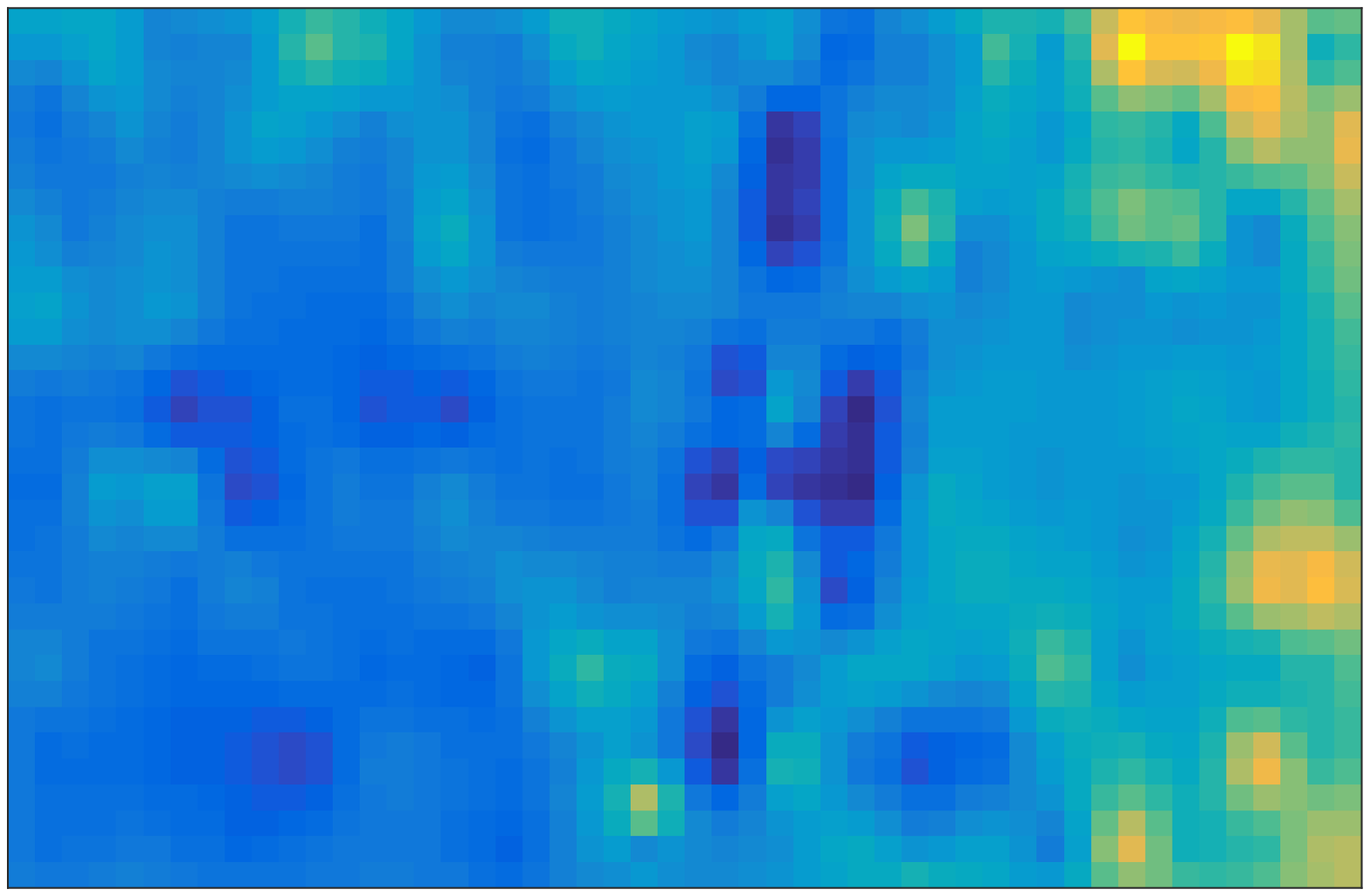}
\end{subfigure}
\begin{subfigure}{0.24\textwidth}
\centering
Mn
\includegraphics[width = 1 \textwidth, keepaspectratio]{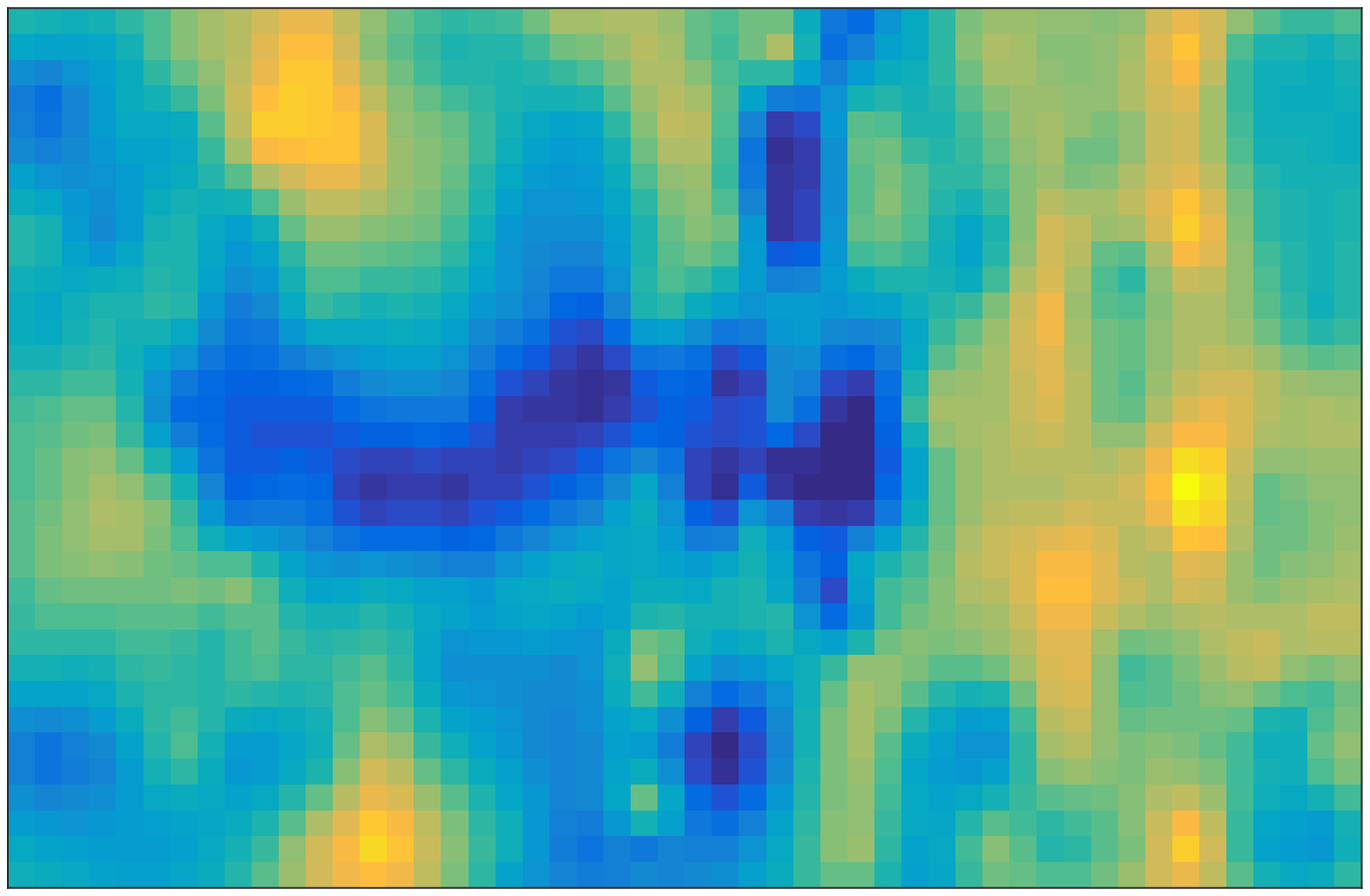}
\end{subfigure}
\begin{subfigure}{0.24\textwidth}
\centering
N
\includegraphics[width = 1 \textwidth, keepaspectratio]{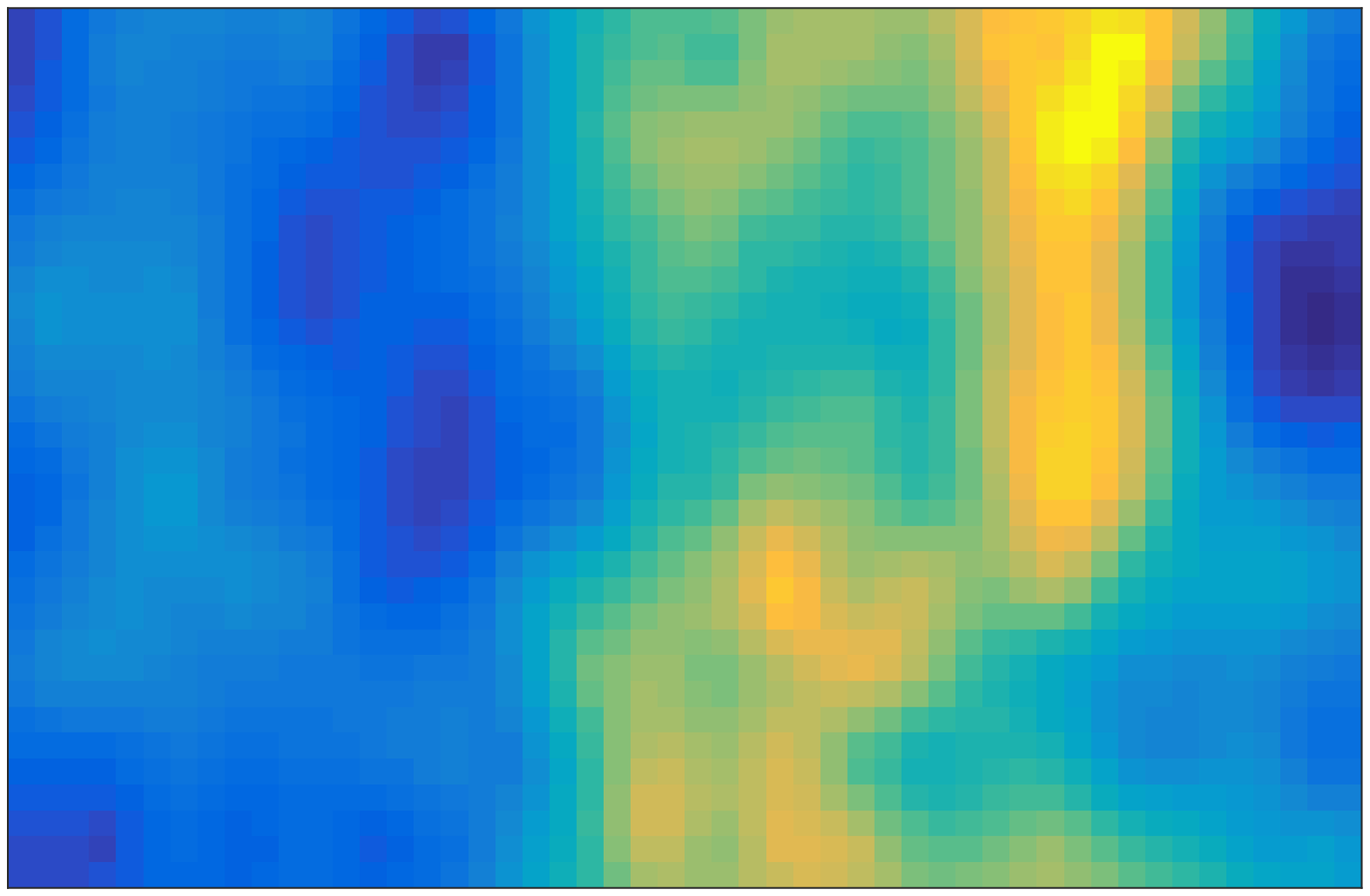}
\end{subfigure}
\begin{subfigure}{0.24\textwidth}
\centering
Nmin
\includegraphics[width = 1 \textwidth, keepaspectratio]{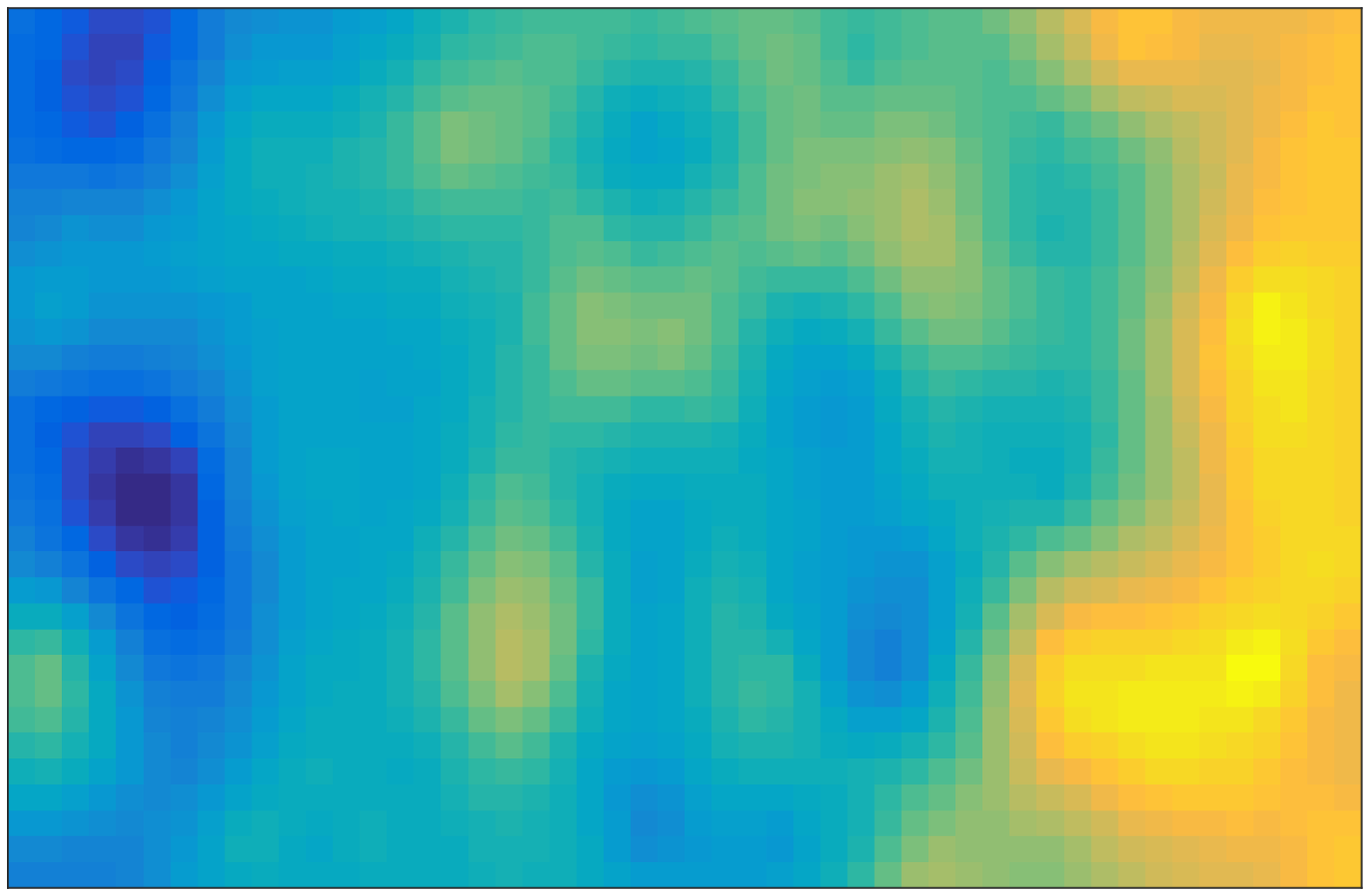}
\end{subfigure} \\

\begin{subfigure}{0.24\textwidth}
\centering
P
\includegraphics[width = 1 \textwidth, keepaspectratio]{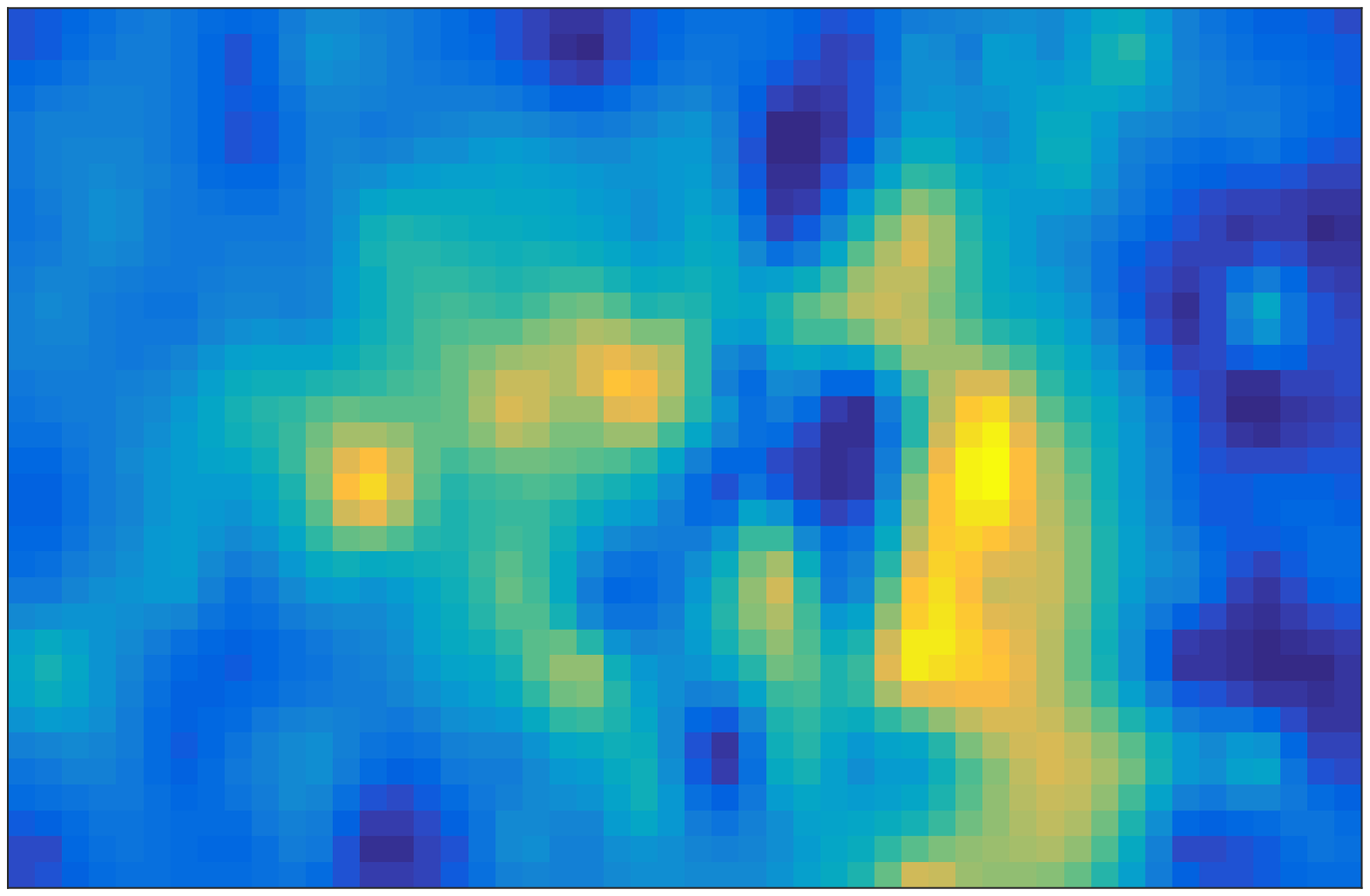}
\end{subfigure}
\begin{subfigure}{0.24\textwidth}
\centering
Zn
\includegraphics[width = 1 \textwidth, keepaspectratio]{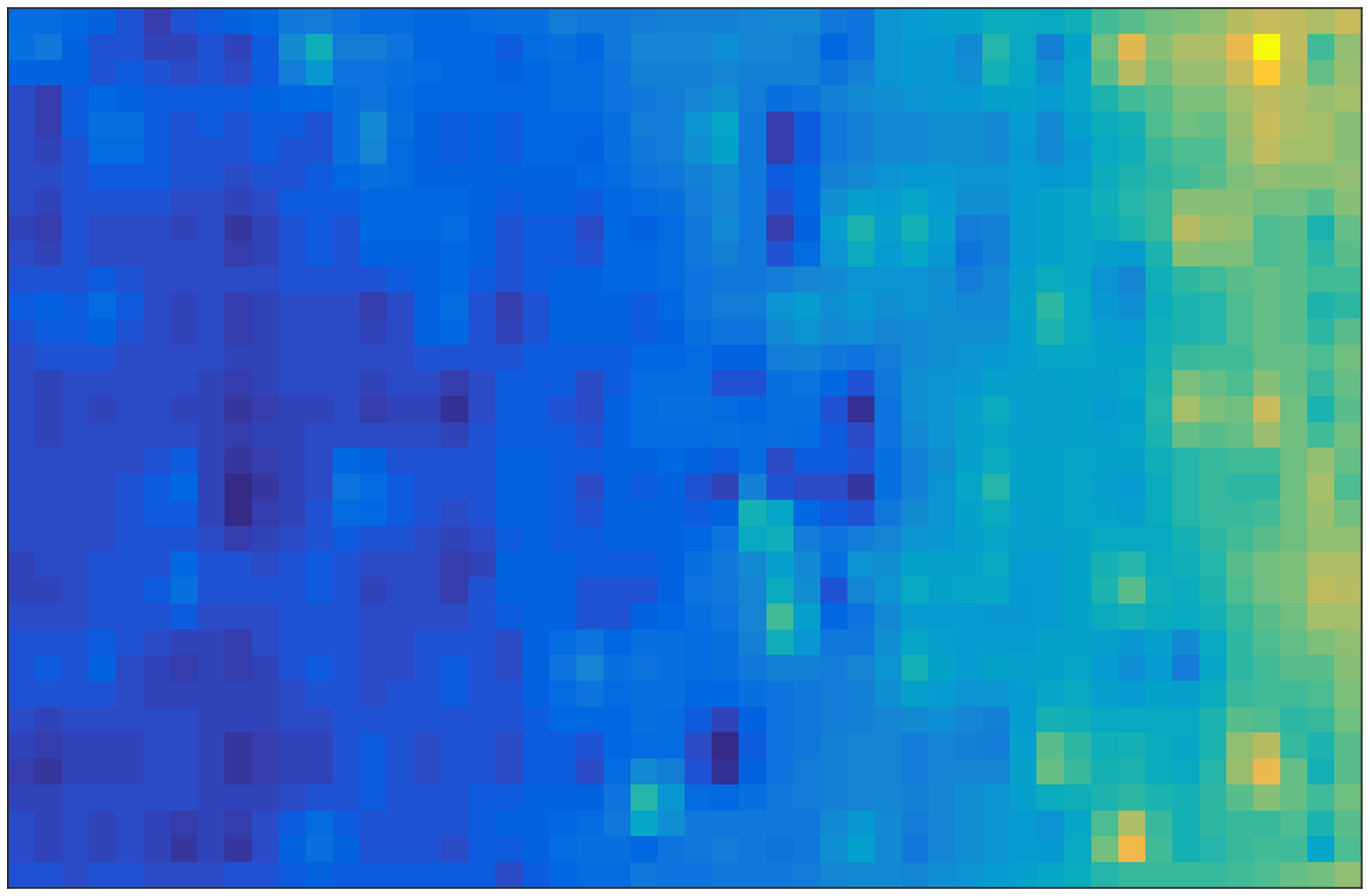}
\end{subfigure}
\begin{subfigure}{0.24\textwidth}
\centering
pH
\includegraphics[width = 1 \textwidth, keepaspectratio]{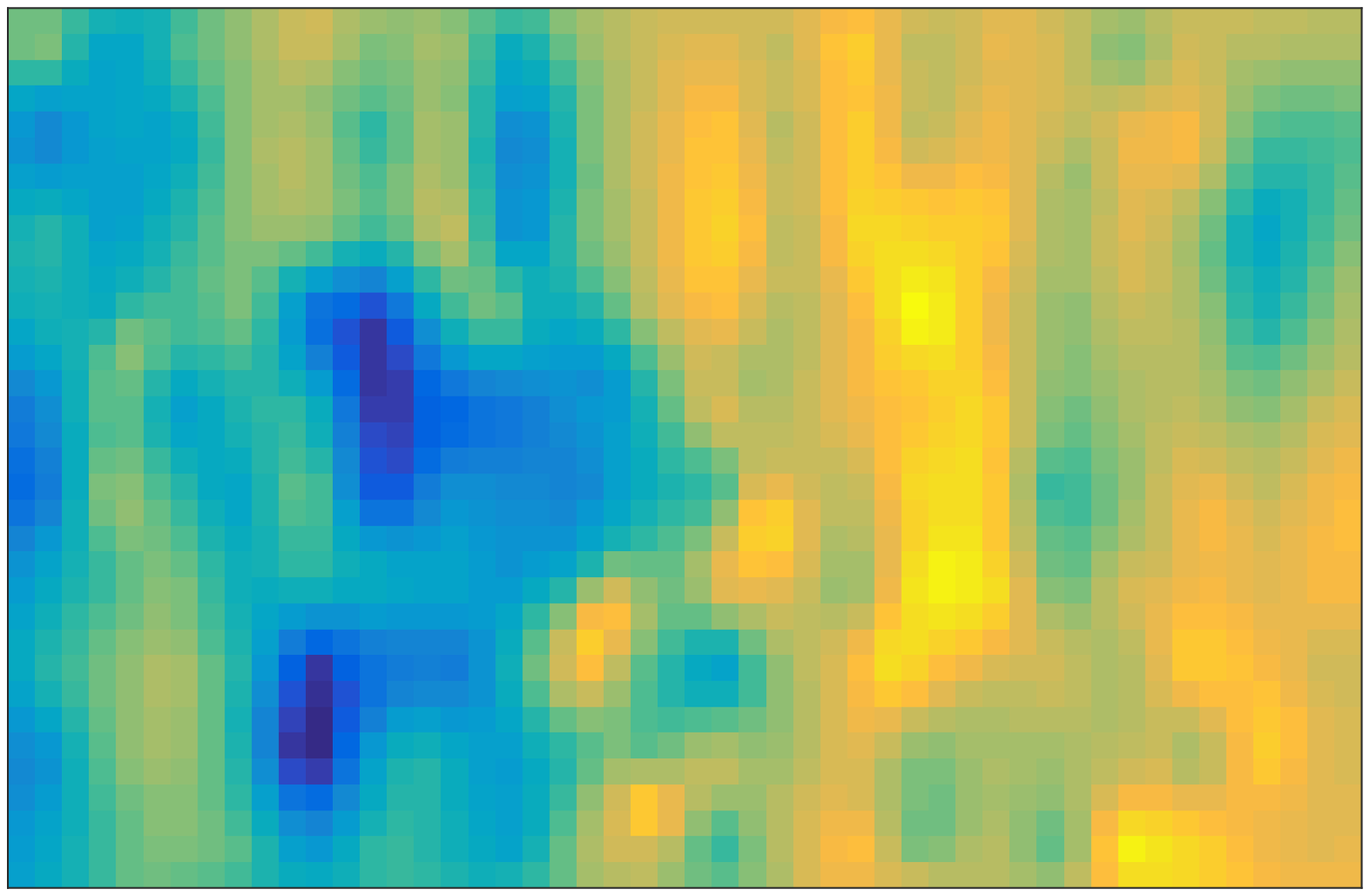}
\end{subfigure} 
\caption{The standardized covariate values on the observational domain.}
\label{fig:allCovs}
\end{figure}

\subsubsection{Models}
As discussed above, it is obvious from Figure \ref{fig:treesIntro} that there is a large area in the middle of the plot where hardly any trees are growing. This indicates that in some parts of the plot, spatial aggregation varies more rapidly than in the other parts. It is likely that some inhibitory factor prevents the trees from growing in that region.  As mentioned earlier, we have 
anecdotal evidence that this area  is covered by a swamp and that the tree species is known to be very unlikely to grow there. 
We test four different models to see how the confounding factor will affect inference. 

The first is a simple Poisson regression model on the covariates, i.e.\ an inhomogeneous Poisson process with linear fixed effects defining the log intensity as $\covars(\psp) \bs{\regCoef}$. We will refer to this model as the Fixed model.  The second model includes a Gaussian field to capture the variability not explained by the covariates. More precisely, we use an LGCP model with log-intensity $\log\lambda(\psp) = \latf(\psp; \covars)$. Here $\latf(\psp; \covars)$ is a Gaussian field with $\expect{\latf(\psp; \covars)} = \covars(\psp) \regCoef$  and a Mat\'ern covariance with standard deviation $\std$ and range $\rangParam$.

Looking at the data, we might  expect the LGCP model to explain the variation in point intensity well, except for the complete lack of observations in the central region coupled with the discontinuity in the observed intensity at the border between the large empty area and the other parts of the plot. If the habitat dependence of the trees is significantly different in these two separated regions, a \ac{lscp} model with a separate class for each of the two regions might provide a better fit. Therefore, the third model is a two-class \ac{lscp} model where the first class is defined as in the LGCP model and the second class has constant intensity. That is, $\log\lambda(\psp) = \mixProb_1(\psp) \latf_1(\psp; \covars) + \mixProb_2(\psp) C_2$. We will refer to this as the \ac{lscp} model. We fixed the parameter of $C_2$ to a small value proportional to the mean intensity among the grid cells with at most $1$ tree.
Finally, we consider a simplified version of this model where $\log\lambda(\psp) = \mixProb_1(\psp) \covars(\psp) \regCoef+ \mixProb_2(\psp) C_2$. This will be referred to as the FixedM model.

The posterior distributions of the parameters and latent fields were estimated using the proposed MCMC method. In order to avoid significant wrap-around effects, the lattice was extended by $350$ m for the level set field and by $220$ m for the latent Gaussian fields of the classes (implicitly assuming correlation ranges smaller than $350$ m for classification and $220$ within classes).The smoothness parameters of the Gaussian fields were fixed at $\smoothParam = 1$ and the following independent prior distributions for the model parameters (when applicable) were used: 
\begin{inparaenum}[\itshape i\upshape)]
\item $\pN(0,10)$ priors for the fixed effects;
\item $\pN(0,4)$-priors for the threshold parameters;
\item an exponential distribution with mean 2, $\text{Exp}(2)$, for the standard deviations of the Gaussian fields  except for $\latf_0$, where $\std_0 = 1$ is fixed;
\item Exp$(200)$ distributions truncated from below at the lattice distance and from above at the lattice extension range for the range parameters $\rangParam_k$;  this ensures that no wrap-around artifacts were introduced and that the correlation range were not smaller than the discretization distance;
\item an Exp$(0.1)$ distribution truncated from above at $1$ for the nugget standard deviation; this yields an expected \textit{a priori} standard deviation of approximately $0.1$ and ensures that the nugget variance  does not dominate the spatial dependency in the level set field. 
\end{inparaenum}

The standard deviations of the Gaussian fields were given exponential priors since this corresponds to the PC prior \citep{lit:sorbye, lit:simpson2} which penalizes deviations from the simpler model without the Gaussian field, where a mean of $2$ penalizes large values. The range parameters were given exponential priors using similar reasoning where no spatial dependency corresponds to the base model. However, ranges below the lattice distance were truncated since no information exist for smaller values due to the spatial discretization. 
The covariates were standardized to mean $0$ and variance $1$. Hence, the fixed effects prior yields a penalisation from the base model of no fixed effects. 
The nugget for the level set field was considered as a deviation from the base model (without a nugget) and hence penalised by an exponential distribution. 

%\subsubsection{Model validation}

To assess the model fit models we used a common approach for point process models \citep{lit:moller2, illian2008statistical, baddeley2015spatial} that compares summary characteristics estimated from the observed point pattern with envelopes based on the summary characteristic estimated for simulated point patterns, generated from each of the four fitted models. 

As a functional summary characteristic we used the centered and variance stabilized $K$-function, commonly referred to as the centered $L$-function, $\hat{L}(r) = \sqrt{\frac{\hat{K}(r)}{\pi}} - r$.  The $L$-function stabilizes the variance such that $\hat{L}(r)$ will be homoscedastic with respect to $r$. Furthermore, the $-r$ term centers the function in the sense that the resulting function for a homogeneous Poisson process has the value of zero for all distances \citep{illian2008statistical}. 
We used isotropic edge correction and calculated the envelope using the functions \textit{Kest} and \textit{envelope} from the \textit{spatstat} package \citep{lit:spatstat}. The $L$-function for the observed point pattern as well as the pointwise sample mean and $90\%$ envelopes from $5 000$ realizations for each of the four models can be seen in the top row of Figure \ref{fig:lEnvelopes}. 

\begin{figure}[t]
\centering
\begin{subfigure}{0.24\textwidth}
\centering
\ac{lscp} model.
\includegraphics[width = \textwidth, keepaspectratio]{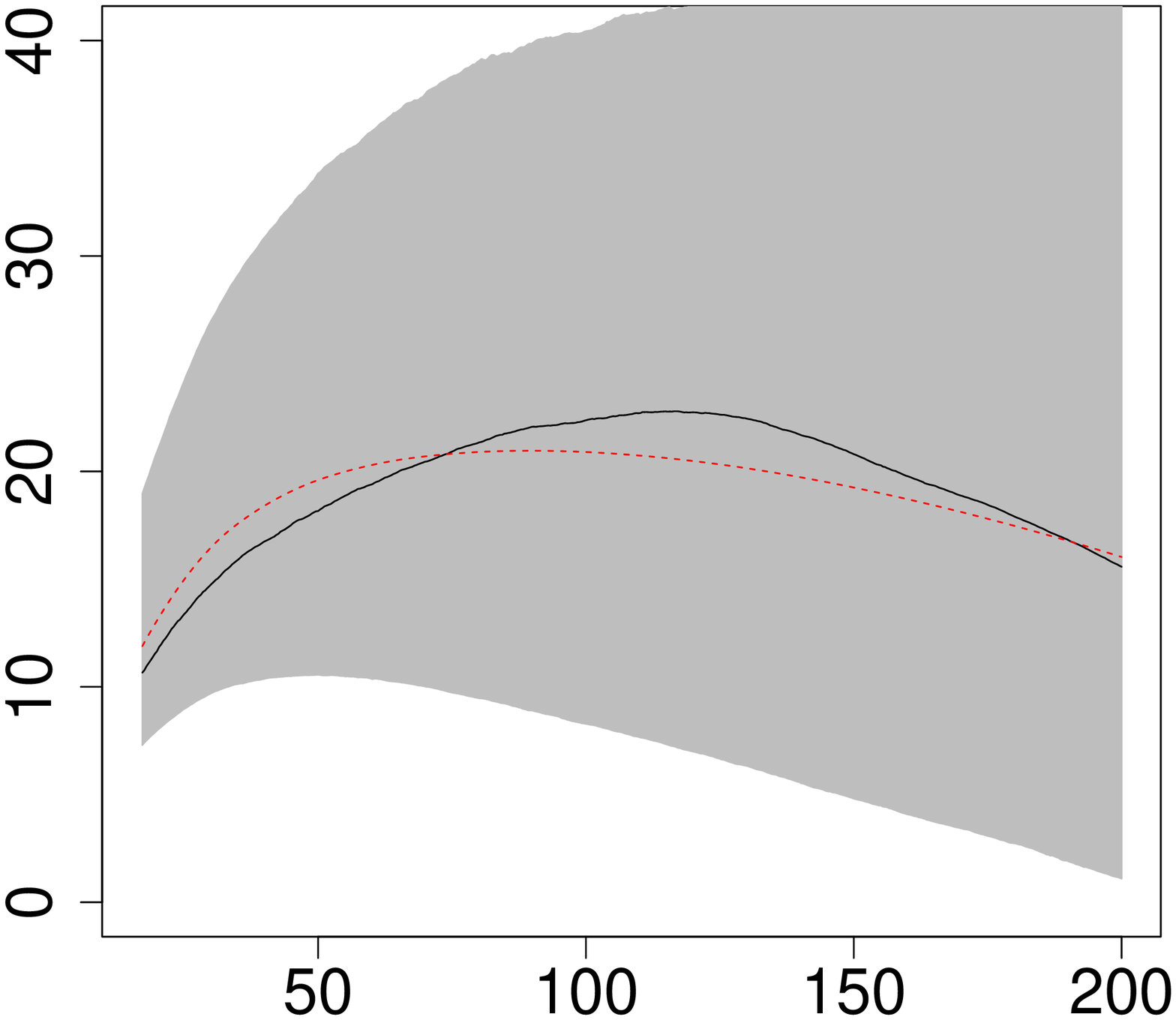}
\end{subfigure}
\begin{subfigure}{0.24\textwidth}
\centering
LGCP model
\includegraphics[width = \textwidth, keepaspectratio]{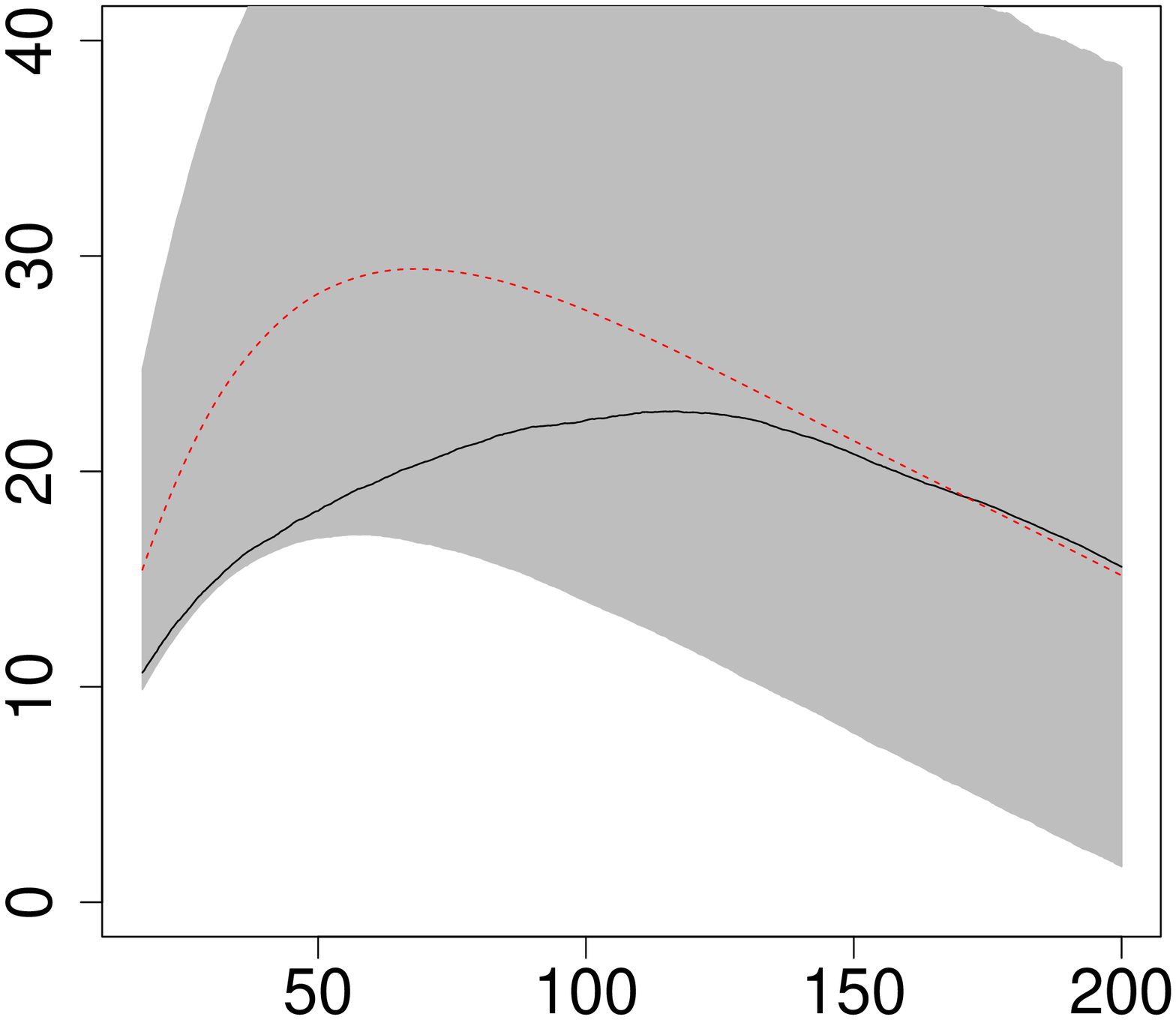}
\end{subfigure} 
\begin{subfigure}{0.24\textwidth}
\centering
FixedM  model.
\includegraphics[width = \textwidth, keepaspectratio]{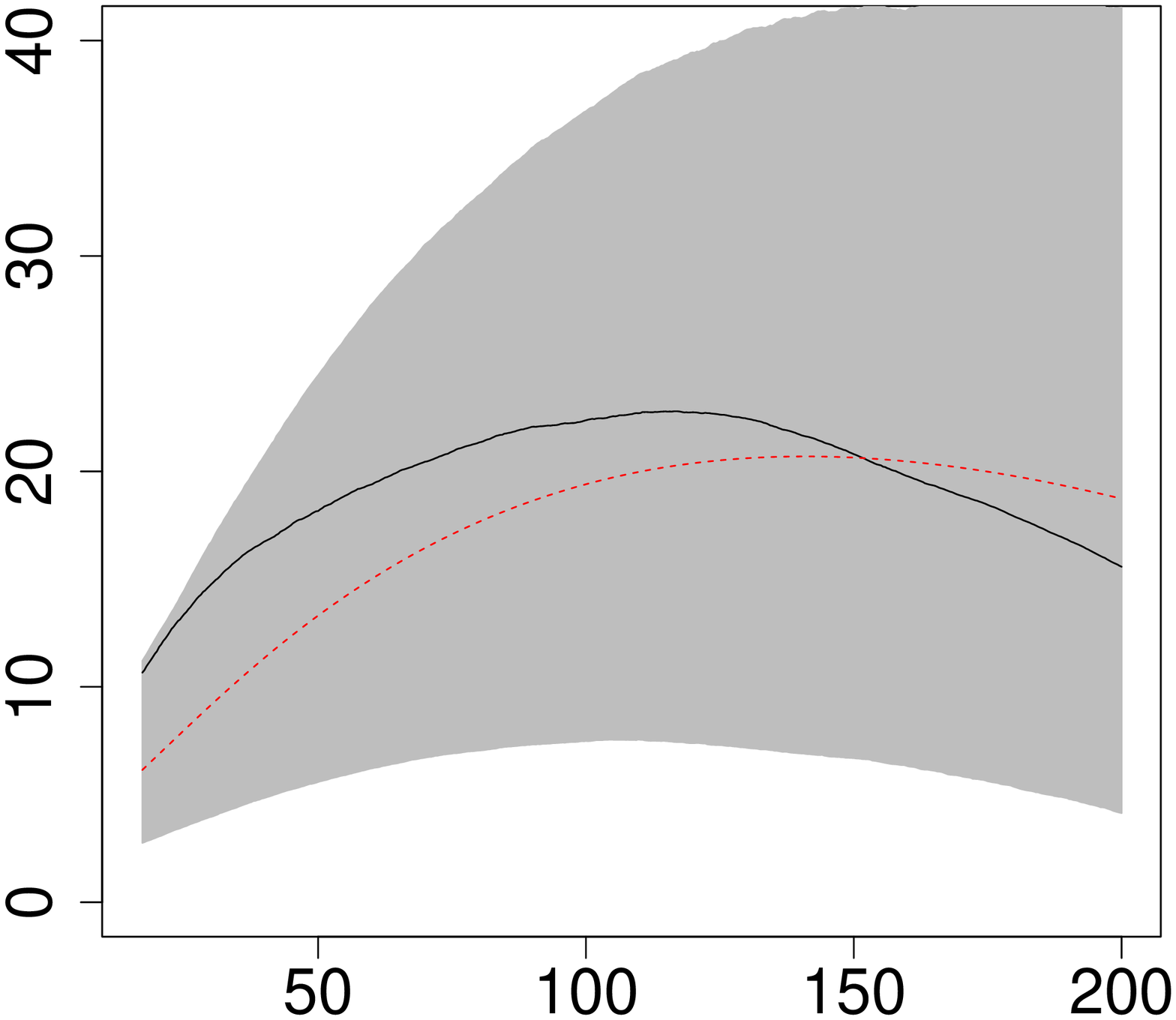}
\end{subfigure}
\begin{subfigure}{0.24\textwidth}
\centering
Fixed model.
\includegraphics[width = \textwidth, keepaspectratio]{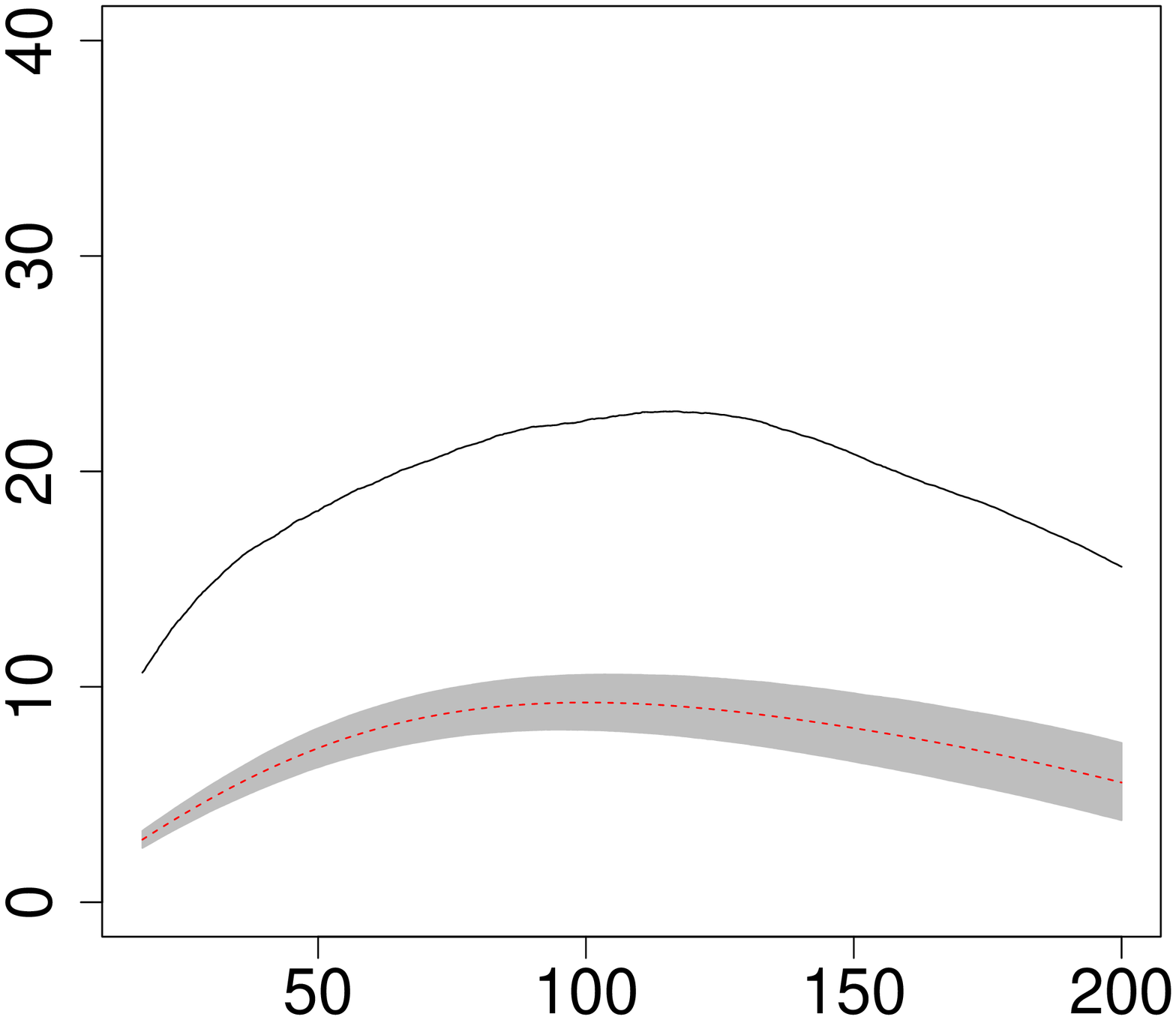}
\end{subfigure}\\
\begin{subfigure}{0.24\textwidth}
\centering
\includegraphics[width = \textwidth, keepaspectratio]{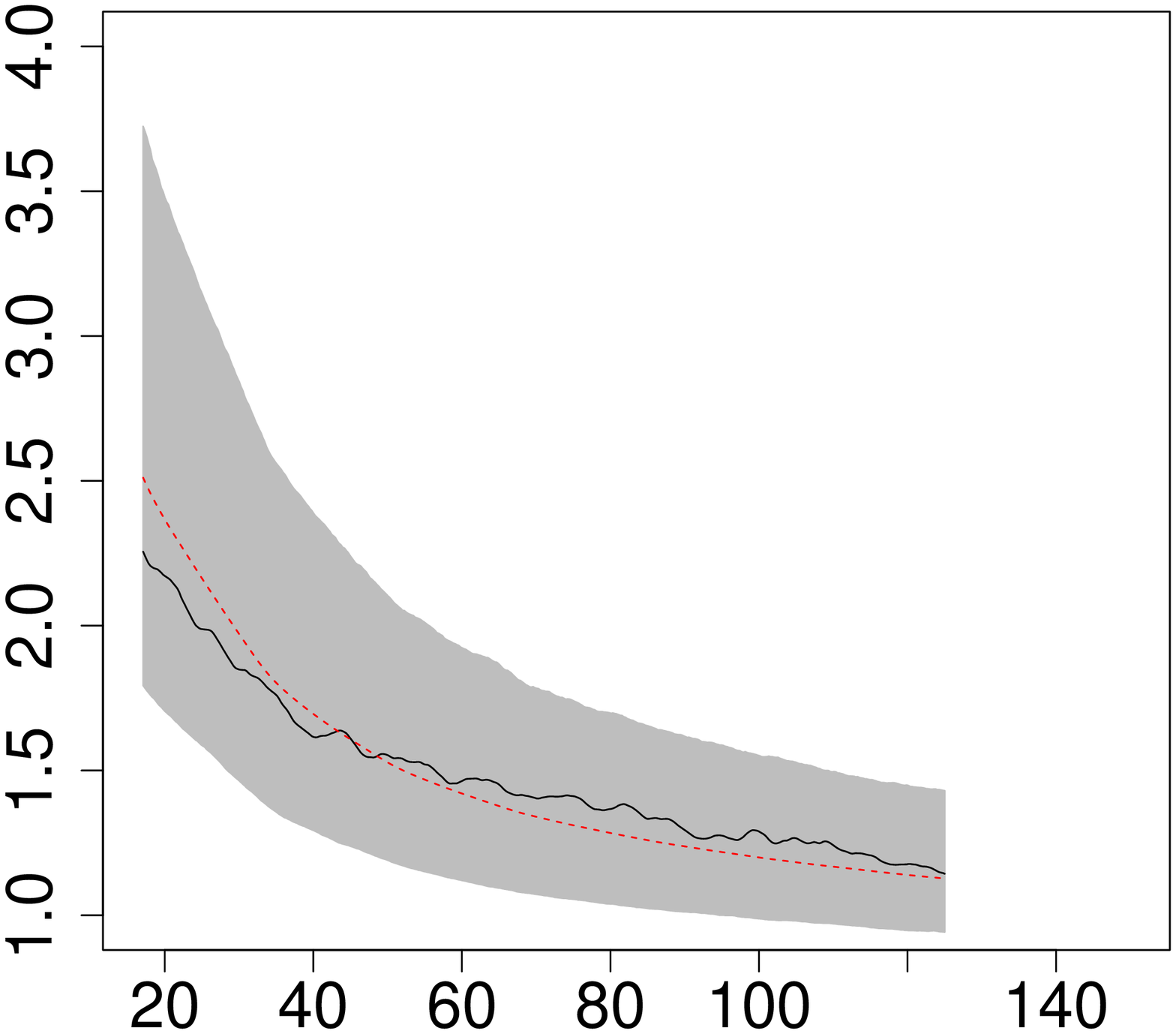}
\end{subfigure}
\begin{subfigure}{0.24\textwidth}
\centering
\includegraphics[width = \textwidth, keepaspectratio]{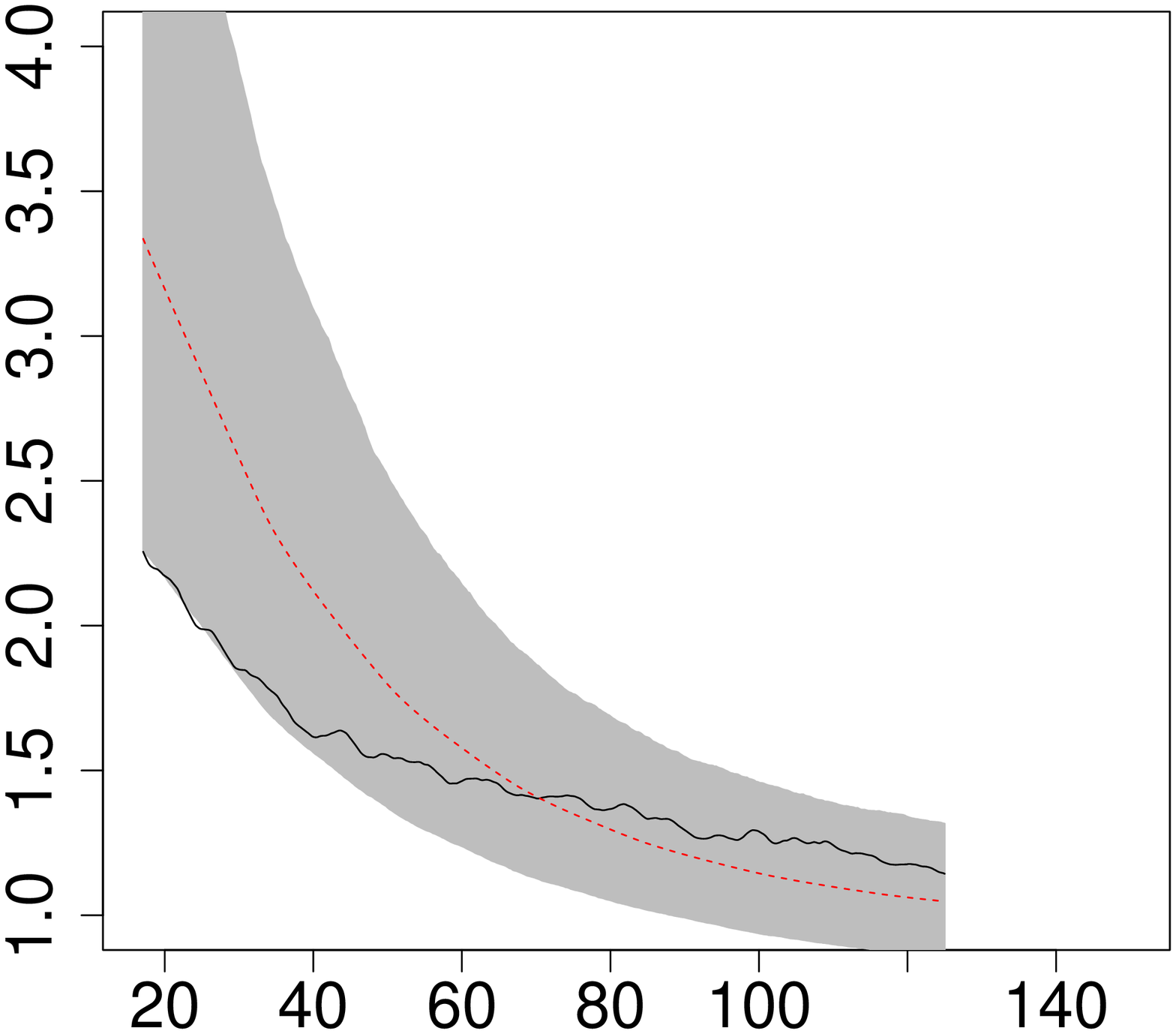}
\end{subfigure} 
\begin{subfigure}{0.24\textwidth}
\centering
\includegraphics[width = \textwidth, keepaspectratio]{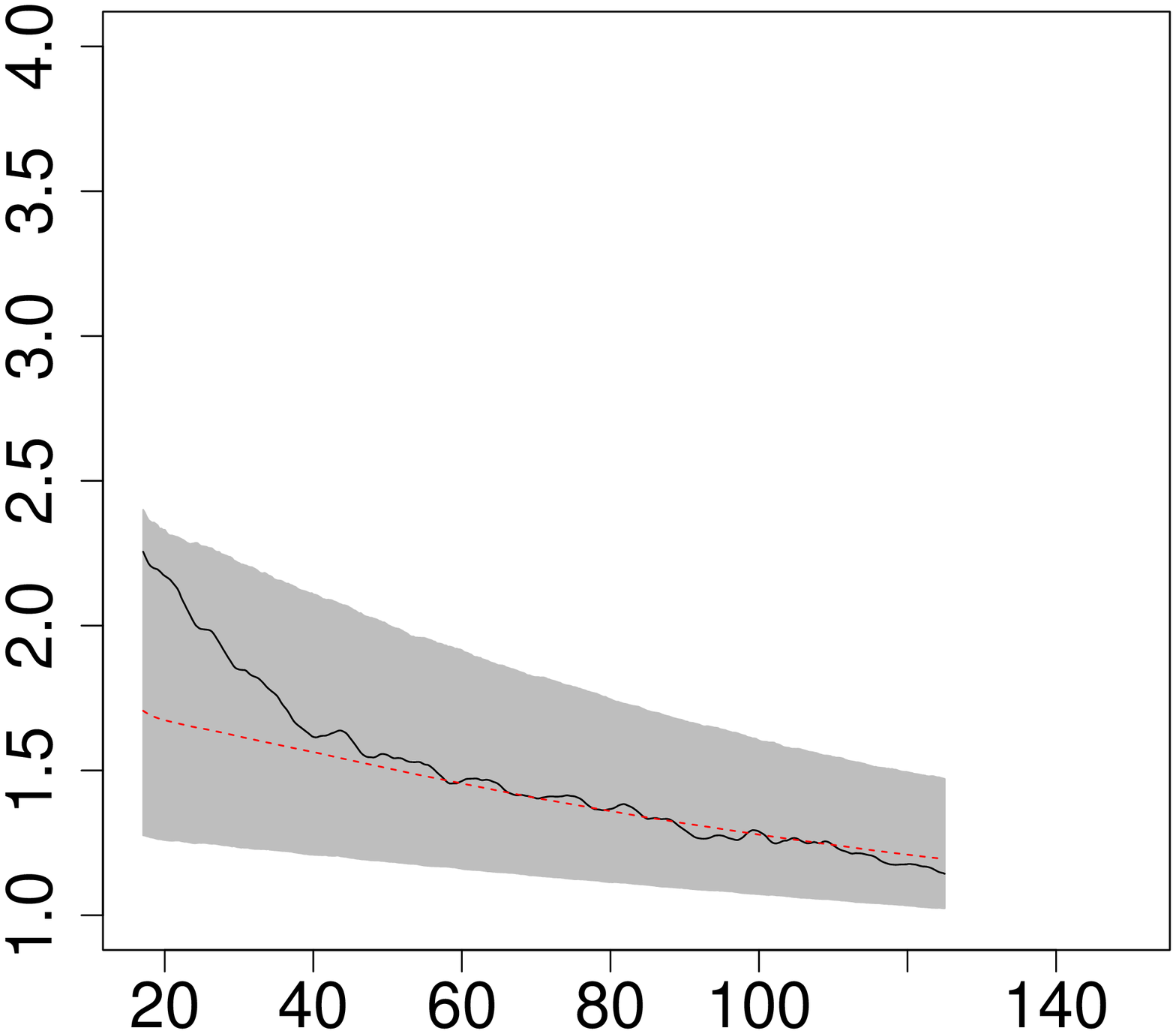}
\end{subfigure}
\begin{subfigure}{0.24\textwidth}
\centering
\includegraphics[width = \textwidth, keepaspectratio]{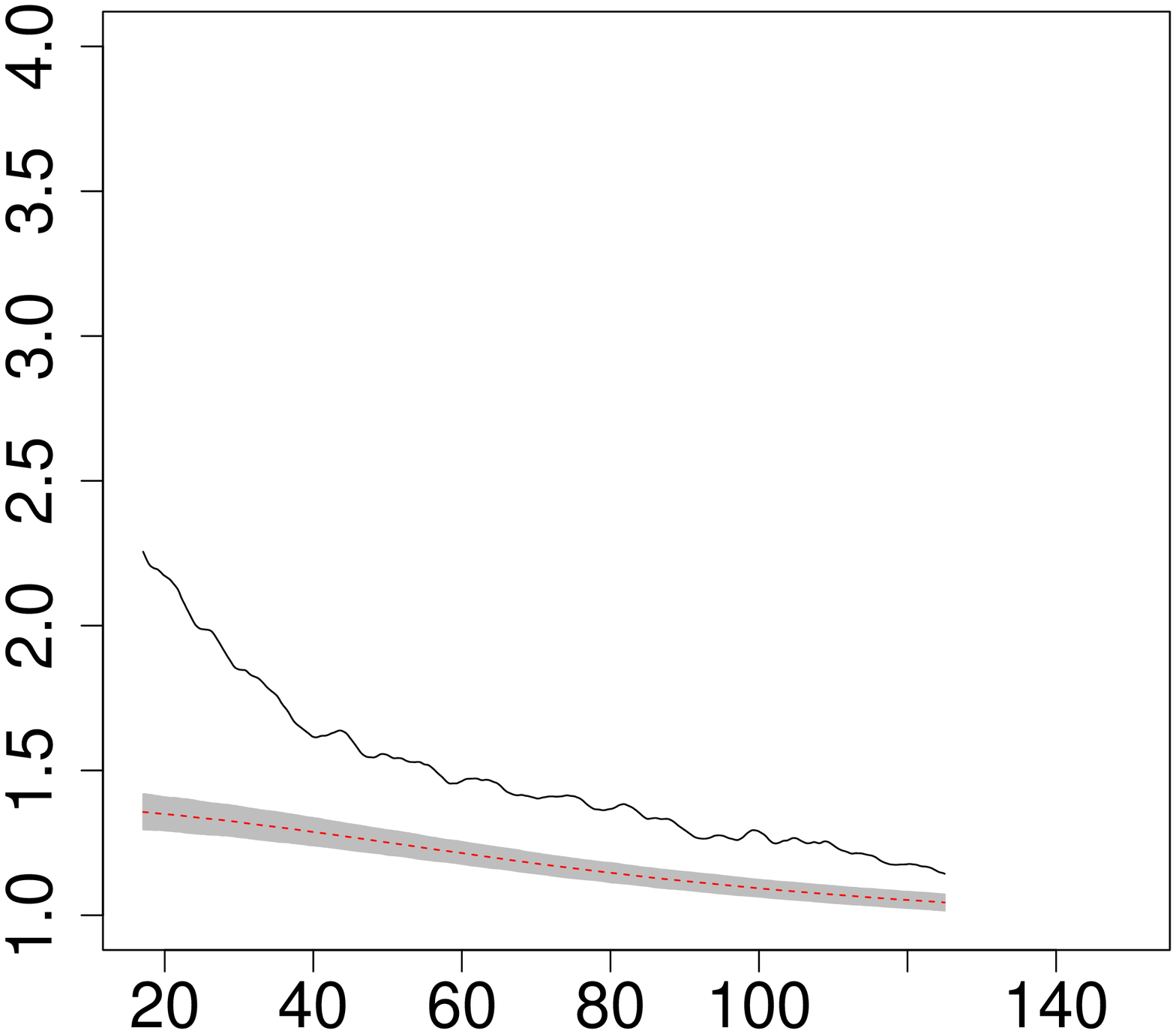}
\end{subfigure}\\
\begin{subfigure}{0.24\textwidth}
\centering
\includegraphics[width = \textwidth, keepaspectratio]{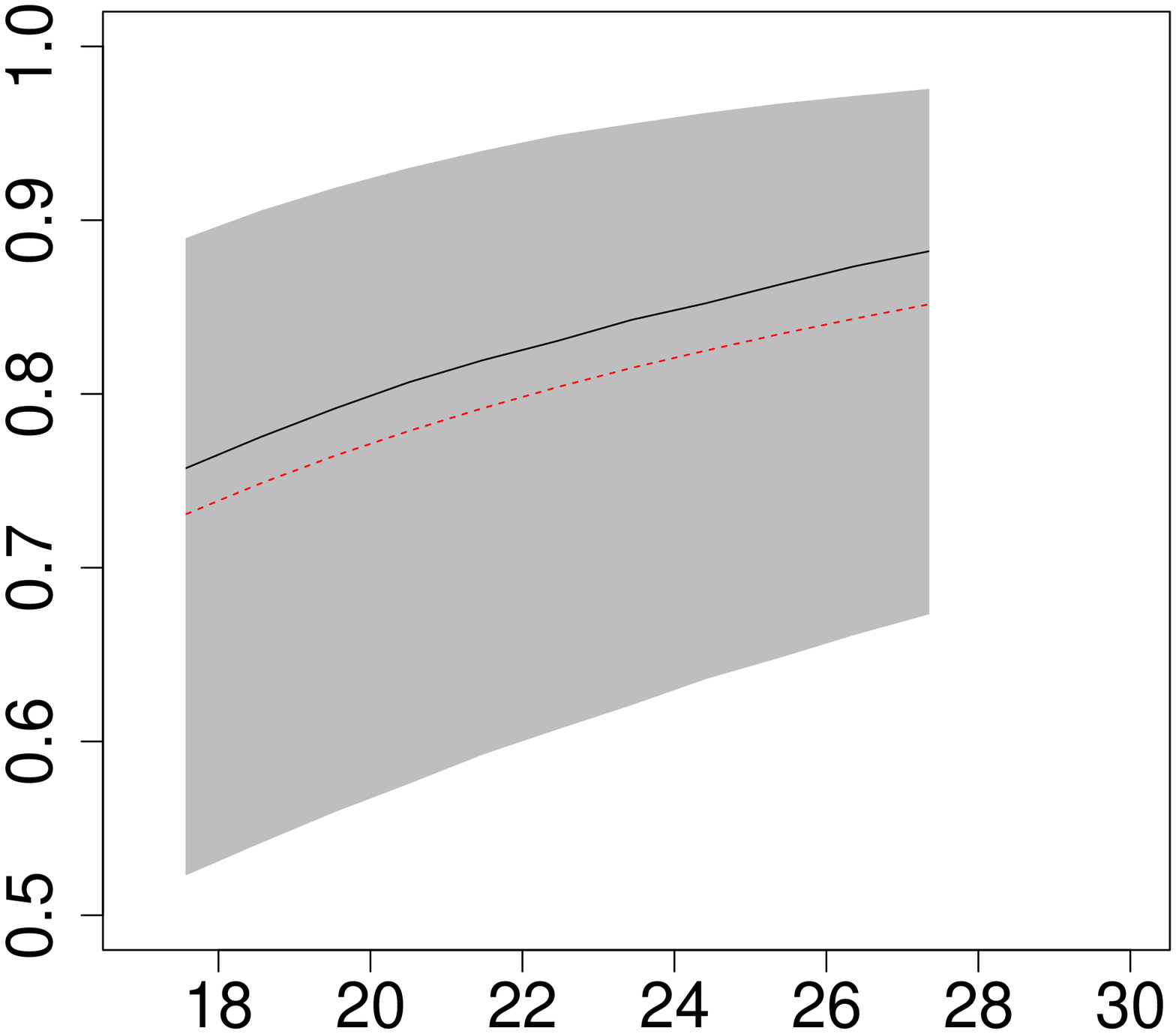}
\end{subfigure}
\begin{subfigure}{0.24\textwidth}
\centering
\includegraphics[width = \textwidth, keepaspectratio]{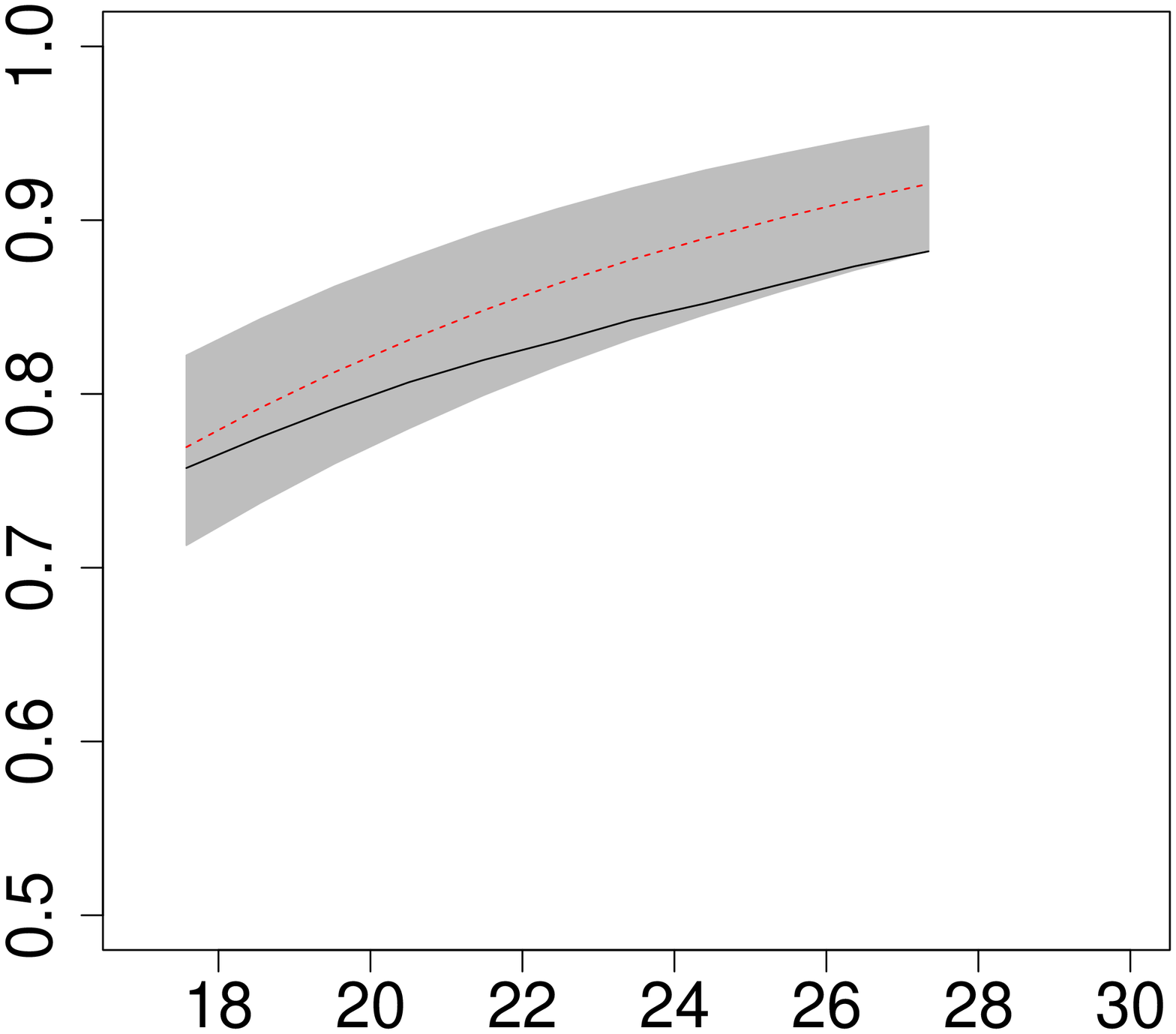}
\end{subfigure} 
\begin{subfigure}{0.24\textwidth}
\centering
\includegraphics[width = \textwidth, keepaspectratio]{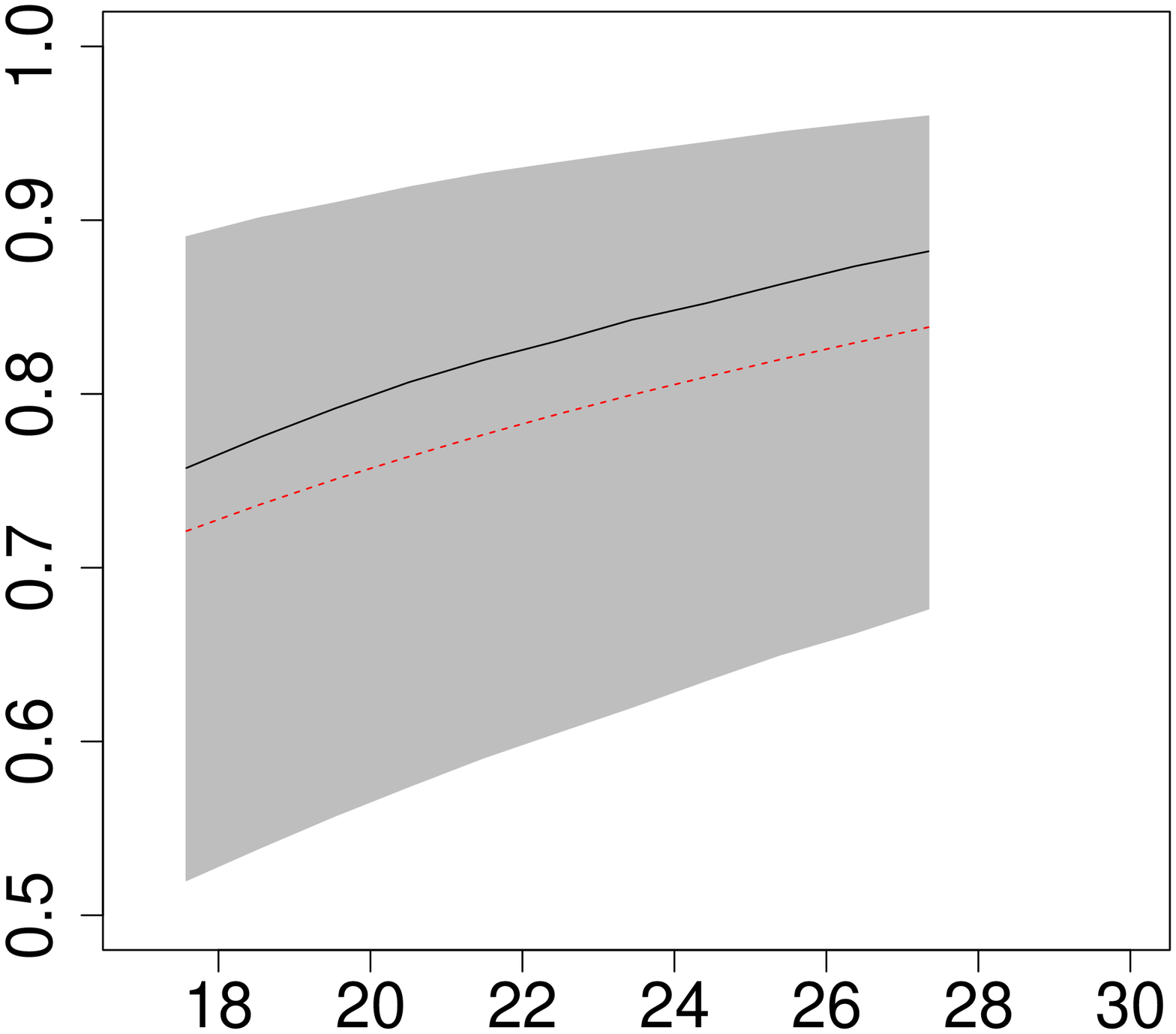}
\end{subfigure}
\begin{subfigure}{0.24\textwidth}
\centering
\includegraphics[width = \textwidth, keepaspectratio]{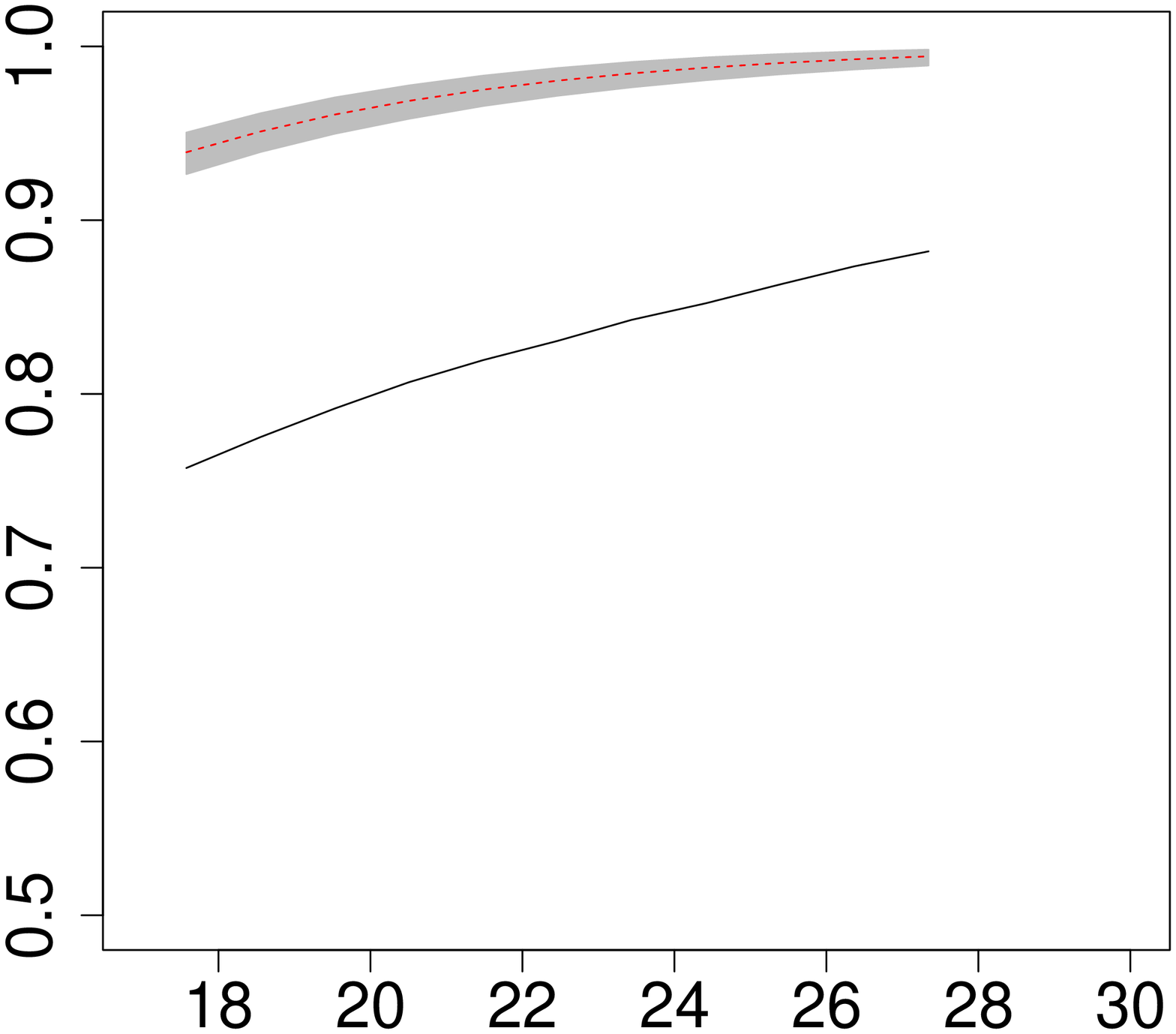}
\end{subfigure}
\caption{Plots of the L-function (top row), pair correlation function (mid row), and empty space function (bottom row). In each case, the black line corresponding to the value estimated from the real point pattern, the thin line corresponding to the model mean and the envelopes being the pointwise $90\%$ envelopes of $5 000$ simulated point patterns for the corresponding model.}
\label{fig:lEnvelopes}
\end{figure}

Not surprisingly, the Fixed model seems clearly inappropriate, as the functional summary characteristics for the observed point pattern is far outside the envelopes for all distances. 
For all other models the estimated function for the observed point pattern remains inside the envelopes. 
The estimated function for the observed pattern and the expected values of the simulated patterns are most similar for the LSCP model. 
For the standard LGCP model the empirical pattern appears to show less clustering than the one expected from the model. This can be seen by low values of the black line relative to the red line for a large range of $r$, and the function only just within the envelope at smaller distances. 
Due to the cumulative nature of the $L$-function it is hard to discern at what inter-point distances this discrepancy occurs in the patterns. 
Interestingly, the differences between the two lines are less drastic for the FixedM model and reversed. The simulated patterns seem to be less clustered.

Secondly, we compute the pair correlation function for the different models as well as $90\%$ pointwise envelopes and expected values in a similar fashion as for the $L$-function. The result is shown in the second row of Figure  \ref{fig:lEnvelopes}, with the function for the empirical pattern again deviating drastically from those for the simulated patterns, for the Fixed model.
For the LGCP model the function for the empirical pattern is only just outside the envelopes and again below the mean function, indicating less clustering. This discrepancy vanished at a a distance $r$ of about $70$ meters. For the FixedM model the deviations are mainly at short distances where the empirical pattern show greater amount of clustering than what would be expected by the model.

Finally, we compute the empty space function and corresponding $90\%$ pointwise envelopes and expected values. These are shown in the bottom row of Figure  \ref{fig:lEnvelopes}, and one can see that the Fixed model once again deviates drastically from the empirical pattern. 
Here the LSCP model show lower values than the empirical pattern while the FixedM model show larger ones. The LGCP model show an increasing over estimation of clustering and the lower border of the envelope touches the empirical value at far right of the range.
Based on the comparison between the pointwise expected values and envelopes with the empirical values, the LSCP model seem to explain the observed point pattern better than the other models.

\subsubsection{Analysis of covariates and spatial structure}

\begin{figure}[t]
\centering
\begin{subfigure}{0.24\textwidth}
\centering
$\fuzzy$
\includegraphics[keepaspectratio, width = 0.95\textwidth]{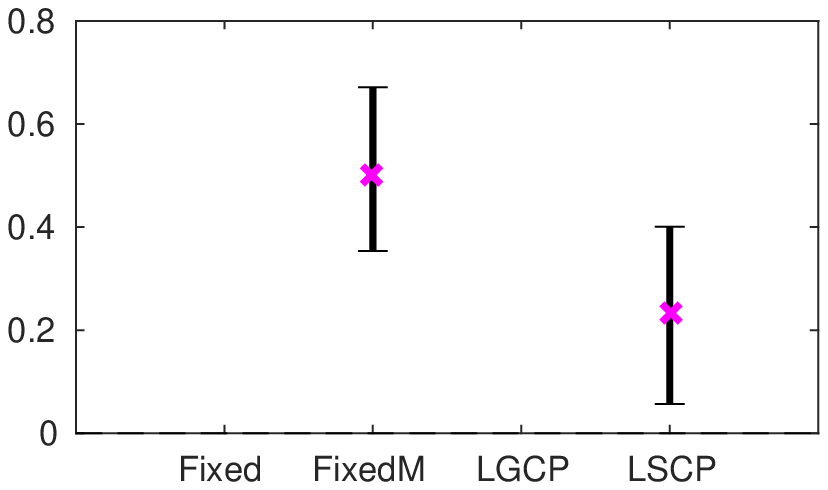}
\end{subfigure}
\begin{subfigure}{0.24\textwidth}
\centering
$\rangParam_0$
\includegraphics[keepaspectratio, width = 0.95\textwidth]{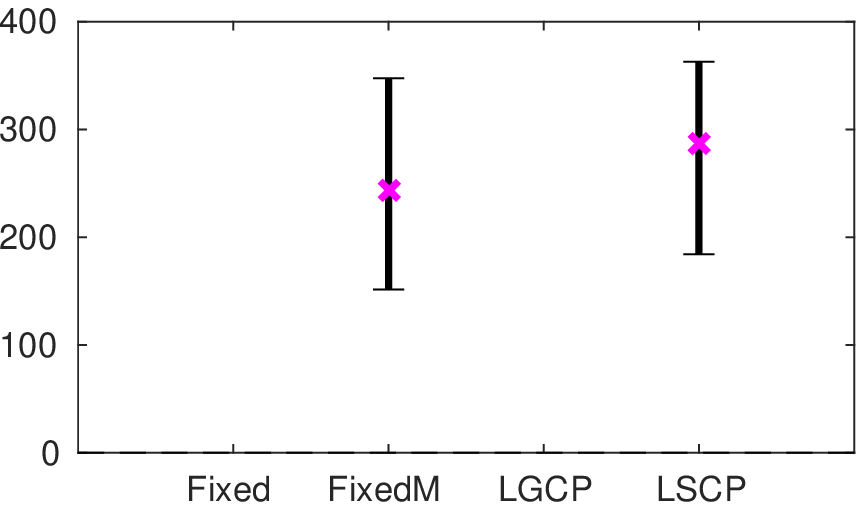}
\end{subfigure}
\begin{subfigure}{0.24\textwidth}
\centering
$\threshParam_1$
\includegraphics[keepaspectratio, width = 0.95\textwidth]{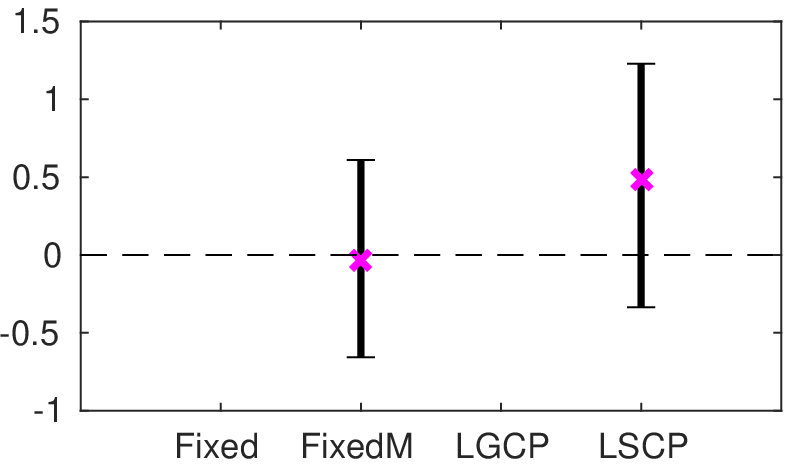}
\end{subfigure}
 \\

\begin{subfigure}{0.24\textwidth}
\centering
$\std_1$
\includegraphics[keepaspectratio, width = 0.95\textwidth]{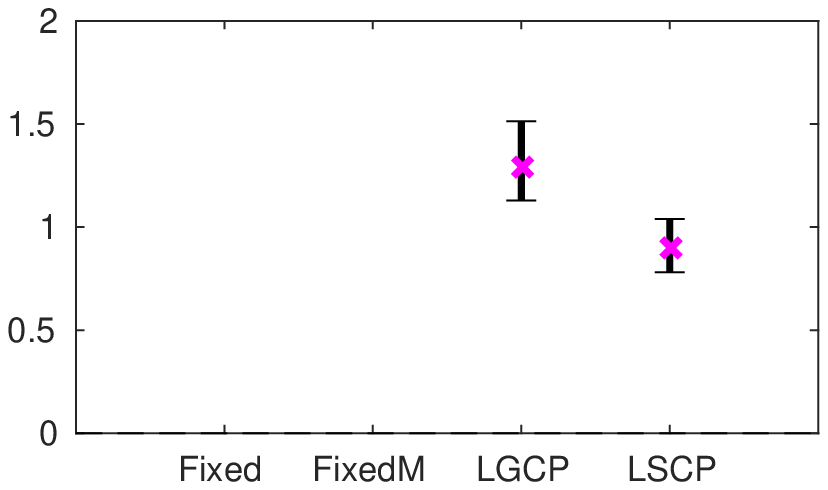}
\end{subfigure}
\begin{subfigure}{0.24\textwidth}
\centering
$\rangParam_1$
\includegraphics[keepaspectratio, width = 0.95\textwidth]{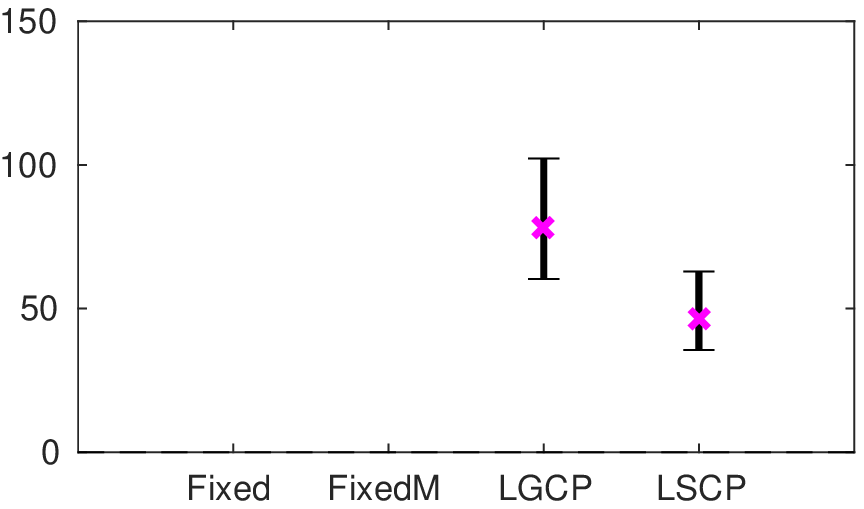}
\end{subfigure}
 \\

\caption{ The mean (cross) and 95\% credibility intervals (lines) for the field parameters.}
\label{fig:fieldParamcredintervals} 
\end{figure}

The models discussed here, relating a spatial pattern to the spatially continuous covariates may be of interest for a number of reasons. Commonly, one seeks to understand habitat preferences of a particular species as reflected in the relationship between the point pattern and the covariates. In addition, it might be of interest to understand the nature of the spatial structure that remains unexplained by the covariates. This might be gleaned from the parameters of the covariance function of the Gaussian random field(s). 

To investigate the spatial structure, we first look at the mean value and $95\%$ credibility intervals for the random field parameters of the models. These are presented in Figure \ref{fig:fieldParamcredintervals}. Observing the difference between $\rangParam_1$, and $\std_1$ values of the LGCP and \ac{lscp} models show how the empty region will affect the estimation of the spatial dependency structure. Here, the \ac{lscp} model shows a significantly lower variance and clearly lower correlation range. This is natural since the Gaussian field for the \ac{lscp} model does not need to explain both the effect of natural spatial dependency between growth of trees as well as the unknown inhibitory effect that causes trees to not grow at all in certain regions of the forest. And finally we note that $\sigma_\epsilon$ has a large effect (signal to noise ratio equals $\frac{1}{\sigma_\epsilon}$) indicating that the Mat\'ern field for $X_0$ cannot explain the classification on its own. This is clearer for the FixedM model, where classification jumps more sporadically between adjacent grid cells due to the over-simplified structure of the classes.

In Figure \ref{fig:postInts} the mean posterior log intensities, $\{\log \ints_i \}_{i=1}^N$ are presented as kriging predictions for each of the four models. 
The figure also shows the posterior probabilities $\condprob{\latf_0(\psp) > \threshParam_1}{\obsf}$, giving an indication of the region with very few trees. 
The posterior log intensity surface of the LSCP shows sharp boundaries contrary to the smoothly varying in the LGCP. The classification in the FixedM model is more noisy than that of the LSCP model and a larger proportion of the observation window is classified as being the \textit{empty region}. Once again, this is connected to the larger value of $\fuzzy$ and caused by the FixedM having to explain the intensity with a much simpler model.

\begin{figure}[t]
\centering
	\begin{subfigure}{0.59\textwidth}
	\centering
		\begin{subfigure}{0.49\textwidth}
		\centering
		\includegraphics[width = 1 \textwidth, keepaspectratio]{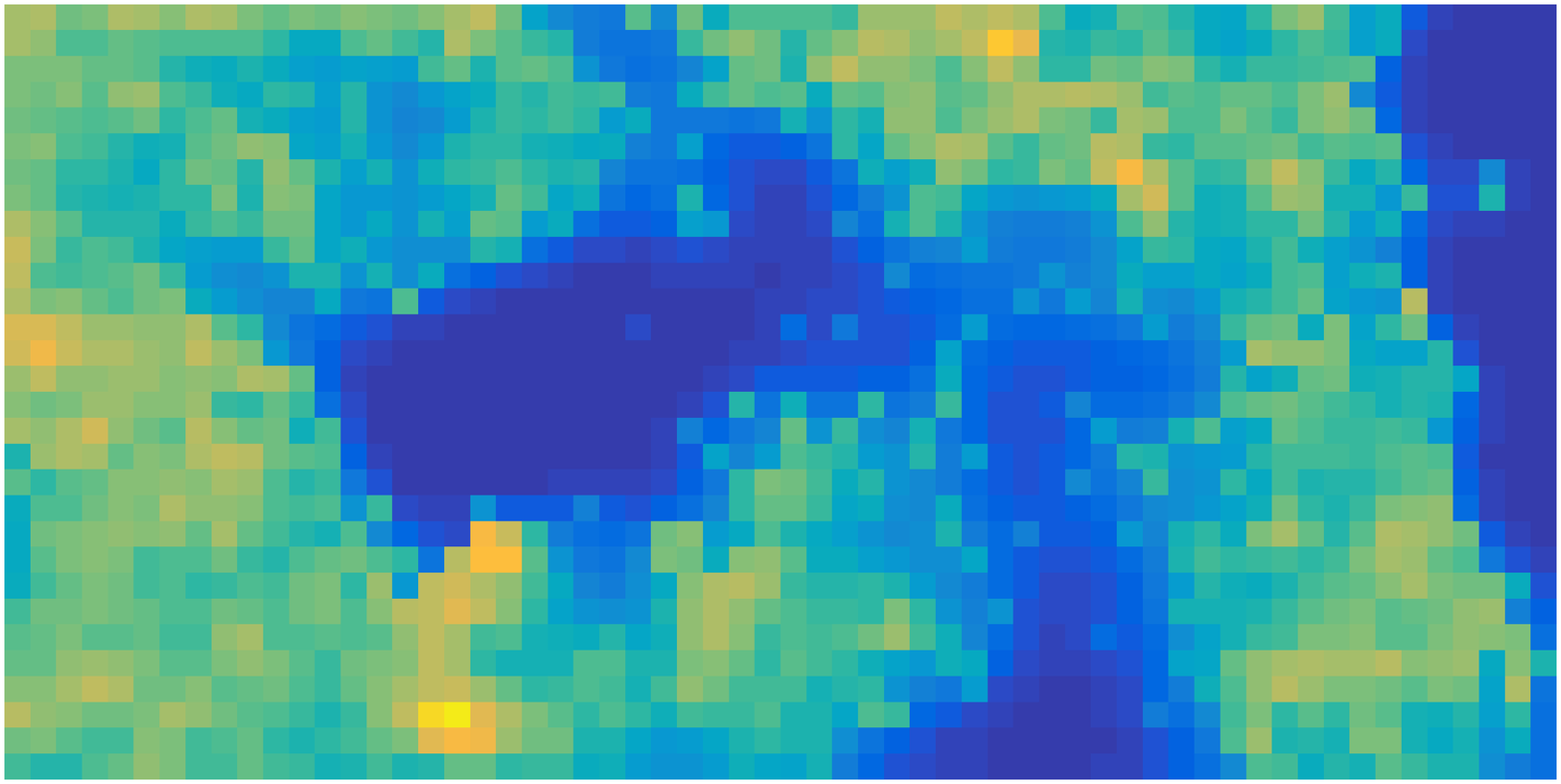}
		\caption{\ac{lscp}}
		\label{fig:lgcpmPostInts}
		\end{subfigure}		
		\begin{subfigure}{0.49\textwidth}
		\centering
		\includegraphics[width = 1 \textwidth, keepaspectratio]{./figs/lgcp_allCovsintsField.eps}
		\caption{LGCP}
		\label{fig:lgcpPostInts}		
		\end{subfigure}\\		
		\begin{subfigure}{0.49\textwidth}
		\centering
		\includegraphics[width = 1 \textwidth, keepaspectratio]{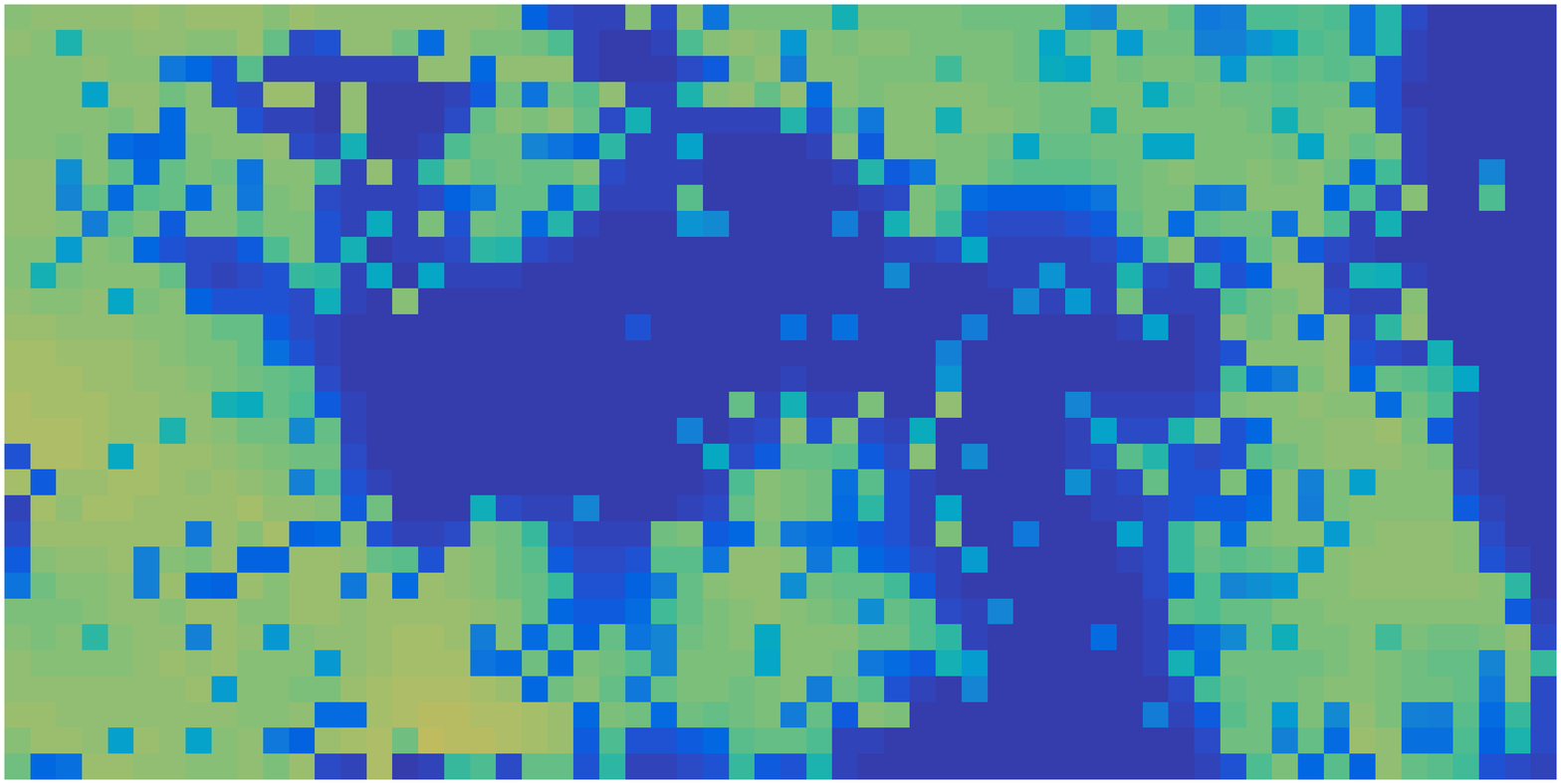}
		\caption{FixedM}
		\label{fig:fixedmPostInts}
		\end{subfigure}		
		\begin{subfigure}{0.49\textwidth}
		\centering
		\includegraphics[width = 1 \textwidth, keepaspectratio]{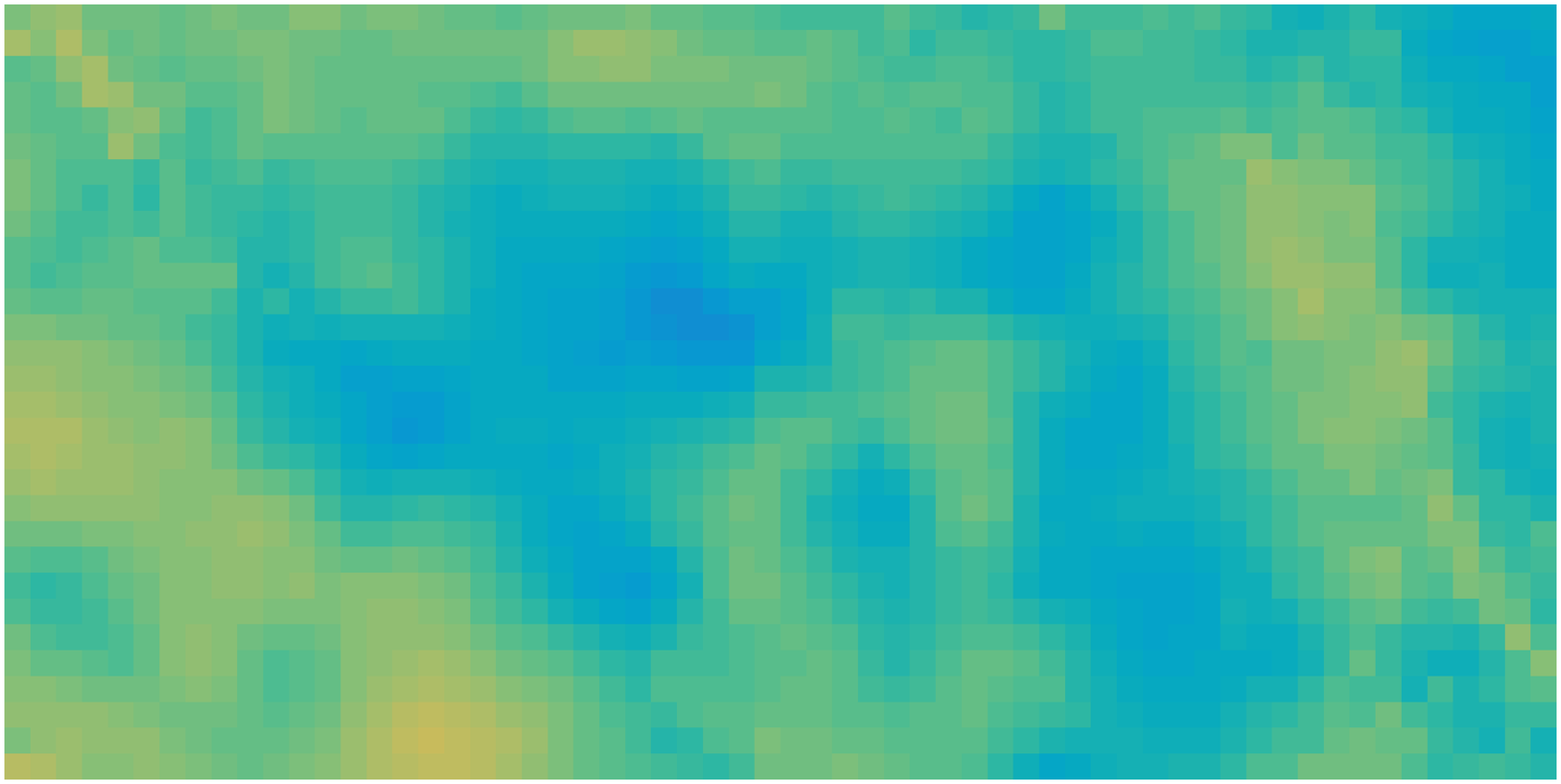}
		\caption{Fixed}
		\label{fig:fixedPostInts}
		\end{subfigure}		 
	\end{subfigure} 
	\begin{subfigure}{0.04\textwidth}
	\centering
	\includegraphics[width = 0.8 \textwidth, height = 9\textwidth]{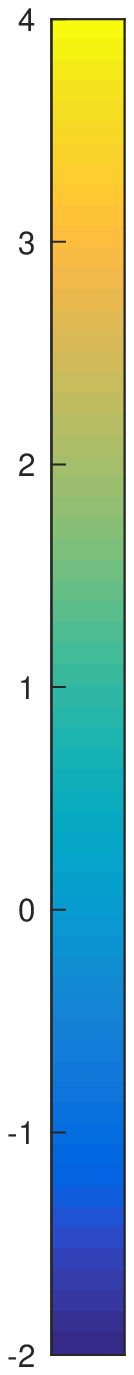}
	\vspace{0.5\textwidth}
	\end{subfigure} 
	\begin{subfigure}{0.3\textwidth}
	\centering
		\begin{subfigure}{0.95\textwidth}
		\centering
		\includegraphics[width = 1 \textwidth, keepaspectratio]{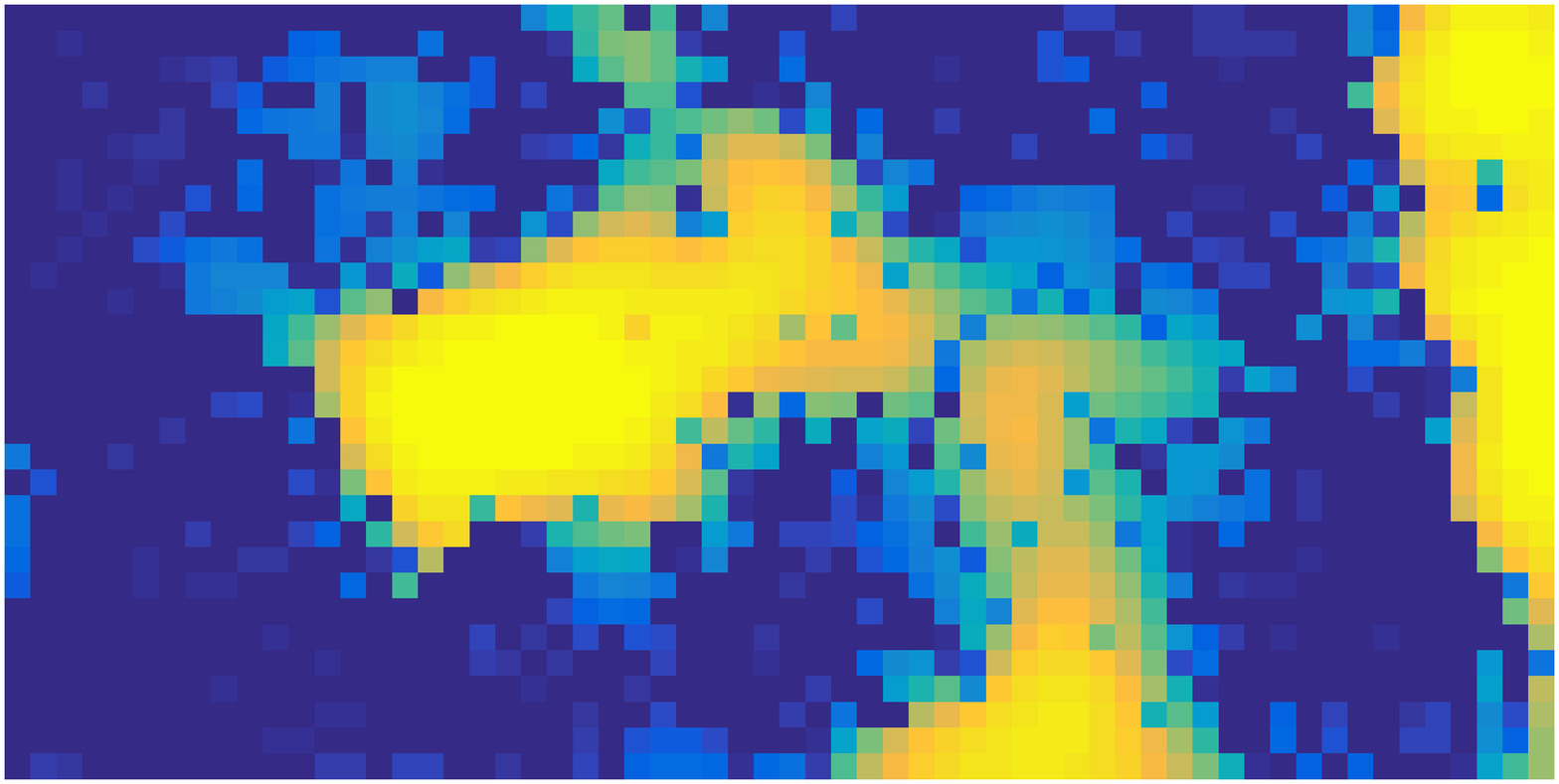}
		\caption{\ac{lscp}}
		\label{fig:lgcpmPostClass}
		\end{subfigure} \\
		\begin{subfigure}{0.95\textwidth}
		\centering
		\includegraphics[width = 1 \textwidth, keepaspectratio]{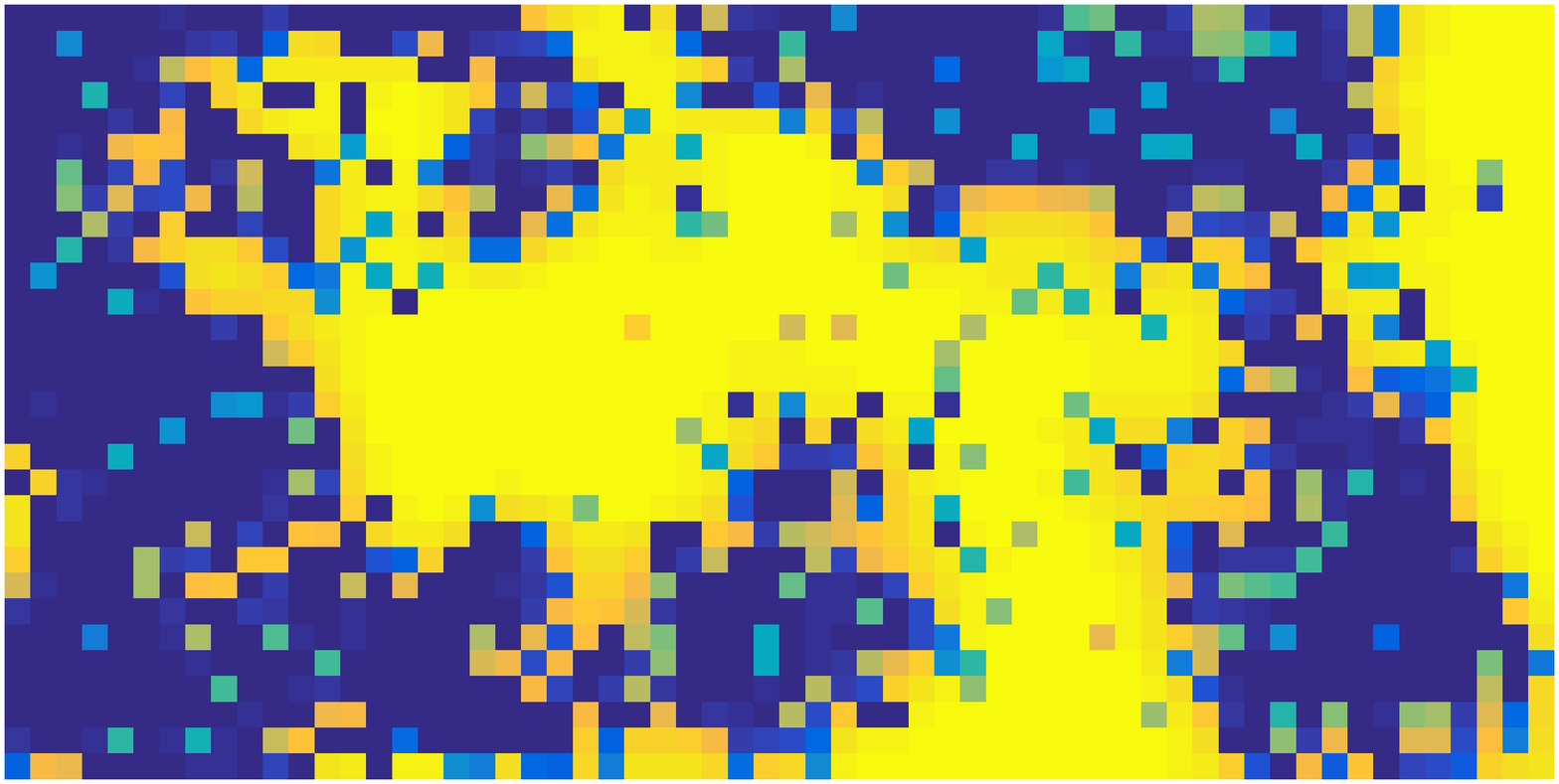}
		\caption{FixedM}
		\label{fig:fixedmPostClass}
		\end{subfigure} 
	\end{subfigure}
	\begin{subfigure}{0.04\textwidth}
	\centering
	\includegraphics[width = 1 \textwidth, height = 9\textwidth]{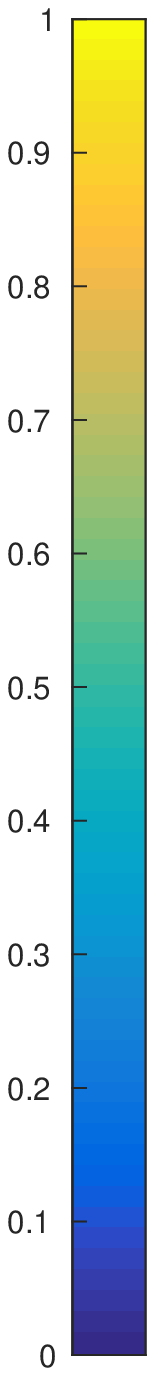}
	\vspace{0.0\textwidth}
	\end{subfigure} 	
	
	\caption{Mean posterior log intensity surface, $\ints$ (left) and mean classification for the models where applicable (right)}
	\label{fig:postInts}
\end{figure}

The relationship between the tree intensity and the covariates is also of interest, and in practice is often the focus of a study and hence the most relevant inference. Recall that 12 covariates of the original 16 covariates are considered her;  11 covariates and one intercept term. Figure \ref{fig:covcredintervals} shows the mean and $95\%$ credibility intervals for each of these covariates for all models.
The first question is which of the covariates have a significant impact on the spatial distribution of the trees and hence reflect a habitat preference of the species. 
To answer this we asses which of the regression coefficients $\beta$ are significantly different from zero. Empirical p-values are computed from the sampled posterior distributions and adjusted for the multiple testing scenario, Holm-Bonferroni correction \citep{lit:holm} is used to acquire rejection regions for each covariate. 
Table \ref{table:covariates} shows the covariates that were considered significant, at a significance level of $5\%$, for each of the four models.

The FixedM model identifies a smaller number of significant covariates than the Fixed model. This is not suprising since the covariates do not need to explain the lack of trees in the empty domain anymore.
For the LGCP and LSCP model, only the intercept is significant. This is probably an effect of the smaller number of degrees of freedom due to the increased number of parameters to estimate, i.e.\ the Gaussian fields.

\begin{figure}[t]
\centering

\begin{subfigure}{0.24\textwidth}
\centering
Intercept
\includegraphics[keepaspectratio, width = 0.95\textwidth]{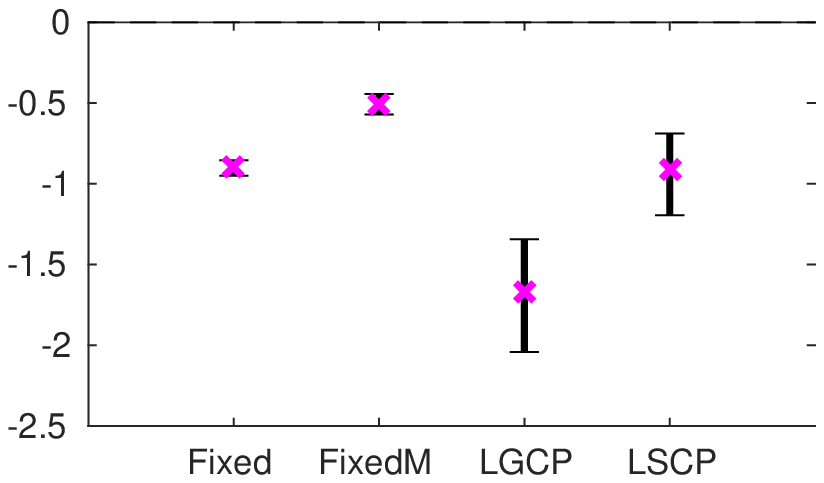}
\end{subfigure}
\begin{subfigure}{0.24\textwidth}
\centering
Elevation
\includegraphics[keepaspectratio, width = 0.95\textwidth]{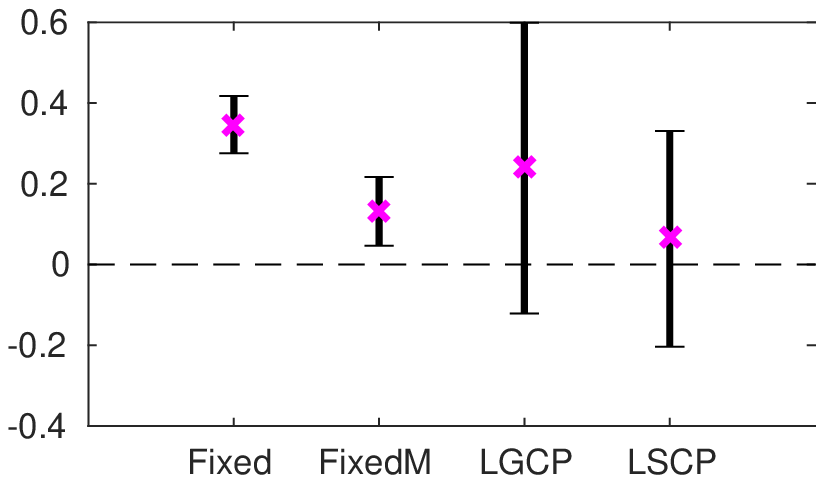}
\end{subfigure}
\begin{subfigure}{0.24\textwidth}
\centering
Slope
\includegraphics[keepaspectratio, width = 0.95\textwidth]{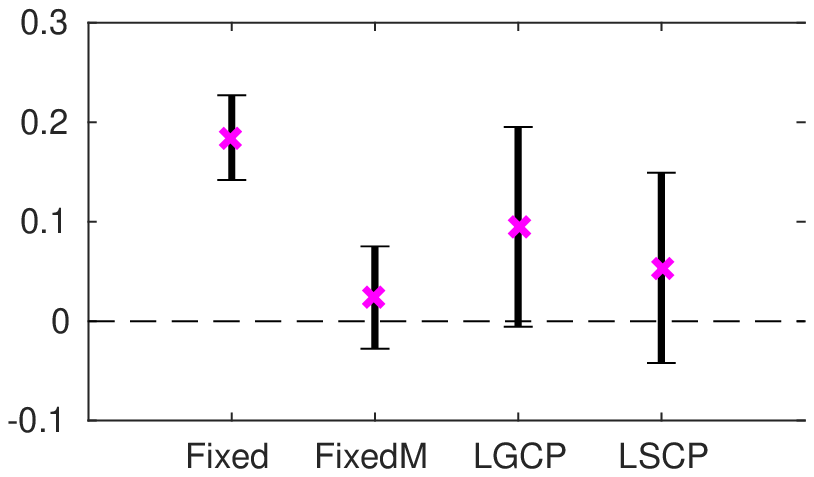}
\end{subfigure}
\begin{subfigure}{0.24\textwidth}
\centering
Al
\includegraphics[keepaspectratio, width = 0.95\textwidth]{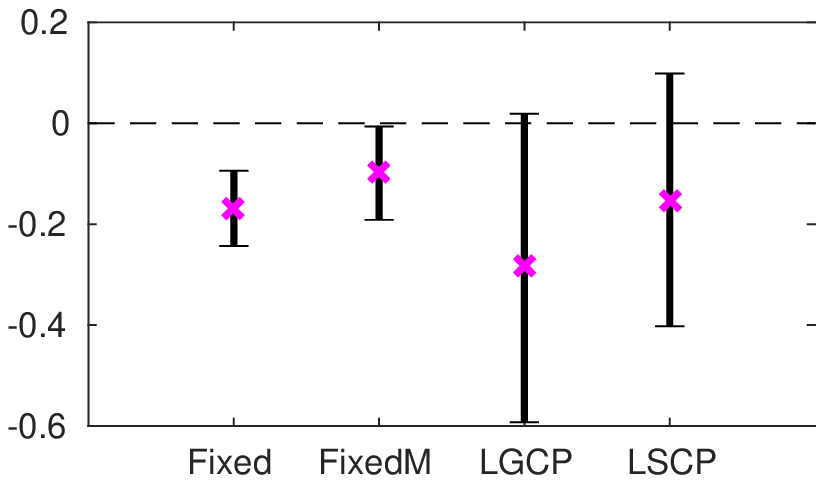}
\end{subfigure} \\
\begin{subfigure}{0.24\textwidth}
\centering
Cu
\includegraphics[keepaspectratio, width = 0.95\textwidth]{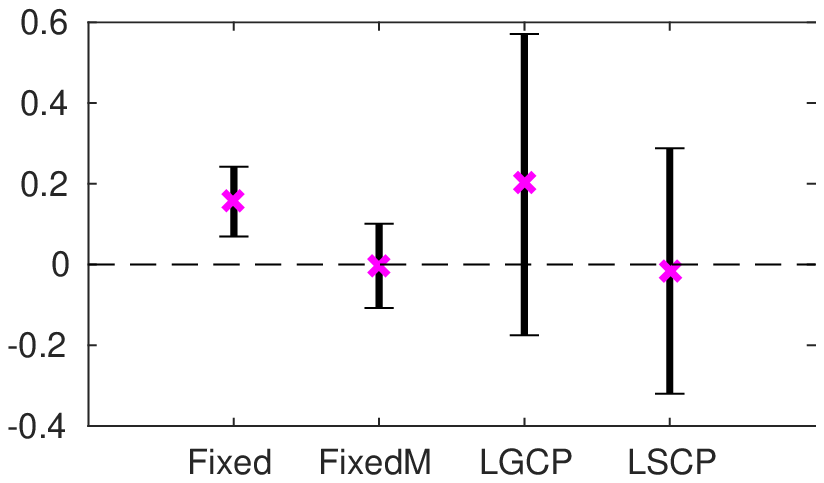}
\end{subfigure}
\begin{subfigure}{0.24\textwidth}
\centering
Fe
\includegraphics[keepaspectratio, width = 0.95\textwidth]{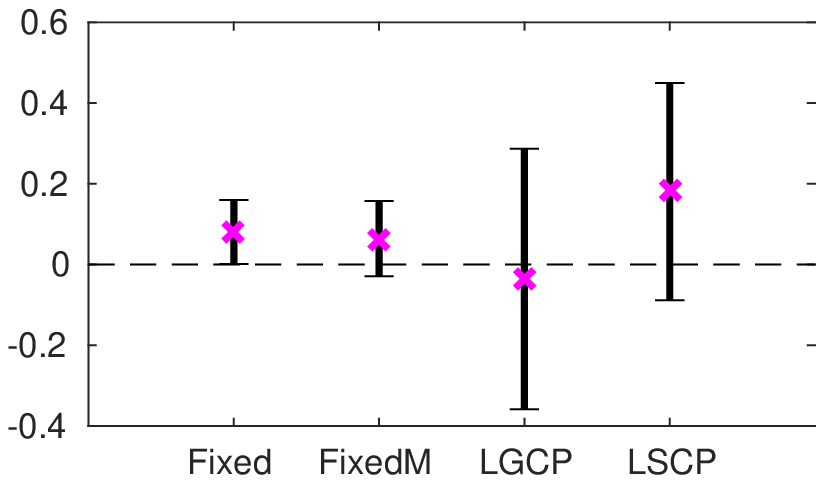}
\end{subfigure} 
\begin{subfigure}{0.24\textwidth}
\centering
Mg
\includegraphics[keepaspectratio, width = 0.95\textwidth]{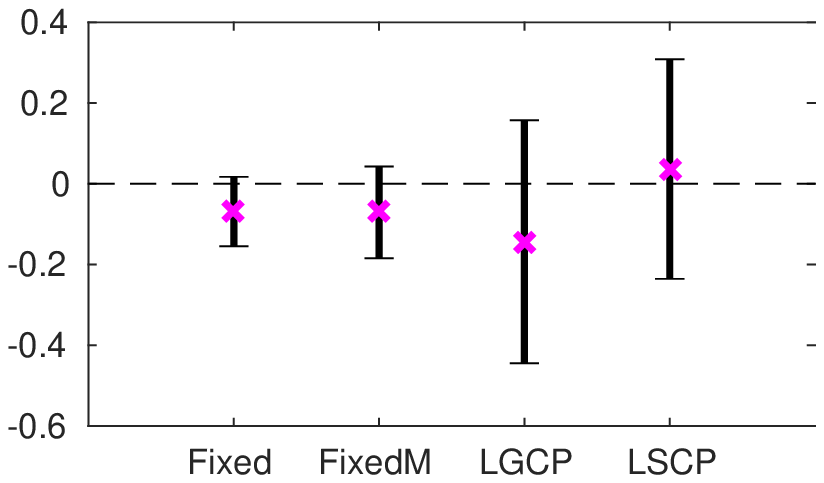}
\end{subfigure}
\begin{subfigure}{0.24\textwidth}
\centering
Mn
\includegraphics[keepaspectratio, width = 0.95\textwidth]{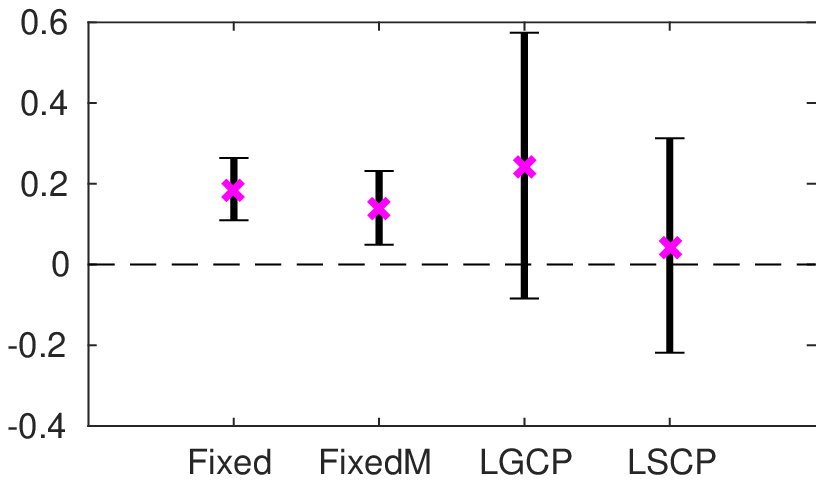}
\end{subfigure} \\

\begin{subfigure}{0.24\textwidth}
\centering
N
\includegraphics[keepaspectratio, width = 0.95\textwidth]{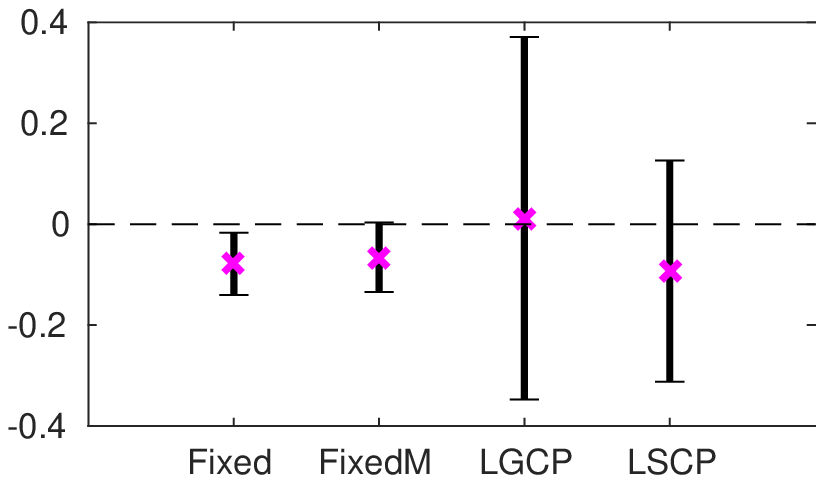}
\end{subfigure} 
\begin{subfigure}{0.24\textwidth}
\centering
Nmin
\includegraphics[keepaspectratio, width = 0.95\textwidth]{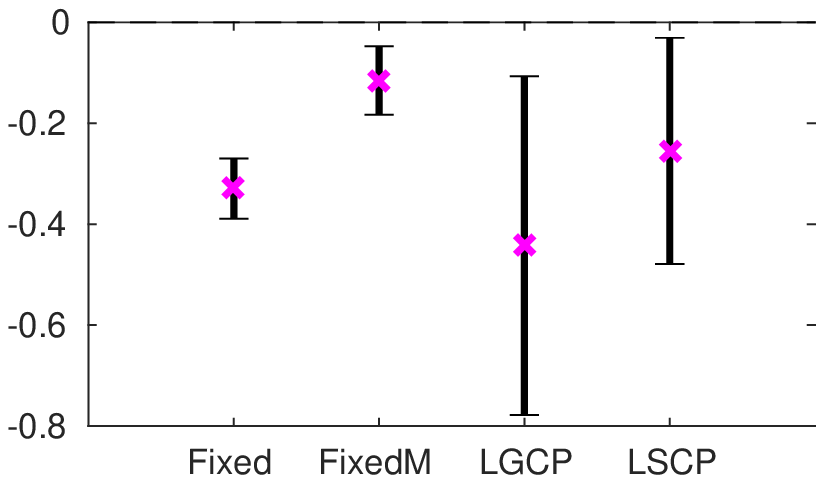}
\end{subfigure}
\begin{subfigure}{0.24\textwidth}
\centering
P
\includegraphics[keepaspectratio, width = 0.95\textwidth]{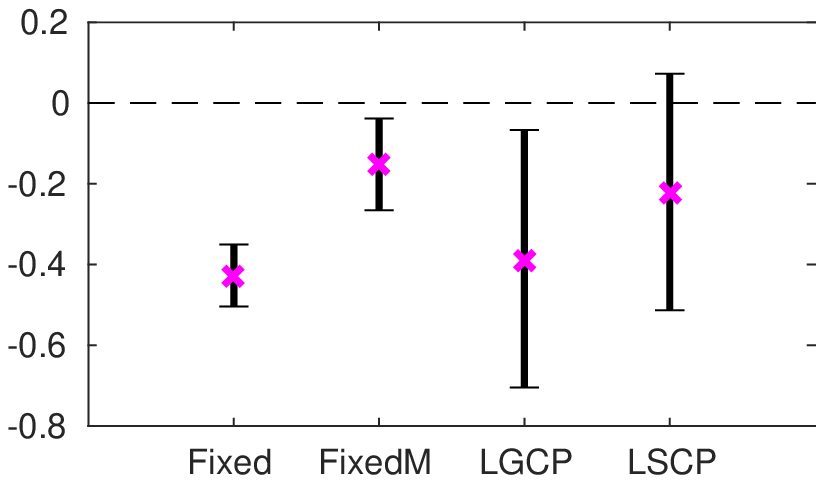}
\end{subfigure}
\begin{subfigure}{0.24\textwidth}
\centering
pH
\includegraphics[keepaspectratio, width = 0.95\textwidth]{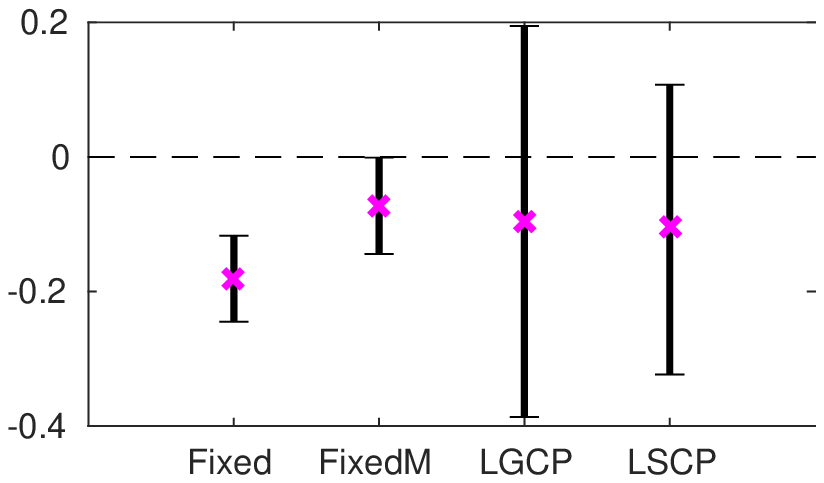}
\end{subfigure} \\

\caption{ The mean (cross) and 95\% credibility intervals (lines) for the posterior marginal distribution of fixed effect for the four different models.}
\label{fig:covcredintervals}
\end{figure}

\begin{table}[t]
\centering
\begin{tabular}{|l|l|}
\hline
Model & Covariates \\
\hline
Fixed & Int, Elev, Slope, Al, Mn, NMin, P, pH, Cu, N \\
FixedM & Int, NMin, Elev, Mn \\
LGCP & Int \\
\ac{lscp} & Int \\
\hline
\end{tabular}
\caption{Significant covariates on a 5\% level for the covariates using Holm-Bonferroni correction to correct for multiple hypothesis tests.}
\label{table:covariates}
\end{table}

% !TEX root = shell.tex
\section{Discussion}
\label{sec:discussion}
We have considered the problem of Bayesian level set inversion for point process data. The proposed model can be seen as a generalization of the log-Gaussian Cox process model where the latent Gaussian field is extended to a level set mixture of Gaussian fields. 
We derived basic model properties and in Appendix \ref{sec:theoretical} showed consistency of the posterior probability measure of finite-dimensional approximations to the continuous model. A computationally efficient MCMC method for Bayesian inference, based on the pCN MALA algorithm, was presented. A topic of further research could be to investigate other, potentially even quicker, estimation methods such as INLA or variational Bayes.

We modelled a point pattern formed by the locations of the trees from a species in a tropical rainforest. The example was of interest since the point pattern show clear signs of being affected by some unknown confounding factor.
Comparisons of functional statistics between simulations from the fitted models and the observed data indicated that allowing for a second class in the model better explains the point pattern behavior. Moreover, the LSCP model stayed close to the expected values for all three functional characteristics investigated while the popular LGCP model did not. 
There are indications that the FixedM model explains the data better than the LGCP model despite the much simpler structure of the earlier model. FixedM has far less degrees of freedom than the LSCP and LGCP models, and is hence less prone to overfitting.  It also shows that it is not overfitting that allows models with two classes to outperform the LGCP. 
The analysis of the tropical rainforest showed that inference on both the Gaussian field parameters and covariates were affected by allowing for a second class in the model. It suggests that the inference drawn based on the LGCP model were biased by the confounding factor. 

Future analysis could consider using fixed effects also in the level set field, $\latf_0$, in order to investigate which covariates that explains the classification. This is another feature of the proposed model that we have not yet investigated.
Further, analysis of multivariate point patterns are possible such as for instance joint analysis of several species of plants. 
This could be performed by introducing multivariate Gaussian random fields for the classes, i.e. for $\{\latf_k\}_{k=1}^K$. 
Another possibility is letting several species share the same level set field, $\latf_0$, or classifications field, $\clf$, but use independent class fields, $\{\latf_k\}_{k=1}^K$. 
In this way, information about $\latf_0$ could be enforced from several point patterns jointly.

\section{Acknowledgements}
\label{sec:acknowledgements}
The authors gratefully acknowledge the financial support from the Knut and Alice Wallenberg Foundation, the Swedish Research Council Grant 2016-04187, and the \r{A}Forsk foundation. 
We would like to thank the people at the Center of tropical forest research, Smithsonian Tropical Research Institute for the extensive forest census plot and for making the data publicly available.
The BCI forest dynamics research project was founded by S.P. Hubbell and R.B. Foster and is now managed by R. Condit, S. Lao, and R. Perez under the Center for Tropical Forest Science and the Smithsonian Tropical Research in Panama. Numerous organizations have provided funding, principally the U.S.\ National Science Foundation, and hundreds of field workers have contributed.

Also thanks to the Barro Colorado soil survey (Jim Dalling, Robert John, Kyle Harms, Robert Stallard and Joe Yavitt and field assistants Paolo Segre and Juan Di Trani) for making the soil sample data publicly available and for answering questions and handing out the original soil sample locations on request.
The Barro Colorado soil survey was funded by NSF DEB021104,021115, 0212284,0212818 and OISE 0314581 as well as the STRI Soils Initiative and CTFS. 

\bibliographystyle{plainnat}
\bibliography{shell}

\appendix

% !TEX root = shell.tex
\section{Theoretical results}\label{sec:theoretical}
In this section, we will theoretically justify the two approximations of the LSCP process that are needed for inference. The first is the finite dimensional approximation from Section \ref{sec:findim} and the second is the truncation needed for the fast Fourier transform in Section \ref{sec:estimation}.

For $k = \{0, ..., K\}$, let $\latf_k$ be a Gaussian random field on the spatial domain $\domSp = [0,1]^d \subset \R^d$, defined on a complete probability space. We will show the results using methods similar to those in \citep{lit:cotter2,lit:iglesias,lit:simpson} and for this it is convenient to represent the fields as Gaussian measures $\mu_0^{(k)}$. To simplify the presentation, we will assume a specific covariance operator related to the Mat\'ern covariance function. However, the results can be extended to more general densely-defined, self-adjoint, positive definite operators and to more general bounded domains. 

Let $\mu_0^{(k)} = \mathcal{N}(0,\mathcal{C})$, where $\mathcal{C} = \tau^2 A^{-\alpha}$ with $A=\kappa^2 - \Delta$. Here $\tau,\kappa^2$ and $\alpha$ are positive parameters and $A : \scrD(A) \subset L_2(\domSp) \rightarrow L_2(\domSp)$, further we impose periodic boundary conditions. Denote the eigenvalues of $A$ as  $\{\lambda_{j}\}_{j\in\mathbb{N}}$ , which are arranged in a nondecreasing order, and the corresponding eigenfunctions as $\{e_j\}_{j\in\mathbb{N}}$, which form a complete orthonormal basis for $L_2(\domSp)$.  The fractional power operator $A^{\alpha}: \scrD(A^{\alpha})\rightarrow L_2(\domSp)$ is defined by 
$$
A^{\alpha}u = \sum_{j\in \mathbb{N}}\lambda_j^{\alpha}\scal{u}{e_j}e_j.
$$
For any $\alpha$, the subspace $\mathcal{H}^{\alpha} := \scrD(A^{\alpha/2})$ is a Hilbert space
$$
\mathcal{H}^{\alpha} = \{u : \sum_{j\in \mathbb{N}}\lambda_j^{\alpha}|\scal{u}{e_j}|^2 < \infty\},
$$
with respect to the inner product $\scal{\phi}{\psi}_{\alpha} = \scal{A^{\alpha/2}\phi}{A^{\alpha/2}\psi}$ and corresponding norm $\|\phi\|_{\alpha} = \sum_{j\in\mathbb{N}}\lambda_j^{\alpha}\scal{\phi}{e_j}^2$. 

With this choice of covariance operator, we have that if $u\sim \mu_0^{k}$, then $u\in \mathcal{H}^{s}$ for any $s < \alpha - d/2$ $\mu_0^{k}$-almost surely \citep[Theorem 1]{lit:dunlop}. Furthermore, $u$ is almost surely p-times differentiable if $\alpha - d/2 > p$. We will need this differentiability and we formulate it as an assumption.

\begin{assum}\label{ass:1}
The classification field $\latf_0$ is almost surely a Morse function with strictly positive variance at all locations in the domain, and for $k>0$ the Gaussian fields $\latf_k$ are almost surely differentiable. 
\end{assum}
The differentiability assumption is satisfied by assuming $\alpha>2$. The Morse function requirement is slightly stronger than $C^2$, but is implied by $\alpha>4$ \citep{lit:adler}. 
Furthermore, we can use a theorem equivalent to the Sobolev embedding theorem for our $\mathcal{H}^s$ space \citep[Theorem 2.10]{lit:stuart}. That is, $\|X_k\|_{L^\infty} \le C \|X_k\|_{s}$ if $X_k \in \mathcal{H}^s$ and $s > d/2$. For our case with periodic boundary conditions the space $\mathcal{H}^s$ is even equivalent to the Sobolev space $H^s$.

We thus have that $\latf_k$ is represented as a Gaussian measure, $\mu_0^{(k)}$, on $\mathcal{H}^{\alpha}$ and we can choose an appropriate $\sigma$-algebra such as the probability space $(\mathcal{H}^\alpha, \Sigma_k, \mu_0^{(k)})$ becomes complete (see \citep{lit:iglesias}). Likewise $\latf = \{\latf\}_{k=0}^K$ can be represented by a product measure $\mu_0$ on the complete measure space $\xspace = (\Omega, \Sigma, \mu_0)$, where $\Omega$ is the product space of each $\mathcal{H}^{\alpha}$ and $\Sigma$  is the corresponding product $\sigma$-algebra.

%The $\mathcal{H}^1$-norm bounds the $H^1$-norm for $u \in \mathcal{D}(A)$ given periodic boundary conditions since $\langle \nabla u, \nabla u \rangle_{L^2} = \langle \nabla u \cdot \bs{n}, u \rangle_{L^2(\partial \domSp)} + \langle A u , u \rangle_{L^2(\partial \domSp)} - \kappa^2 \|  u \|_{L^2}$.

Since the \ac{lscp} model defines the point process as a non-homogeneous Poisson process conditioned on $\latf$, the likelihood potentials for the continuous and finite dimensional models, defined in Section \ref{sec:findim}, are
\begin{align}
\Phi(\latf; \obsf) &= \int_{\domSp} \ints(s; \latf)ds - \sum_{\psp_j \in \obsf} \log \ints(\psp_j;\latf), \label{eq:phi} \\
\Phi^N(\latf; \obsf) &= \sum_{i \in N} \left( |\domSp_i| \ints(\tilde{\psp}_i;\latf) - \obsf_i  \log \ints(\tilde{\psp}_i;\latf)   \right).
\label{eq:phiN}
\end{align}
Here, $N$ is the number of discretized regions in the lattice approximation and $Y_i$ denotes the number of observations in $\domSp_i$. Further, $\tilde{\psp}_i$ is the midpoint of each $\domSp_i$, and $\psp_j$ is the location of the $j$th point in the point pattern $\obsf$. Based on these likelihoods, we can now define the corresponding posterior measures as follows. 

\begin{prop}\label{prop:posteriors}
If Assumption \ref{ass:1} holds, we can define posterior measures using Radon-Nikodym derivative with respect to $\mu_0$:
\begin{equation}
\begin{split}\label{eq:postmu}
\frac{d\mu}{d\mu_0}(\latf) &= \frac{1}{C_{\mu}(\obsf)}\exp \left( -\likPot(\latf; \obsf) \right), \\
\frac{d\mu^N}{d\mu_0}(\latf) &= \frac{1}{C_{\mu^N}(\obsf)}\exp \left( -\likPot^N(\latf ; \obsf) \right), 
\end{split}
\end{equation}
where $C_{\mu}(\obsf)$ and $C_{\mu^N}(\obsf)$ are normalizing constants.		
\end{prop}
The proof is given in Appendix \ref{sec:proofs}. Since only the discretized model can be used for inference, it is important to know that the approximation $\mu^N$ converges to the true posterior, $\mu$, as the discretization becomes finer. The following theorem shows that this indeed is the case with respect to the total variation distance, $d_{\text{TV}}(\mu, \mu^N) = 2\sup_{E \in \mathcal{F}_{\xspace}} |\mu(E) - \mu^N(E)|$.

\begin{thm}\label{theorem:postconvergence}
Let Assumption \ref{ass:1} hold and let $\mu^N$ and $\mu$ be the posterior measures defined in \eqref{eq:postmu}. Then $d_{\text{TV}}(\mu, \mu^N) \to 0$ as $N\to \infty$.
\end{thm}
The proof is given in Appendix \ref{sec:proofs}. Also the latent fields, $\latf$, need to be approximated by finite dimensional representations for inference. We will do this by truncating the basis expansion of the field to $p$ terms:
$$
\latf \approx \tilde{X} = \sum_{j=1}^p\xi_j\lambda_j^{\alpha}e_j,
$$
where $\xi_j$ are independent standard normal variables. We will refer to the model using a discretization of the observational domain and finite dimensional approximations of $\latf$ as the \textit{fully discretized model}. The advantage with using this truncation is that we can use the fast Fourier transform for simulating the field. To show that we still have convergence under this approximations, note that the finite dimensional approximation of $\latf$ can be viewed as an orthogonal projection of $\latf$ on to the space spanned by the eigenfunctions $\{e_j\}_{j \le p}$ as is done in \citet{lit:cotter2}. We define the projection operator $P^p$ such that $\tilde{\latf}(\psp) = P^p \latf(\psp)$. It is now possible to define a posterior probability measure for $\tilde{\mu}^N$ by it's Radon-Nikodym derivative as
\begin{align}
\frac{d\tilde{\mu}^N}{d\mu_0}(\latf) = \frac{1}{C_{\tilde{\mu}^N}(\obsf)}\exp \left( -\likPot^N(P^p \latf; \obsf) \right) \label{eq:postmuNp}.
\end{align}
An important consequence of this definition is that the posterior measure is absolutely continuous with respect to $\mu_0$ and measurable with respect to $\Sigma$. The interpretation of $\tilde{\mu}^N$ is that the data will only affect the projection, $P^p\latf$. We can now show that also under this approximation, we get convergence to the true posterior.

\begin{thm}\label{theorem:postTruncConvergence}
	 Let the measure $\tilde{\mu}^N$ be defined by \eqref{eq:postmuNp}, and let the measure $\mu$ be defined by \eqref{eq:postmu}. If $\mu_0$ satisfies Assumption \ref{ass:1}, then $d_{\text{TV}}(\mu, \tilde{\mu}^N) \to 0$ as $N\to \infty$ and $p\to \infty$. 
\end{thm}
The proof is given in Appendix \ref{sec:proofs}.

% !TEX root = shell.tex
\section{Proofs}
\label{sec:proofs}
\begin{proof}[Proof of Proposition \ref{prop:intensmoments}]
For the first moment, note that
\begin{align}
\expect{ \ints(\psp) } &= \mathbb{E}\left[\exp(\latf(\psp)) \right]  = \sum_{k=1}^K   \mathbb{E}\left[\exp(\latf(\psp)) | \latf_0(\psp) \in (c_{k-1},c_k] \right] \prob{ \latf_0(\psp) \in (c_{k-1},c_k]}  \\
&= \sum_{k=1}^K   \mathbb{E}\left[\exp(\latf_k(\psp) + \mu_k(\psp))  \right] \prob{ \latf_0(\psp) \in (c_{k-1},c_k]} \\
&=\sum_{k=1}^K  \exp\left(\mu_k(\psp) + \frac{r_k(0)}{2}\right) \prob{ \latf_0(\psp) \in (c_{k-1},c_k]},
\end{align}
where the final equality follows from the explicit form of the expectation of a log normal random variable. The second moment follows by similar calculations.   
\end{proof}

\begin{proof}[Proof of Proposition \ref{prop:emptyspace}]
The inhomogeneous empty space function, $F(\psp_0, r)$ is defined as the probability of having at least one point inside a ball of radius $r$ centered at $\psp_0$, i.e. $F(\psp_0, r) = \prob{N(\obsf; B(\psp_0, r)) > 0}$. Here, $N(\obsf; A)$ is the number of points inside the domain $A$ for a realization of the point process, $\obsf$. Hence $F(\psp_0, r) = 1 - \prob{N(\obsf; B(\psp_0, r)) = 0}$. Now,
\begin{align}
\prob{N(B(\obsf; \psp_0, r)) = 0} &= \expect{ \exp\left( -\int_{B(\psp_0, r)} e^{\sum_{k=1}^K \clf_k(\psp) \left(\latf_k(\psp) + \mu_k(\psp)\right)} d\psp \right) } \\
&= \expect{ \exp\left( -\int_{B(\psp_0, r)} \sum_{k=1}^K \clf_k(\psp) e^{ \latf_k(\psp) + \mu_k(\psp)} d\psp \right) }\\
&= \expect{ \prod_{k=1}^K \exp\left( -  \int_{\domSp_k \cap B(\psp_0,r)}  e^{ \mu_k(\psp)} e^{ \latf_k(\psp)} d\psp \right) }.
\end{align}
\end{proof}

Due to the product space interpretation of $\latf$ as the collection $\{\latf_k\}_k$, we define norms on $\latf$ as $\|\latf\|_{(\cdot)} = \sum_{k=0}^K \|\latf_k\|_{(\cdot)}$.
That is, a norm on realizations of all Gaussian random fields jointly are defined as the sum of the norm for each of the $K+1$ fields. 
%By this definition, $\|\latf\|_{(\cdot)}$ is a norm. For example, $\|\latf\|_{1} = \sum_{k=0}^K \|\latf_k\|_{1}$.

To simplify the proofs we note that the potential $\Phi$ can be written as a composition of two functions: The potential
$\Phi(\latf; \obsf) = \Phi_P(G(\latf);\obsf)$ where $\Phi_P: L_2(\domSp) \times \yspace \rightarrow \mathbb{R}$ is the continuous Poisson log-likelihood function and
$G: \mathcal{H}^{\alpha} \rightarrow L_2(\domSp)$  is 
$$G(X ) = \sum_{k=1}^K \mixProb_k(\cdot) X_k(\cdot) = \log(\lambda(\cdot)), $$
where $\mixProb_k$ is the classification function, $\mixProb_k(\psp) = \indicator{\threshParam_{k-1} \le \latf_0(\psp) < \threshParam_k}$.
Similarly $\Phi^N(\latf;\obsf) = \Phi^N_P(G(\latf);\obsf)$ where $\Phi^N_P$ is the Poisson log-likelihood function for the discretized domain. 

To prove Proposition \ref{prop:posteriors}, we will need two lemmas, where the first gives bounds for the likelihood potentials.
\begin{lemma} \label{lemma:assHolds}
Let $\|\obsf\|_{\yspace}$ denote the number of points in a given point pattern. For $\Phi$ in \eqref{eq:phi} and $\Phi^N$ in \eqref{eq:phiN} we then have that:
\begin{enumerate}
	\item[(i)] For every $r>0$, $\epsilon > 0$, and $s > 1$ with $\latf\in \mathcal{H}^s$ and $\obsf \in \yspace$ with $||\obsf||_{\yspace} \leq r$, there exists a constant $M(\epsilon, r) \in \R$ such that $\Phi(\latf;\obsf) \geq M( \epsilon, r) - \epsilon ||\latf ||^2_{s}$. 
	\item[(ii)] For every $r>0$, and $s > 1$ all $\latf\in \mathcal{H}^s$ and all $\obsf\in \yspace$ with $\max\{||\latf||_{s},||\obsf||_{\yspace}\}<r$ we have $\Phi(\latf; \obsf) \leq |\domSp|e^{Cr} + C^2r^2$.
\end{enumerate}	
\end{lemma}

\begin{proof} \label{proof:assHolds}
	To show  (i) note that
	\begin{align}
	\Phi_P(G(\latf);\obsf) &= \int_{\domSp} \exp\left(G(\latf) \right)ds - \sum_{s_j \in \obsf}  G(\latf) \ge  - \sum_{s_j \in \obsf} G(\latf) 
	\ge - \|\obsf\|_{\yspace}  \|G(\latf)\|_{L^{\infty}(\domSp)}  	\\
	&\ge - r \|G(\latf)\|_{L^{\infty}(\domSp)}.
	\end{align}
	By Assumption \ref{ass:1} and the Sobolev embedding theorem we have that $\|\latf\|_{L^\infty(\domSp)} \le C \|\latf\|_{s}$. Thus $\|G(\latf)\|_{L^{\infty}(\domSp)}  \le \|\latf\|_{L^\infty(\domSp)} \le C \|\latf\|_{s} $ and we have $\Phi_P(G(\latf),\obsf) \ge -rC \|\latf\|_{s}$.
	Now, $0 \le (\frac{Cr}{2\sqrt{\epsilon}}- \sqrt{\epsilon} \|\latf\|_{s})^2 = \frac{C^2r^2}{4\epsilon} + \epsilon \|\latf\|^2_{s} - Cr \|\latf\|_{s}$. Hence
	$$Cr\|\latf\|_{s} \le \epsilon \|\latf\|^2_{s} + \frac{C^2r^2}{4\epsilon} = \epsilon \|\latf\|^2_{s} - M(\epsilon, r). $$		
	By the same argument, 
	\begin{align*}
	\Phi_P^N(G(\latf);\obsf) &= \sum_{i \in I^N} \left( |\domSp_i| \exp\left(G(\latf)(s) \right) - \obsf_i G(\latf)(s_i)   \right) 
	\ge M(\epsilon, r) - \epsilon\|\latf\|^2_{s}.
	\end{align*}

	Statement (ii) holds for $\Phi$ since
	\begin{align}
	\Phi_P(G(\latf);\obsf) &= \int_{\domSp} \exp\left(G(\latf)(\psp) \right)d\psp - \sum_{\psp_j \in \obsf}  G(\latf)(\psp_j)\\
	&\le |\domSp|e^{ \|G(\latf)\|_{L^{\infty}(\domSp)}} + \|\obsf\|_{\yspace} \|G(\latf)\|_{L^{\infty}(\domSp)}  \\
	&\le |\domSp|e^{Cr} + Cr^2 \le |\domSp|e^{Cr} + C^2r^2,
	\end{align}
	and the same for $\Phi^N$ since
	\begin{align}
	\Phi_P^N(G(\latf),\obsf) &= \sum_{i \in I^N} \left( |\domSp_i| e^{\|G(\latf)\|_{s}}  - \obsf_i G(\latf)(s_i)   \right) 
	\le |\domSp|e^{Cr} + C^2r^2.
	\end{align}	
\end{proof}

The second lemma we need concerns the regularity of the level sets of $\latf_0$. Let $S^0_k(\latf_0) = \{\psp: \latf_0(\psp) = \threshParam_k\}$ be the level set of $\latf_0$ for the level $c_k$ and set $S^0(\latf_0) = \cup_{k=1}^K S^0_k(\latf_0)$. Further, let $J_s$ denote the set of indices for all subregions $\domSp_j$ that do not intersect with $S^0(\latf_0)$, that is, $j \in J_s$ if $\domSp_j \cap \domSp_k^0 = \emptyset$ for all $1 \leq k \leq K$, and define $S(\latf) =  \cup_{j\in J_s}\domSp_j$ as the set of all subregions where the level sets are not included. We then have the following result about $S(\latf)$, and $\mathcal{L}_d(S(\latf))$where $\mathcal{L}_d$ denotes the Lebesgue measure in dimension $d$.
\begin{lemma}
\label{lemma:finitecurve}
Let Assumption \ref{ass:1} hold, then
\begin{itemize}
\item $\mathcal{L}_2(S^0(\latf)) = 0$ a.s.
\item $\expect{\mathcal{L}_2(S^C(\latf))} \rightarrow 0$ as $N\rightarrow \infty$.
\item For any finite set of points $\obsf$, $\expect{ \|S^C(\latf) \cap \obsf\|_{\yspace} } \rightarrow 0$ as $N\rightarrow \infty$.
\end{itemize}
\end{lemma}
\begin{proof}
That $\mathcal{L}_2(S^0(\latf)) = 0$ a.s. follows from Proposition 2.8 in \citet{lit:iglesias}. 

We will now show that $\expect{\mathcal{L}_2( S^C(\latf ))}$ goes to zero. Note that a curve segment of length $l$ can at most cover 4($\frac{l}{h} + 1$) subregions $\domSp_j$. Hence, the number of subregions $\domSp_j$ that have a level crossing, $N-|J_s|$, is bounded by $\sum_{i=1}^{N^*}4(l_i/h+1)$, where $N^*$ is the number of disjoint line segments in $S^0(X_0)$ and $l_i$ the length of $i$th segment. This gives that
 $$
 \expect{\mathcal{L}_2(S^C(\latf))} \leq h^2\left(\frac{4}{h}\expect{\mathcal{L}_1(S^0(X_0))} + 4\expect{N^*}\right) \leq 4h(\expect{\mathcal{L}_1(S^0(X_0))} + hN^*).
 $$
Thus, the result follows if we can bound $\expect{\mathcal{L}(S^0(X_0)))}$ and $\expect{N^*}$.
By assumption $\latf_0$ satisfies the conditions of Rice Theorem \citep{lit:azais}, which gives that $\expect{\mathcal{L}_1(S^0(X_0)))}  < \infty$.  Let $N_k$ denote the number of local maxima of $X_0$ over the level $c_k$ and let $N^0 = \sum_k N_k$. Since $\expect{N^*}$ is bounded by $\expect{N^0}$, and Rice Theorem bounds $\expect{N^0}$, the result follows.

Finally, we show that $\expect{ \|S^C(\latf) \cap \obsf\|_{\yspace} }$ goes to zero. We only consider the case $K=1$ and  $Y=\{y\}$, as the general result follows directly given that the claim holds for this special case.  Let $B(y,\epsilon_N)$ be a ball centered at $y$, where $\epsilon_N$ is chosen so that the subregions covering $y$ are contained in the ball. To prove the result we need to show that $\mathbb{P}(\mathcal{L}_i(B(y,\epsilon_N)\cap  X^{-1}_0(c_1)) >0) \rightarrow 0$ as $N\rightarrow 0$, for both $i=0,1$, where $\mathcal{L}_i$ are the Lipschitz-Killing curvatures.  Since $X_0$ is a Morse function and $\mathcal{L}_i(B(y,\epsilon_N))\rightarrow 0$, Theorem 15.9.4 in \cite{lit:adler} shows that $\mathbb{E}[\mathcal{L}_i(B(y,\epsilon_N)\cap  X^{-1}_0(c_1)))] \rightarrow 0$ for $i=0,1$. Thus $ \mathbb{P}(\mathcal{L}_1(B(y,\epsilon_N)\cap  X^{-1}_0(c_1))>0) \rightarrow 0$ as $\mathcal{L}_1(B(y,\epsilon_N)\cap  X^{-1}_0(c_1))$ is non-negative random variable. Since any $B(y,\epsilon_N)$ converges to a point, it follows that $\mathcal{L}_0(B(y,\epsilon_N)\cap  X^{-1}_0(c_1))$ (the Euler characteristic) converges to a non-negative random variable, and thus $\mathbb{P}(\mathcal{L}_0(B(y,\epsilon_N)\cap  X^{-1}_0(c_1)) >0 ) \rightarrow 0$.

\end{proof}

\begin{proof}[Proof of Proposition \ref{prop:posteriors}] \label{proof:posteriors}
We only state the proof for $\mu$ since the proof for $\mu^N$ follows similarly. To show the result we must show that the $\Phi$ is a measurable function, and then that the measure is normalizable. To prove measurability it suffices, by Lemma 6.1 in \citet{lit:iglesias}, to show that that $\likPot$ is continuous $\mu_0$-almost surely. Thus for $\hat{X},\tilde{X}\in\mathcal{H}^s, s > 1$, we must show that $|\likPot(\hat{\latf}) - \likPot(\tilde{\latf})| \rightarrow 0$ as $\|\hat{\latf} - \tilde{\latf}\|_{s} \rightarrow 0$. Note that
\begin{align}\label{eq:proofineq10}
|\likPot(\hat{\latf}) - \likPot(\tilde{\latf})| \le \int_{\domSp} |e^{G(\hat{\latf})(\psp)} - e^{G(\tilde{\latf})(\psp)}| d\psp + \sum_{\psp_j \in \obsf}| G(\hat{\latf})(\psp_j) - G(\tilde{\latf})(\psp_j)|.
\end{align}
We show continuity of the two terms separately. For the first term in \eqref{eq:proofineq10} it follows that
\begin{align}
\int_{\domSp} |e^{G(\hat{\latf})(\psp)} - e^{G(\tilde{\latf})(\psp)}| d\psp &\le \int_{\domSp}  \exp \left(|G(\hat{\latf})(\psp)| + |G(\tilde{\latf})(\psp)|\right)|G(\hat{\latf})(\psp) - G(\tilde{\latf})(\psp)| d\psp \\
&\le C|\domSp| e^{C\|\hat{\latf}\|_{s} + C\|\tilde{\latf}\|_{s}}\|G(\hat{\latf}) - G(\tilde{\latf})\|_{s}.
\end{align}		
Here the first inequality is due to the mean value theorem, and the second inequality comes from using Sobolev's embedding theorem, and H\"{o}lders inequality.
Since $\|G(\hat{\latf}) - G(\tilde{\latf})\|_{s} \le \sum_{k=1}^K \|\latf\|_{s} \|\pi_k(\hat{\latf}_0) - \pi_k(\tilde{\latf}_0)\|_{s} + \|\hat{\latf} - \tilde{\latf}\|_{s}$, it suffices to show that $\pi_k$ is continuous. By Lemma \ref{lemma:finitecurve}, $\mathcal{L}(S^0(\latf)) = 0$ a.s.~and since $\mixProb_k(\cdot)$ is constant on $S^0(\latf)^C$ it is also a.s.~continuous. By Proposition 2.6 in \citet{lit:iglesias}, $\mixProb_k(\cdot)$ is therefore continuous on $L_2(\domSp)$ and thus also on $\mathcal{H}^1$ since it is a.s.~constant.

The second term in \eqref{eq:proofineq10} can be bounded by $C \|\obsf\|_{\yspace} \|\hat{\latf} - \tilde{\latf}\|_{s}$ a.s.~since $|G(\hat{\latf})(\psp) - G(\tilde{\latf})(\psp)| \le \| \hat{\latf} - \tilde{\latf}\|_{L^{\infty}(\domSp)}$. Finally, by Lemma \ref{lemma:assHolds} the function $\likPot$ is bounded from above and below, and thus the measure can be normalized.
\end{proof}

From here on we will simplify the notation by omitting the observed point pattern from the likelihood potential and the constants, i.e. $\likPot(\latf) = \likPot(\latf;\obsf)$ and $C_{\mu} = C_{\mu}(\obsf)$. 
\begin{comment}
, as well as the following lemma that is a consequence of Fernique's theorem.
\begin{lemma} \label{lemma:boundexp}
For any $C \in \R$, we have that
$$\int_{\xspace} e^{C \|\latf\|_{H^1}}  d\mu_0(\latf) < \infty.$$
\end{lemma}
\begin{proof}
$$0 \le \left(\frac{C}{2\sqrt{\epsilon}}- \sqrt{\epsilon} \|\latf\|_{H^1} \right)^2 = \frac{C^2}{4\epsilon} + \epsilon \|\latf\|^2_{H^1} - C \|\latf\|_{H^1}.$$
Hence $C \|\latf\|_{H^1} \le \frac{C^2}{4\epsilon} + \epsilon \|\latf\|^2_{H^1}$ and the results follow from Fernique's theorem.
\end{proof}
\end{comment}

\begin{proof}[Proof of Theorem \ref{theorem:postconvergence}]
By \citet[Lemma 6.36]{lit:stuart}, the Hellinger distance bounds the total variation norm, so it suffices to show convergence in Hellinger distance. Take $X\in\mathcal{H}^s, s > 1$. By the triangle inequality,
\begin{align}
2d_{\text{Hell}}&(\mu, \mu^N)^2 = \int \left( \sqrt{\frac{d\mu}{d\nu}} - \sqrt{\frac{d\mu^N}{d\nu}} \right)^2 d\mu_0(\latf) 
= \int \left( \frac{e^{-\frac{1}{2}\likPot(\latf)}}{\sqrt{C_{\mu}}} - \frac{e^{-\frac{1}{2}\likPot^N(\latf)}}{\sqrt{C_{\mu^N}}} \right)^2 d\mu_0(\latf) \\
&\le 
\frac{1}{C_{\mu}} \int  \left| e^{-\frac{1}{2}\likPot(\latf)} - e^{-\frac{1}{2}\likPot^N(\latf)}  \right|^2  d\mu_0(\latf) + \left| \frac{1}{\sqrt{C_{\mu}}} - \frac{1}{\sqrt{C_{\mu^N}}} \right|^2 \int  e^{-\likPot^N(\latf)}  d\mu_0(\latf) \\
&= I_1 + I_2 ,
\label{eq:proofHellBound}
\end{align}
 where $C_{\mu} = \int e^{-\likPot(\latf)}  d\mu_0(\latf)$ and $C_{\mu^N} = \int e^{-\likPot^N(\latf)}  d\mu_0(\latf)$. We now first show that $I_2$ can be bounded by $I_1$ and then show that $I_1\rightarrow 0$ as $N\rightarrow \infty$.
Note that
\begin{align}
I_2 &\le  (C_{\mu} - C_{\mu^N})^2\frac{1}{4} \left( \min\{ C_{\mu}, C_{\mu^N} \} \right)^{-3} C_{\mu^N} \\
&= \left| \int e^{-\likPot(\latf)}   -  e^{-\likPot^N(\latf)}  d\mu_0(\latf)\right|^2 \frac{1}{4} \left( \min\{ C_{\mu}, C_{\mu^N} \} \right)^{-3} C_{\mu^N}\\
&\le   \frac{C_{\mu^N}}{ 4\min\{ C_{\mu}, C_{\mu^N} \}^{3}} \left( \int \left| e^{-\likPot(\latf)} - e^{-\likPot^N(\latf)} \right| d\mu_0(\latf) \right)^2  \\
& \le \frac{C_{\mu^N}}{ 4\min\{ C_{\mu}, C_{\mu^N} \}^{3}}\int \left| e^{-\frac{1}{2}\likPot(\latf)} - e^{-\frac{1}{2}\likPot^N(\latf)} \right|^2 d\mu_0(\latf) \int  e^{\epsilon\|\latf\|^2_{1} - M(\epsilon, \|\obsf\|_{\yspace})} d\mu_0(\latf) \\
& \le C I_1.
\end{align}
Here the third inequality is due to H\"{o}lder's inequality and Ferniques theorem \citep[Theorem A.3]{lit:cotter2}.
Now to bound $I_1$ note that 
\begin{align}
I_1 &\le \frac{1}{4 C_{\mu}} \int  e^{\epsilon\|\latf\|^2_{s} - M(\epsilon, \|\obsf\|_{\yspace})}| \likPot(\latf) - \likPot^N(\latf) |^2  d\mu_0(\latf). 
\end{align}
Since the function $G$ is Lipschitz continuous on $S(X)$ (see Lemma \ref{lemma:finitecurve}) we get
\begin{align}
\left| \likPot(\latf) - \Phi^N(\latf) \right| 
&\le Ce^{C\|\latf\|_{s}} |\domSp| h  
+ C\|\obsf\|_{\yspace} h  \\
&+ C \|\latf\|_ {s}( e^{C\|\latf\|_{s}} \mathcal{L}( S^C(\latf)) + \|S^C(\latf) \cap \obsf\|_{\yspace} ),
\end{align}	
and thus 
\begin{align}
I_1 &\le \frac{2}{4 C_{\mu}}C \int  e^{3\epsilon \|\latf\|^2_{s} - 3M(\epsilon, 1 + \|\obsf\|_{\yspace})} ( |\domSp| + \|\obsf\|_{\yspace})^2 h^{2}   d\mu_0(\latf)  \\
&+ \frac{2}{4 C_{\mu}} C\int   e^{3\epsilon \|\latf\|^2_{s} - 3M(\epsilon, 1 + \|\obsf\|_{\yspace}) } ( \mathcal{L}(S^C(\latf) ) + \| S^C(\latf) \cap \obsf\|_{\yspace} )^2  d\mu_0(\latf).  
\end{align}
Now the first integral on the right hand side clearly goes to zero as $N\rightarrow \infty$. The second integral can be bounded by
\begin{align}
&   \sqrt{\expect{ e^{6\epsilon \|\latf\|^2_{s} - 6M(\epsilon, \|\obsf\|_{\yspace})} }}\sqrt{ \expect{ ( \mathcal{L}(S^C(\latf) ) + \| S^C(\latf) \cap \obsf\|_{\yspace} )^4}} \\
&\le C_2\sqrt{ \expect{ ( \mathcal{L}(S^C(\latf) ) + \| S^C(\latf) \cap \obsf\|_{\yspace} )^4}} \\
&\le C_2( \mathcal{L}(\domSp ) + \|  \obsf\|_{\yspace} )^3 \left(\expect{ \mathcal{L}(S^C(\latf))  } + \expect{ \|S^C(\latf) \cap \obsf\|_{\yspace}  } \right),
\end{align}
and as  $N\rightarrow \infty$ this also goes to zero by Lemma \ref{lemma:finitecurve}.
\end{proof}

\begin{proof}[Proof of Theorem \ref{theorem:postTruncConvergence}]
Denote the posterior measure for the fully discretized model by $\tilde{\mu}^N$. The TV distance between the posterior measures can be bounded as
	\begin{align}
	d_{TV}(\mu, \tilde{\mu}^N) \le d_{TV}(\mu, \mu^N) + d_{TV}(\mu^N, \tilde{\mu}^N),
	\end{align}
	where the first term goes to zero by theorem \ref{theorem:postconvergence}.
Clearly, $\tilde{\mu}^N$ as given in \eqref{eq:postmuNp} defines a posterior measure with respect to $\mu_0$ by the same arguments as in the proof of Proposition \ref{prop:posteriors}, and it coincides with $\mu^N$ on the span of $\{e_j\}_{j>p+1}$.  We can therefore bound $d_{TV}(\mu^N, \tilde{\mu}^N)$ using the same method as in the proof of theorem \ref{theorem:postconvergence}, this gives that $2d_{\text{Hell}}(\mu^N, \tilde{\mu}^N)^2 \le I_1 + I_2$, where now,
	\begin{align}
	I_1 &= \frac{1}{C_{\mu^N}} \int_{\xspace}  \left| e^{-\frac{1}{2}\likPot^N(\latf)} - e^{-\frac{1}{2}\likPot^N(P^p\latf)}  \right|^2  d\mu_0(\latf)  \\
	I_2 &= \left| \frac{1}{\sqrt{C_{\mu^N}}} - \frac{1}{\sqrt{C_{\tilde{\mu}^N}}} \right|^2 \int  e^{-\likPot^N(P^p\latf)}  d\mu_0(\latf) .
	\end{align}
We can again bound $I_2$ by $CI_1$, so what remains to be shown is that $I_1$ goes to zero as $p\rightarrow\infty$. Let $X\in\mathcal{H}^{s}, s > 1$. Since $P^p$ is a projection, we then clearly have that $\|P^pX\|_s \leq \|X\|_s$. By Lemma \ref{lemma:assHolds}(i) and H\"{o}lders inequality
	\begin{align}
	I_1 &\le \frac{1}{4 C_{\mu}} \int  e^{\epsilon\|\latf\|^2_{s} - M(\epsilon, \|\obsf\|_{\yspace})}| \likPot^N(\latf) - \likPot^N(P^p\latf) |  d\mu_0(\latf) \\
	& \leq  C \sqrt{\int e^{\epsilon\|\latf\|^2_{s} - M(\epsilon, \|\obsf\|_{\yspace})} d\mu_0(\latf)  \expect{ | \likPot^N(\latf) - \likPot^N(P^p\latf) |^2  } }.\\ 
	\end{align}
We will now focus on bounding the expectation above. Using Ferniques theorem 
	\begin{align}
	\left|\likPot^N(\latf) - \likPot^N(P^p\latf)\right| &= \sum_{i =1}^N|\domSp_i|(  e^{G(\latf)(\psp_i)} - e^{G(P^p\latf)(\psp_i)} ) - \obsf_i ( G(P^p\latf)(\psp_i) - G(\latf)(\psp_i) ) \\
	&\le e^{  \epsilon \|\latf\|^2_{s} - M(\epsilon, \|\obsf\|_{\yspace}) } \sum_{i =1}^N (|\domSp_i| + Y_i)| G(P^p\latf)(\psp_i) - G(\latf)(\psp_i)|. \\
	\end{align}	
	Using the inequalities
	\begin{align}
	| G(P^p\latf)(\psp) - G(\latf)(\psp)| &\le  \sum_{k=1}^K \left| \pi_k(\latf_0)(\psp) \latf_k(\psp) - \pi_k(P^p\latf_0)(\psp) P^p\latf_k(\psp) \right| \\
	\le \sum_{k=1}^K &\left( \left| \latf_k(\psp) - P^p\latf_k(\psp) \right| + C\|\latf\|_{s} \left| \pi_k(\latf_0)(\psp) - \pi_k(P^p\latf_0)(\psp)  \right| \right)
\end{align}	
	yields that
	\begin{equation}
	\begin{split}
	&\expect{ | \likPot^N(\latf) - \likPot^N(P^p\latf) |^2  }  \leq C  \expect{ \left( \sum_{i \in I^N}(|\domSp_i| + Y_i)  \sum_{k=1}^K \left| \latf_k(\psp_i) - P^p\latf_k(\psp_i) \right|\right)^2} \\
	&\quad + C \expect{\left( \sum_{i \in I^N} (|\domSp_i| + Y_i) \sum_{k=1}^K \left| \pi_k(\latf_0)(\psp_i) - \pi_k(P^p\latf_0)(\psp_i)  \right|\right)^2 }.
	\end{split}	
	\label{eq:proofa51}
	\end{equation}
	
Note that $|D_i|\propto N^{-1}$ and that $X(\psp)$ is bounded for each $\psp\in\domSp$ almost surely. Let $Q^pX = X - P^pX$ and note that $Q^pX$ for each $\psp\in\domSp$ is a mean-zero Gaussian variable with a variance $\sigma_p^2$ that goes to zero as $p\rightarrow\infty$. Thus, the first term in \eqref{eq:proofa51} clearly goes to zero as $p\rightarrow\infty$.
Since $\left| \pi_k(\latf_0)(\psp_i) - \pi_k(P^p\latf_0)(\psp_i)  \right|$ is bounded by one, the second term in \eqref{eq:proofa51} can be bounded by
	\begin{align}
	 C \sum_{i =1}^N (|\domSp_i| + Y_i)  \mathbb{E}\left[\sum_{k=1}^K \left| \pi_k(\latf_0)(\psp_i) - \pi_k(P^p\latf_0)(\psp_i)  \right| \right].
	\end{align}
Here the expectation can be bounded as
	\begin{align}
 \mathbb{E}\left[ \sum_{k=1}^K \left| \pi_k(\latf_0)(\psp_i) - \pi_k(P^p\latf_0)(\psp_i)  \right|\right] \le K  \max_k \left\{ \right.& \left. \prob{\latf_0(\psp_i) \le \threshParam_k \cap P^p\latf_0(\psp_i) > \threshParam_k} \right.  \\
	&\left.+  \prob{\latf_0(\psp) > \threshParam_k \cap P^p\latf_0(\psp) < \threshParam_k} \right\}.
	\end{align}
	We now show how to bound the first probability, and the second probability is bounded by similar calculations. Define the events $A= \{\latf_0(\psp_i) \le \threshParam_k \cap P^p\latf_0(\psp_i) > \threshParam_k\} $ and $B = \{ P^pX_0(s_i) \in [c_k, c_k + \epsilon] \}$. It follows that 
	\begin{align*}
		\mathbb{P}(A) &= \mathbb{P}(A|B) \mathbb{P}(B) + \mathbb{P}(A|B^C) \mathbb{P}(B^C)  \leq \mathbb{P}(B) + \mathbb{P}(A|B^C) \\
		&\leq  \mathbb{P}( P^pX_0(s_i) \in [c_k, c_k + \epsilon] ) + 
		\mathbb{P}( Q^pX_0(s_i)\leq-\epsilon).
	\end{align*}
	Now set $\epsilon = \sqrt{\sigma_p}$ and recall that $\prob{Z>t} < \frac1{\sqrt{2\pi}t}e^{-t^2/2}$ if $Z\sim \pN(0,1)$. This gives that 
	 \begin{align}
	 \prob{A} &\leq  \prob{0<P^p\latf_0(\psp_i) \le \sqrt{\sigma_p} } +  \prob{Q^p\latf_0(\psp_i) \leq -\sqrt{\sigma_p} } 
\leq  C\sqrt{\sigma_p} + \frac{\sqrt{\sigma_p}}{\sqrt{2\pi}} e^{-\frac1{2\sigma_p}},
		\end{align}
which goes to zero as $p\rightarrow \infty$, and thus so does the final expectation in \eqref{eq:proofa51}.

\end{proof}

%\input{appMatern.tex}
%\input{appMCMC.tex}

%\todo{The following appendices should be removed}
%\input{appMCMCModel.tex}
%\input{appGRF.tex}

\end{document}